\documentclass[12pt]{article}

\usepackage{xr-hyper}
\usepackage{amsmath, amsfonts, amssymb, amsthm, fullpage, color, bbm, enumerate, titlesec, graphicx, epstopdf, multirow, array, pdflscape, rotating}

\usepackage[onehalfspacing]{setspace}

\definecolor{darkred}{rgb}{0.5,0,0}

\titleformat*{\paragraph}{\sc}

\usepackage{hyperref}
\hypersetup{allbordercolors={1 1 1},allcolors=darkred,colorlinks=true}
\usepackage[capitalize,nameinlink,noabbrev]{cleveref}

\usepackage[small,bf]{caption}

\usepackage[round]{natbib}
\usepackage[runin]{abstract}

\newcolumntype{C}[1]{>{\centering\arraybackslash}p{#1}}

\theoremstyle{plain}
\newtheorem{asn}{Assumption}
\crefname{asn}{Assumption}{Assumptions}
\newtheorem{lem}{Lemma}
\crefname{lem}{Lemma}{Lemmas}

\newtheorem{prop}{Proposition}
\crefname{prop}{Proposition}{Propositions}

\newcommand{\sfrac}[2]{\textstyle{\frac{#1}{#2}}\displaystyle{}}

\DeclareMathOperator*{\var}{Var}
\DeclareMathOperator*{\cov}{Cov}
\DeclareMathOperator*{\corr}{Corr}

\allowdisplaybreaks
\begin{document}

\title{\texorpdfstring{\vspace{-1.5\baselineskip}}{} Instrumental Variable Identification of \texorpdfstring{\\}{}Dynamic Variance Decompositions\thanks{Email: {\tt mikkelpm@princeton.edu} and {\tt ckwolf@mit.edu}. We received helpful comments from the editor Harald Uhlig, three anonymous referees, Isaiah Andrews, Tim Armstrong, Dario Caldara, Thorsten Drautzburg, Domenico Giannone, Yuriy Gorodnichenko, Ed Herbst, Marek Jaroci\'{n}ski, Peter Karadi, Lutz Kilian, Michal Koles\'{a}r, Byoungchan Lee, Sophocles Mavroeidis, Pepe Montiel Olea, Ulrich M\"{u}ller, Emi Nakamura, Giorgio Primiceri, Eric Renault, Giovanni Ricco, Luca Sala, J\'{o}n Steinsson, Jim Stock, Mark Watson, and seminar participants at several venues. The first draft of this paper was written while Wolf was visiting the Bundesbank, whose hospitality is gratefully acknowledged. Wolf also acknowledges support from the Alfred P. Sloan Foundation and the Macro Financial Modeling Project. Plagborg-M{\o}ller acknowledges that this material is based upon work supported by the NSF under Grant {\#}1851665. Any opinions, findings, and conclusions or recommendations expressed in this material are those of the authors and do not necessarily reflect the views of the NSF.}}
\author{\begin{tabular}{c p{0.5cm} c}
Mikkel Plagborg-M{\o}ller && Christian K. Wolf \\
Princeton University && MIT \& NBER
\end{tabular}}
\date{This version: \today\texorpdfstring{\\}{} First version: August 10, 2017}
\maketitle

\vspace{-1.5em}
\begin{abstract}
Macroeconomists increasingly use external sources of exogenous variation for causal inference. However, unless such external instruments (proxies) capture the underlying shock without measurement error, existing methods are silent on the importance of that shock for macroeconomic fluctuations. We show that, in a general moving average model with external instruments, variance decompositions for the instrumented shock are interval-identified, with informative bounds. Various additional restrictions guarantee point identification of both variance and historical decompositions. Unlike SVAR analysis, our methods do not require invertibility. Applied to U.S. data, they give a tight upper bound on the importance of monetary shocks for inflation dynamics.
\end{abstract}

\clearpage


\section{Introduction}

In recent years, and in parallel to popular microeconometric identification strategies, empirical practice in applied macroeconometrics has turned towards ``external'' sources of plausibly exogenous variation. Such external instrumental variables (IVs, or proxy variables) are now routinely used to estimate causal effects through a simple Two-Stage Least Squares version of Local Projections \citep{Jorda2005,Ramey2016}. Appealingly, this approach is valid even without the assumption of \emph{invertibility} -- the ability to recover structural shocks from current and past (but not future) values of the observed macro variables \citep{Nakamura2017WP,Stock2018}.

However, applied researchers are often not just interested in dynamic causal effects, but also want to learn about a particular shock's contribution to macroeconomic fluctuations \citep{Christiano1999,Beaudry2006,Smets2007}. If the IV is a perfect measure of the underlying structural macro shock, then the desired variance decompositions are readily computed from standard Local Projection regression output \citep{Gorodnichenko2017}. In many applications, though, it is likely that external shock measures are contaminated by substantial measurement error, causing attenuation bias. For example, \citet{Gertler2015} use high-frequency changes in asset prices around monetary policy announcements as credible instruments for monetary shocks; since these instruments at best capture a subset of all monetary shocks, simple direct regressions on the IV are likely to substantially understate the importance of monetary disturbances. Up to this point, the only possible alternative approach was to combine the IV with conventional Structural Vector Autoregressive (SVAR) methods \citep{Stock2008,Mertens2013}, thus automatically imposing the otherwise unnecessary and empirically dubious invertibility assumption.

In this paper, we show precisely to what extent external instruments are informative about shock importance. Throughout, we consider an unrestricted linear moving average model, disciplined only by IVs. This model nests conventional, invertible SVARs, as well as essentially all linearized macro models. We prove three main results. First, without further restrictions, the variance decomposition of the instrumented shock's contribution to macroeconomic fluctuations is interval-identified, with informative lower and upper bounds. Second, if the researcher is willing to impose the assumption of recoverability -- i.e., that the shock is spanned by current, past \emph{and future} values of the observed macro variables -- then both variance decompositions and historical decompositions (the shock's contribution to realized fluctuations) are point-identified. Third, we derive a simple Granger causality pre-test for invertibility that we show exploits the strongest possible testable implication. We complement this set of theoretical results with an extensive code suite that implements all our inference procedures.

We adopt the exact same structural vector moving average (SVMA) model with external IVs as in \citet{Stock2018}, but focus on variance decompositions, rather than impulse responses. The key identifying assumption of this model is the availability of external instruments that correlate with the shock of interest, but are otherwise dynamically uncorrelated with all other macro shocks. Importantly, the IVs may be contaminated by classical measurement error. \citet{Stock2018} show that, in this SVMA-IV model, \emph{relative} impulse responses (which normalize the impact effect) are point-identified, cf. also \citet{Mertens2015}. While such relative impulse responses do not require identification of the scale of the underlying shock, scale inevitably matters for variance and historical decompositions, and so lies at the heart of the identification challenge we face in this paper.

We bound the importance of the instrumented structural shock from above and from below by viewing the model as a dynamic measurement error model. Our question is: Given the second moments (autocovariances) of the macro variables and the IVs, what can be said about (forecast or unconditional) variance decompositions? The identification challenge is that we do not know the signal-to-noise ratio of the IV \emph{a priori}; however, we prove that it is possible to bound this ratio using the moments of the data. At one extreme, our lower bound corresponds to the previously discussed approach of treating the IV as the shock (zero measurement error). If -- as seems likely in practice -- the IV is actually not perfect, then this lower bound may substantially understate the true importance of the shock. At the other extreme, given that we observe a certain degree of co-movement between the IV and the macro observables at various leads and lags, we know that measurement error also cannot be too pervasive. We translate this intuition into formal bounds and prove that these bounds are \emph{sharp}, i.e., they exhaust all the information about variance decompositions contained in the second moments of the data.

We also characterize the set of additional assumptions that researchers could impose to \emph{point}-identify both variance and historical decompositions. Here our main result is that point identification obtains if the instrumented shock is assumed to be \emph{recoverable}, i.e., spanned by all lags \emph{and leads} of the endogenous macro variables. Appealingly, recoverability obtains in any macro model with as many observables as shocks; in particular, it holds even in many models with news and noise shocks, unlike the strictly stronger (and, as we show, testable) invertibility assumption made in SVAR analysis \citep{Leeper2013}.

We provide the applied researcher with an easy-to-use code suite that constructs confidence intervals for all parameters of interest. In a first step, we use a \emph{reduced-form} VAR in macro variables and IVs as a convenient tool for approximating the second moments of the data. The second step then constructs sample analogues of our identification bounds and inserts these into the confidence procedure of \citet{Imbens2004}; alternatively, we also provide confidence intervals valid under the additional point-identifying restriction of recoverability. We prove that our confidence intervals have asymptotically valid frequentist coverage under weak \emph{nonparametric} conditions on the data generating process.

To demonstrate the feasibility and applicability of our procedures, we bound the importance of monetary shocks for inflation dynamics in the U.S. We employ the high-frequency IV proposed by \citet{Gertler2015}, mentioned above. As discussed in \citet{Ramey2016}, the rising importance of forward guidance since the early 1990s is likely to invalidate the invertibility assumption and so threatens consistency of the standard SVAR-IV estimator used by \citeauthor{Gertler2015}. Indeed, we find that the data are consistent with substantial non-invertibility. Applying our robust methodology, we find that monetary shocks are almost irrelevant for aggregate inflation in our post-1990 sample: The 90\% confidence intervals for the forecast variance contribution of monetary shocks rules out values above 8\% at all horizons. Thus, to the extent that inflation is a monetary phenomenon, it is so because of the systematic part of U.S. monetary policy, not because of its erratic conduct.

Finally, we use a series of analytical and quantitative examples to give intuition for why, in spite of its weak identifying assumptions, our method will often manage to give very tight upper bounds on shock importance, consistent with our findings in the monetary application.

\paragraph{Literature.}
\citet{Plagborg2020} prove that the invertibility-robust Local Projection IV impulse response estimator has the same estimand as a recursive SVAR that includes the IV and orders it first. This paper complements our other work by analyzing the identification of variance and historical decompositions, which requires completely different mathematical arguments.

Non-invertibility and its effects on SVAR identification have received substantial attention in recent years (see the references in \citealp[sec. 2.3]{PlagborgMoller2019}). Previous work has emphasized that, in the empirically relevant case of foresight about economic fundamentals or policy (``news''), conventional SVAR analysis invariably fails: Rational expectations equilibria create non-invertible SVMA representations, and so SVARs cannot correctly recover the structural shocks \citep{Leeper2013,Wolf2018}. In contrast, non-invertibility poses no challenge to the methods developed in this paper. We also show that, in the SVMA-IV model, the degree of invertibility is set-identified. Our proposed test of invertibility is related to the Granger causality tests developed in SVAR settings by \citet{Giannone2006} and \citet{Forni2014}. Finally, the weaker notion of ``recoverability'' studied here has independently been proposed by \citet{Chahrour2018} outside the context of external IV identification.\footnote{Recoverability is formally equivalent to the assumption that the structural shock is spanned by current and \emph{future} reduced-form VAR forecast errors. Such dynamic rotations of $u_t$ have been exploited in non-IV settings by \citet{Lippi1994}, \citet{Mertens2010}, and \citet{Forni2017EJ,Forni2017AEJMacro}.}

\paragraph{Outline.}
\cref{sec:framework} defines the SVMA-IV model and the parameters of interest, and states the identification problem. \cref{sec:identif} derives our identification results. \cref{sec:how_to} gives a practical overview of our procedures and their implementation. \cref{sec:mp_application} applies the procedures to bound the importance of monetary shocks. \cref{sec:illustration} illustrates the usefulness and interpretation of the upper bound on shock importance through analytical examples. \cref{sec:simulation} compares the finite-sample performance of our procedures to the SVAR-IV approach through simulations. \cref{sec:concl} concludes. Proofs of our main results are relegated to \cref{sec:appendix}. The Matlab code suite and a supplemental appendix are available online.\footnote{\url{https://github.com/mikkelpm/svma_iv}}


\section{Econometric framework}
\label{sec:framework}

We begin by defining the econometric model and the parameters of interest. Then we state the identification problem.

\subsection{Model}
\label{sec:model}

Following \citet{Stock2018}, we assume a SVMA-IV model. This model allows for an unrestricted linear shock transmission mechanism and, unlike standard SVAR analysis, does not require shocks to be invertible. We also assume the availability of valid external IVs (proxy variables) -- variables that correlate with the shock of interest, but not with the other shocks. For notational clarity, we assume throughout that all time series below have zero mean and are strictly non-deterministic.

First, we define the SVMA model, which places no restrictions on the linear transmission of the vector of shocks $\varepsilon_t$ to the vector of observed endogenous variables $y_t$.

\begin{asn} \label{asn:svma}
The $n_y$-dimensional vector $y_t = (y_{1,t},\dots,y_{n_y,t})'$ of observed macro variables is driven by an unobserved $n_\varepsilon$-dimensional vector $\varepsilon_t=(\varepsilon_{1,t},\dots,\varepsilon_{n_\varepsilon,t})'$ of exogenous economic shocks,
\begin{equation} \label{eqn:svma}
y_t = \Theta(L)\varepsilon_t,\quad \Theta(L) \equiv \sum_{\ell=0}^\infty \Theta_{\ell} L^\ell,
\end{equation}
where $L$ is the lag operator. The matrices $\Theta_\ell$ are each $n_y \times n_\varepsilon$ and absolutely summable across $\ell$. $\Theta(x)$ is assumed to have full row rank for all complex scalars $x$ on the unit circle. The shocks are mutually orthogonal white noise processes:
\[\varepsilon_t \sim \mathit{WN}(0,I_{n_\varepsilon}),\]
where $I_n$ denotes the $n$-dimensional identity matrix.
\end{asn}
The $(i,j)$ element $\Theta_{i,j,\ell}$ of the moving average coefficient matrix $\Theta_\ell$ is the \emph{impulse response} of variable $i$ to shock $j$ at horizon $\ell$. The $j$-th column of $\Theta_\ell$ is denoted by $\Theta_{\bullet, j,\ell}$ and the $i$-th row by $\Theta_{i, \bullet,\ell}$. The full-rank assumption guarantees a nonsingular stochastic process. This condition requires $n_\varepsilon \geq n_y$, but -- crucially -- we do not assume that the number of shocks $n_\varepsilon$ is known. The mutual orthogonality of the shocks is the standard assumption in empirical macroeconomics. The model is semiparametric in that we place no \emph{a priori} restrictions on the coefficients of the infinite moving average, except to ensure a valid stochastic process. In particular, the infinite-order SVMA model \eqref{eqn:svma} is consistent with all discrete-time Dynamic Stochastic General Equilibrium (DSGE) models and all stable SVAR models for $y_t$.

Second, we assume the availability of one or more external IVs for the shock of interest, with the shock of interest specified to be the first one, $\varepsilon_{1,t}$. Each of the $n_z$ IVs $z_t = (z_{1,t},\dots,z_{n_z,t})'$ are assumed to correlate with the first shock but not the other shocks, after controlling for lagged variables: For all $i=1,\dots,n_z$,
\begin{equation} \label{eqn:iv_moment_conds}
E(\tilde{z}_{i,t}\varepsilon_{1,t}) \neq 0,\quad E(\tilde{z}_{i,t} \varepsilon_{j,\tau})=0 \; \text{for all } (j,\tau) \neq (i,t),
\end{equation}
where $\tilde{z}_{i,t}$ is the population residual from projecting $z_{i,t}$ on all lags of $\lbrace z_t,y_t \rbrace$. The key exclusion restriction is that the shock of interest $\varepsilon_{1,t}$ is the only contemporaneous shock to correlate with the IVs $z_t$. Thus, $\tilde{z}_t$ is a proxy for $\varepsilon_{1,t}$ (up to scale) that is contaminated by classical measurement error. This is a strong assumption that must be carefully defended in applications. \citet{Ramey2016} and \citet{Stock2018} survey the extensive applied literature that has constructed plausibly valid external IVs for various shocks.

Using linear projection notation, we can equivalently express the IV exclusion restrictions \eqref{eqn:iv_moment_conds} as the assumption that the IVs $z_t$ are proportional to the shock of interest $\varepsilon_{1,t}$ plus classical measurement error $v_t$ (and possibly lagged observed variables).

\begin{asn} \label{asn:iv}
The IVs $z_t=(z_{1,t},\dots,z_{n_z,t})'$ satisfy
\begin{equation} \label{eqn:iv}
z_t = \sum_{\ell=1}^\infty (\Psi_\ell z_{t-\ell} + \Lambda_\ell y_{t-\ell} ) + \alpha \lambda \varepsilon_{1,t} + \Sigma_v^{1/2} v_t,
\end{equation}
where $\Psi_\ell$ is $n_z \times n_z$, $\Lambda_\ell$ is $n_z \times n_y$, $\lambda$ is an $n_z$-dimensional vector normalized to unit Euclidean length and with its first nonzero element being positive, $\alpha \geq 0$ is a scalar, and $\Sigma_v$ is a symmetric positive semidefinite $n_z \times n_z$ matrix. The elements of $\Psi_\ell$ and $\Lambda_\ell$ are absolutely summable across $\ell$, and the polynomial $x \mapsto \det(I_{n_z} - \sum_{\ell=1}^\infty \Psi_\ell x^\ell)$ has all its roots outside the unit circle. The disturbance vector $v_t$ is a white noise process that is dynamically uncorrelated with the structural shocks $\varepsilon_t$:
\[v_t \sim \mathit{WN}(0,I_{n_z}),\quad \cov(\varepsilon_t, v_\tau) = 0_{n_\varepsilon \times n_z} \text{ for all } t,\tau.\]
\end{asn}
The interpretation of external IVs as noisy measures of true shocks is discussed in \citet{Mertens2013} and \citet{Stock2018}. In our notation \eqref{eqn:iv}, the scale parameter $\alpha$ (along with the residual variance-covariance matrix $\Sigma_v$) measures the overall strength of the IVs, while the unit-length vector $\lambda$ determines which IVs are stronger than others. The assumptions on the coefficients $\Psi_\ell$ and $\Lambda_\ell$ ensure stationarity. We emphasize that the linearity of equation \eqref{eqn:iv} is not a structural assumption; it arises from a linear projection (as in the ``first stage'' of cross-sectional IV). In particular, \cref{asn:iv} is consistent with the IV being a binary or censored series, since such a variable can still satisfy the moment conditions \eqref{eqn:iv_moment_conds} that are equivalent with equation \eqref{eqn:iv}.

Since we restrict attention to identification from second moments, we may without loss of generality simplify notation by assuming that all disturbances are Gaussian.
\begin{asn} \label{asn:shocks}
$(\varepsilon_t',v_t')'$ is i.i.d. jointly Gaussian.
\end{asn}
The Gaussianity assumption is strictly for notational convenience. We could instead have maintained the above white noise assumptions (which allow for conditional heteroskedasticity) and phrased all our results using linear projection notation. The sole meaningful restriction is that we only exploit second moments of the data for identification, as is standard in the applied macro literature, and without loss of generality for Gaussian data.\footnote{If we were to take the assumption of i.i.d. shocks seriously,  and the shocks were \emph{not} Gaussian, higher-order moments of the data would be informative about the parameters. However, we agree with most of the literature that the assumption of i.i.d. shocks is too strong due to the likely presence of stochastic volatility.} We drop the Gaussianity assumption when developing inference procedures in \cref{sec:how_to}.

Finally note also that \cref{asn:svma,asn:iv,asn:shocks} together imply that the $(n_y+n_z)$-dimensional data vector $(y_t',z_t')'$ is strictly stationary.

\subsection{Parameters of interest}
\label{sec:params}

We are interested in the propagation of the first structural shock $\varepsilon_{1,t}$ to the macroeconomic aggregates $y_t$. This section lists the parameters of interest to the applied macroeconomist.

\paragraph{Impulse responses.}
As discussed above, the $(i,1)$ element $\Theta_{i,1,\ell}$ of the moving average coefficient matrix $\Theta_\ell$ is the impulse response of variable $i$ to shock $1$ at horizon $\ell$. We distinguish such \emph{absolute} impulse responses from \emph{relative} impulse responses $\Theta_{i,1,\ell}/\Theta_{1,1,0}$, which give the response of $y_{i,t+\ell}$ to a shock to $\varepsilon_{1,t}$ that increases $y_{1,t}$ by one unit on impact.

\paragraph{Invertibility and recoverability.}
The shock $\varepsilon_{1,t}$ is said to be \emph{invertible} if it is spanned by past and current (but not future) values of the endogenous variables $y_t$: $\varepsilon_{1,t} = E(\varepsilon_{1,t} \mid \lbrace y_\tau\rbrace_{-\infty<\tau\leq t})$. This condition may or may not hold in a given moving average model \eqref{eqn:svma}, depending on the impulse response parameters $\Theta_\ell$. Conventional SVAR analysis invariably imposes invertibility, since the SVAR model obtains from the additional assumptions that $n_\varepsilon=n_y$ and that $\Theta(L)$ has a \emph{one-sided} inverse, so the shocks $\varepsilon_t = \Theta(L)^{-1}y_t$ are spanned by current and past data. However, in many structural macro models, at least some of the shocks cannot be recovered from only lagged macro observables, i.e., the moving average representation is noninvertible. For example, this is often the case in models with news (anticipated) shocks or noise (signal extraction) shocks \citep{Blanchard2013,Leeper2013}. Furthermore, if $n_\varepsilon>n_y$, it is impossible for all shocks to be invertible.

A continuous measure of the \emph{degree} of invertibility is the $R^2$ value in a population regression of the shock on past and current observed variables (\citealp[pp. 243--245]{Sims2006}; \citealp{Forni2018}). More generally, we define
\begin{equation} \label{eqn:R2ell}
R_\ell^2 \equiv \var(E(\varepsilon_{1,t} \mid \lbrace y_\tau \rbrace_{-\infty<\tau\leq t+\ell})),
\end{equation}
the population R-squared value in a projection of the shock of interest on data up to time $t+\ell$ (recall that $\var(\varepsilon_{1,t})=1$). If the shock is invertible in the sense of the previous paragraph, then $R_0^2=1$. Hence, if $R_0^2<1$, then no SVAR model can generate the impulse responses $\Theta(L)$, although the model is \emph{nearly} consistent with SVAR structure if $R_0^2 \approx 1$ \citep{Wolf2018}.

A weaker condition than invertibility is that the shock of interest is \emph{recoverable} from all leads and lags of the endogenous variables -- that is, if $E(\varepsilon_{1,t} \mid \lbrace y_\tau \rbrace_{-\infty<\tau<\infty})=\varepsilon_{1,t}$, or equivalently if $R_\infty^2 = 1$. A sufficient condition is that $n_\varepsilon=n_y$, since then $\Theta(L)$ automatically has a \emph{two-sided} inverse \citep[Thm. 3.1.3]{Brockwell1991}, and thus the shocks $\varepsilon_t = \Theta(L)^{-1}y_t$ are spanned by current, past, and \emph{future} data. This is the case in many DSGE models with news (i.e., anticipated) shocks \citep[e.g.,][]{Leeper2013}.

\paragraph{Variance decompositions.}
Variance decompositions are the key parameters of interest in this paper. We focus in the main text on the \emph{forecast variance ratio} (FVR), where the FVR for the shock of interest for variable $i$ at horizon $\ell$ is defined as
\[\mathit{FVR}_{i,\ell} \equiv 1 - \frac{\var(y_{i,t+\ell} \mid \lbrace y_\tau \rbrace_{-\infty<\tau\leq t}, \lbrace \varepsilon_{1,\tau} \rbrace_{t<\tau<\infty})}{\var(y_{i,t+\ell} \mid \lbrace y_\tau \rbrace_{-\infty<\tau\leq t})} = \frac{\sum_{m=0}^{\ell-1} \Theta_{i,1,m}^2}{\var(y_{i,t+\ell} \mid \lbrace y_\tau \rbrace_{-\infty<\tau\leq t})}.\]
The FVR measures the reduction in the econometrician's forecast variance that would arise from being told the entire path of future realizations of the first shock. The larger this measure is, the more important is the first shock for forecasting variable $i$ at horizon $\ell$. The FVR is always between 0 and 1.

\cref{sec:identif_vd} defines and provides identification analysis for two additional variance decomposition concepts. First, the \emph{forecast variance decomposition} (FVD) is like the FVR but instead conditions on the history of all past shocks $\lbrace \varepsilon_\tau \rbrace_{-\infty<\tau\leq t}$, rather than the history of observables $\lbrace y_\tau \rbrace_{-\infty<\tau\leq t}$. Under invertibility, the FVR and FVD are identical (since then the information set $\lbrace y_\tau \rbrace_{-\infty<\tau\leq t}$ equals the information set $\lbrace \varepsilon_\tau \rbrace_{-\infty<\tau\leq t}$), explaining why the previous SVAR literature has not distinguished between the two. Second, we consider the \emph{unconditional} frequency-specific variance decomposition (VD) of \citet[sec. 3.4]{Forni2018}.

\paragraph{Historical decomposition.}
The \emph{historical decomposition} of variable $y_{i,t}$ at time $t$ attributable to the shock of interest is defined as $E(y_{i,t} \mid \lbrace \varepsilon_{1,\tau} \rbrace_{-\infty<\tau\leq t}) = \sum_{\ell=0}^\infty \Theta_{i,1,\ell}\varepsilon_{1,t-\ell}$.

\subsection{Identification problem}
\label{sec:id_problem}

Our goal for the remainder of the paper is to answer the question: Given \cref{asn:svma,asn:iv,asn:shocks}, what do the second moments (autocovariances) of the data $(y_t',z_t')'$ say about the parameters of interest defined above? In particular, can we test whether the shock $\varepsilon_{1,t}$ is invertible?

\citet{Stock2018} showed that \emph{relative} impulse responses are point-identified in the SVMA-IV model. To see this transparently, consider the case with a single IV, so $\lambda=1$. Since
\begin{equation} \label{eqn:irf}
\cov(y_{i,t},z_t \mid \lbrace y_\tau,z_\tau \rbrace_{-\infty<\tau<t}) = \alpha\Theta_{i, 1,\ell},
\end{equation}
the absolute impulse responses $\Theta_{i, 1,\ell}$ for all variables $i$ and all horizons $\ell$ are identified up to the single scale parameter $\alpha$. Thus, the \emph{relative} impulse responses $\Theta_{i, 1,\ell}/\Theta_{1, 1,0}$ are point-identified, as $\alpha$ drops out from the fraction.

The main challenge addressed in this paper is that (partial) identification of variance and historical contributions requires (partial) identification of the \emph{absolute} impulse responses, and thus of the scale parameter $\alpha$.


\section{Identification results}
\label{sec:identif}

This section contains our main theoretical identification results. Readers who are primarily interested in practical implementation are encouraged to skip ahead to \cref{sec:how_to}. For exposition, we start in \cref{sec:identif_stat} by deriving results for a simple static version of our SVMA-IV model. We then turn to the general dynamic model in \cref{sec:identif_dyn_scale,sec:identif_dyn_interest,sec:point_id}, applying the static results to the frequency domain representation of the data. We initially focus on the case with a single IV, but we discuss the straight-forward extension to multiple IVs in \cref{sec:iv_multi}.

\subsection{Static model}
\label{sec:identif_stat}

To build intuition, consider a static version of the SVMA-IV model with a single instrument:
\begin{align*}
y_t &= \Theta_{\bullet,1,0} \varepsilon_{1,t} + \xi_t, \\
z_t &= \alpha \varepsilon_{1,t} + \sigma_v v_t, \\
(\varepsilon_{1,t},v_t,\xi_t)' &\stackrel{i.i.d.}{\sim} N\left(0, \begin{pmatrix}
I_2 & 0_{2 \times n_y} \\
0_{n_y \times 2} & \Sigma_\xi
\end{pmatrix}\right).
\end{align*}
Here $\alpha,\sigma_v \geq 0$ are scalars, $\xi_t \equiv \sum_{j=2}^{n_\varepsilon} \Theta_{\bullet,j,0}\varepsilon_{j,t}$ is an $n_y$-dimensional random vector that captures all the structural shocks other than the one of interest, and $\Sigma_\xi \equiv \var(\xi_t)$.\footnote{While the static model is primarily intended to provide intuition about the analysis of the SVMA-IV model, the results in this subsection are directly relevant for identification in the more restrictive SVAR model with an external IV. In that framework, $y_t$ would denote the $n_y$ reduced-form VAR residuals, which are linear functions of the vector $\varepsilon_t$ of $n_\varepsilon$ contemporaneous structural shocks.}

Our main parameter of interest is the Forecast Variance Ratio
\[\textit{FVR}_{i,1} = 1 - \frac{\var(y_{i,t} \mid \varepsilon_{1,t})}{\var(y_{i,t})} =  \frac{\Theta_{i,1,0}^2}{\var(y_{i,t})}.\]
This is just the population R-squared value in the (infeasible) regression of $y_{i,t}$ on $\varepsilon_{1,t}$. Since $\cov(y_{i,t},z_t) = \alpha\Theta_{i,1,0}$, it is easy to see that the FVR is identified up to a factor $1/\alpha^2$:
\begin{equation} \label{eqn:static_fvr}
\textit{FVR}_{i,1} = \frac{1}{\alpha^2} \times \frac{\cov(y_{i,t},z_t)^2}{\var(y_{i,t})}.
\end{equation}
Thus, we ask: What does the variance-covariance matrix of the data $(y_t',z_t)'$ say about the scale parameter $\alpha^2$?

Our key insight is that the static model is nothing but a multivariate classical measurement error model: Whereas we would like to measure the R-squared value from a regression of $y_t$ on $\varepsilon_{1,t}$, we only observe the noisy proxy $z_t$ for the ``regressor''. Intuitively, the contribution of $\varepsilon_{1,t}$ to $y_t$ is not point-identified because the signal-to-noise ratio $\alpha^2/\sigma_v^2$ of the proxy $z_t$ is not known \emph{a priori}. For example, upon observing a small correlation between the IV and macro observables, we do not know whether this correlation is small because of measurement error or because the shock is unimportant. Nevertheless, the moments of the data are informative about the signal-to-noise ratio. At one extreme, the IV can never be more than perfect -- at best, there is no measurement error (infinite signal-to-noise ratio). At the other extreme, the signal-to-noise ratio cannot be zero, since then the IV would not correlate at all with macro observables. We now formalize this intuition.\footnote{Our bounds do not follow from existing results in the literature on measurement error in linear regression \citep[e.g.,][]{Klepper1984}, since our parameters of interest are not regression coefficients.}

\paragraph{Lower bound on shock importance.}
We begin with a lower bound on the importance of the shock (and so on the amount of measurement error), or equivalently an upper bound on $\alpha^2$. To derive this bound, simply observe that
\[\alpha^2 \leq \alpha^2 + \sigma_v^2 = \var(z_t).\]
This inequality binds when there is no measurement error in the IV, i.e., when $\sigma_v=0$.

Mapping this upper bound on $\alpha^2$ into a lower bound on the FVR via \eqref{eqn:static_fvr}, we get
\begin{equation} \label{eqn:static_fvr_lower}
\mathit{FVR}_{i,1} \geq \frac{\cov(y_{i,t},z_t)^2}{\var(z_t)\var(y_{i,t})} = \corr(y_{i,t},z_t)^2.
\end{equation}
The lower bound corresponds to the population R-squared value in a regression of $y_{i,t}$ on $z_t$, that is, a regression which treats the IV as if it were a perfect measure of the shock $\varepsilon_{1,t}$ (up to scale). The attenuation bias imparted by the measurement error $v_t$ implies that this regression yields a lower bound on the true FVR.

\paragraph{Upper bound on shock importance.}
To derive the upper bound on the importance of the shock (and on the amount of measurement error), or equivalently the lower bound on $\alpha$, define first $z_t^\dagger \equiv E(z_t \mid y_t)$ and $\varepsilon_{1,t}^\dagger \equiv E(\varepsilon_{1,t} \mid y_t)$. Then, by standard linear projection algebra, we must have
\[\var(z_t^\dagger) = \alpha^2\var(\varepsilon_{1,t}^\dagger) \leq \alpha^2 \var(\varepsilon_{1,t}) = \alpha^2.\]
Intuitively, $\alpha^2=\var(E(z_t \mid \varepsilon_{1,t}))$ is the explained sum of squares from a projection of $z_t$ on the shock $\varepsilon_{1,t}$. This must weakly exceed the explained sum of squares $\var(z_t^\dagger)=\var(E(z_t \mid y_t))$ from a projection of $z_t$ on $y_t$, simply because the variables in $y_t$ are effectively noisy measures of the shock $\varepsilon_{1,t}$ contaminated by other structural shocks $\xi_t$, and uncorrelated with $v_t$. In other words, the explanatory power of the variables $y_t$ for the IV $z_t$ puts a lower bound on the possible signal-to-noise ratio $\alpha^2/\sigma_v^2=\alpha^2/(\var(z_t)-\alpha^2)$. The inequality above binds when the shock is invertible ($\varepsilon_{1,t}^\dagger \equiv E(\varepsilon_{1,t} \mid y_t)=\varepsilon_{1,t}$), i.e., when the macro observables $y_t$ explain as much of the variation in the IV as the shock $\varepsilon_{1,t}$ itself does.

Mapping the lower bound on $\alpha^2$ into an upper bound on the FVR via \eqref{eqn:static_fvr}, we get
\begin{equation} \label{eqn:static_fvr_upper}
\mathit{FVR}_{i,1} \leq \frac{\cov(y_{i,t},z_t)^2}{\var(z_t^\dagger)\var(y_{i,t})} = \frac{\cov(y_{i,t},z_t^\dagger)^2}{\var(z_t^\dagger)\var(y_{i,t})} = \corr(y_{i,t},z_t^\dagger)^2.
\end{equation}
The upper bound corresponds to treating the projection $z_t^\dagger = \alpha \varepsilon_{1,t}^\dagger$ of the IV on the macro observables as a perfect measure of the shock (up to scale). This is correct if indeed the shock were invertible ($\varepsilon_{1,t}^\dagger=\varepsilon_{1,t}$), but otherwise overstates the importance of the shock. Intuitively, unless the shock is in fact invertible, the upper bound mistakenly attributes too much of the lack of co-movement between $y_t$ and $z_t$ to measurement error (rather than the actual limited importance of $\varepsilon_{1,t}$).

Whereas the lower bound \eqref{eqn:static_fvr_lower} on the FVR for variable $i$ does not depend on the entire set of observed macro aggregates $y_t$, the upper bound \eqref{eqn:static_fvr_upper} decreases monotonically as we add more variables to the vector $y_t$. In particular, the upper bound equals the trivial bound of 1 if there is only one observable ($n_y=1$), since in this case we cannot rule out that the scalar time series $y_t$ is driven entirely by the first shock, with the imperfect correlation between $y_t$ and $z_t$ \emph{purely} caused by measurement error.\footnote{Mathematically, when $y_t$ is a scalar, then $z_t^\dagger = E(z_t \mid y_t) \propto y_t$, so $\corr(y_t,z_t^\dagger)= \pm 1$.} However, when $n_y \geq 2$, the upper bound is generally below 1. We present an analytical example in \cref{subsec:info_obs} that shows how the addition of extra observables helps sharpen identification, and clarifies the conditions under which we can expect the upper bound to be close to the true FVR.

\paragraph{Identified set.}
The bounds $\alpha^2 \in [\var(z_t^\dagger),\var(z_t)]$ are \emph{sharp}, i.e., exploit all information contained in the second moments of the data, in the following sense. Suppose we are given any non-singular variance-covariance matrix for the data $(y_t',z_t)'$, as well as any value of $\alpha^2$ in our interval. We can then choose appropriate values of the remaining parameters such that the model matches the given variance-covariance matrix of the data.\footnote{ This is achieved by the choices $\Theta_{\bullet,1,0} = \frac{1}{\alpha}\cov(y_t,z_t)$, $\sigma_v^2 = \var(z_t) - \alpha^2$, and $\Sigma_\xi = \var(y_t) - \frac{1}{\alpha^2}\cov(y_t,z_t)\cov(y_t,z_t)'$. This choice of $\sigma_v^2$ is nonnegative since $\var(z_t) \geq \alpha^2$, and \cref{thm:semidef} in \cref{sec:proofs_lemma} implies that the choice of $\Sigma_\xi$ is a positive semidefinite matrix since $\alpha^2 \geq \var(z_t^\dagger) = \cov(z_t,y_t)\var(y_t)^{-1}\cov(y_t,z_t)$.}

Under what conditions are the bounds on $\alpha^2$ -- and thus on the FVR -- likely to be tight (i.e., close to the true FVR)? We can express the identified set for $1/\alpha^2$ in terms of the underlying model parameters as follows:
\[\frac{1}{\alpha^2} \in \left[\frac{\alpha^2}{\alpha^2 + \sigma_v^2} \times \frac{1}{\alpha^2}\;,\; \frac{1}{\var(E(\varepsilon_{1,t} \mid y_t))} \times \frac{1}{\alpha^2}\right].\]
The lower bound is closer to the true value $1/\alpha^2$ when the actual signal-to-noise ratio $\alpha^2/\sigma_v^2$ is larger, i.e., when the IV is stronger. The upper bound is closer to the true FVR when the degree of invertibility $R_0^2 = \var(\varepsilon_{1,t}^\dagger) = \var(E(\varepsilon_{1,t} \mid y_t))$ is larger, i.e., when the macro variables $y_t$ are more informative about the hidden shock $\varepsilon_{1,t}$. Finally, the identified set is never empty, and it collapses to a point only in case of a perfect IV \emph{and} invertibility.

\paragraph{Point identification.}
Point identification obtains if the researcher assumes either that the IV is perfect ($\sigma_v=0$), in which case the lower bound for the FVR binds, or that the shock of interest is invertible ($\varepsilon_{1,t}^\dagger=\varepsilon_{1,t}$), in which case the upper bound binds.

\subsection{Dynamic model: shock scale}
\label{sec:identif_dyn_scale}

We now analyze identification in the general dynamic model of \cref{sec:model}. The key idea in our proofs is to apply the logic of the static model frequency-by-frequency to the frequency domain representation of the data.

As in the static case, we begin in this section by characterizing the identified set for the scale parameter $\alpha$. While not economically interesting in itself, this scale parameter is ultimately key to identification of our actual parameters of interest. We maintain \cref{asn:svma,asn:iv,asn:shocks} throughout, but for the moment consider the case of a single IV ($n_z=1$), leaving the generalization to \cref{sec:iv_multi}. That is, $z_t$ is a scalar and $\lambda=1$ in equation \eqref{eqn:iv}. We write $\Sigma_v^{1/2}=\sigma_v \geq 0$, a scalar.

\paragraph{Preliminaries.}
It will prove convenient to define the IV projection residual that removes any dependence on lagged observed variables:
\begin{equation} \label{eqn:z_tilde}
\tilde{z}_t \equiv z_t - E(z_t \mid \lbrace y_\tau,z_\tau\rbrace_{-\infty<\tau<t}) = \alpha  \varepsilon_{1,t} + \sigma_v v_t.
\end{equation}
Note that $\tilde{z}_t$ is serially uncorrelated by construction.

Next, we need to define our notation for spectral density matrices. For any two jointly stationary vector time series $a_t$ and $b_t$ of dimensions $n_a$ and $n_b$, respectively, define the $n_a \times n_b$ cross-spectral density matrix function \citep[ch. 4 and 11]{Brockwell1991}
\[s_{ab}(\omega) \equiv \frac{1}{2\pi}\sum_{\ell=-\infty}^\infty e^{-i\omega\ell}\cov(a_t,b_{t-\ell}),\quad \omega \in [0,2\pi].\]
For any vector time series $a_t$, we denote its spectrum by $s_a(\omega) \equiv s_{aa}(\omega)$.

\paragraph{Lower bound on shock importance.}
We again begin with a lower bound on shock importance (or the amount of measurement error), which corresponds to an upper bound on the scale parameter $\alpha$. As in the static model, we find
\begin{equation} \label{eqn:alpha_ub}
\alpha^2 \leq \alpha^2 + \sigma_v^2 = \var(\tilde{z}_t) \equiv \alpha_{UB}^2.
\end{equation}
Thus, once we look at the residualized IV in \eqref{eqn:z_tilde}, the bound construction works as in the static case, with the boundary $\alpha = \alpha_{UB}$ corresponding to a perfect IV.

\paragraph{Upper bound on shock importance.}
For the upper bound on shock importance (or the lower bound on $\alpha^2$), we apply a version of the argument from the static case to the joint spectrum of the data at every frequency. First, as in the static case, we define the projections of $\tilde{z}_t$ and $\varepsilon_{1,t}$, respectively, just now onto all leads and lags of the endogenous variables $y_t$:
\begin{align}
\tilde{z}_t^\dagger &\equiv E(\tilde{z}_t \mid \lbrace y_\tau \rbrace_{-\infty<\tau<\infty}), \label{eqn:z_tilde_dagger} \\
\varepsilon_{1,t}^\dagger &\equiv E(\varepsilon_{1,t} \mid \lbrace y_\tau \rbrace_{-\infty<\tau<\infty}). \nonumber
\end{align}
Note that $\tilde{z}_t^\dagger = \alpha \varepsilon_{1,t}^\dagger$, since the measurement error $v_t$ is dynamically uncorrelated with $y_t$. Applying the same logic as in the static case at an arbitrary frequency $\omega \in [0,2\pi]$, we have
\begin{equation} \label{eqn:alpha_lb_ineq}
s_{\tilde{z}^\dagger}(\omega) = \alpha^2 s_{\varepsilon_1^\dagger}(\omega) \leq \alpha^2 s_{\varepsilon_1}(\omega) = \alpha^2 \times \sfrac{1}{2\pi}.
\end{equation}
The last equality uses that the shock $\varepsilon_{1,t}$ is white noise with variance 1. Similar to the static case, the inequality above arises because the ``explained sum of squares'' $s_{\varepsilon_1^\dagger}(\omega)$ from a frequency-specific projection of the shock $\varepsilon_{1,t}$ on all leads and lags of the macro observables $y_t$ must be less than the ``total sum of squares'' $s_{\varepsilon_1}(\omega)$.\footnote{\citet[Remark 3, p. 439]{Brockwell1991} show that $s_{\tilde{z}^\dagger}(\omega) = s_{y\tilde{z}}(\omega)^*s_y(\omega)^{-1}s_{y\tilde{z}}(\omega)$ and $s_{\varepsilon_1^\dagger}(\omega) = s_{y\varepsilon_1}(\omega)^*s_y(\omega)^{-1}s_{y\varepsilon_1}(\omega)$. Since the joint spectrum is positive semidefinite, $s_{\varepsilon_1}(\omega) \geq s_{\varepsilon_1^\dagger}(\omega)$ for all $\omega$.\label{fn:bw}} Exploiting the inequality \eqref{eqn:alpha_lb_ineq} at all frequencies, we obtain the lower bound
\begin{equation} \label{eqn:alpha_lb}
\alpha^2 \geq \textstyle 2\pi \sup_{\omega \in [0,\pi]} s_{\tilde{z}^\dagger}(\omega) \equiv \alpha_{LB}^2.
\end{equation}
The bound binds if at \emph{some} frequency $\omega \in [0,\pi]$ the observed macro aggregates are perfectly informative about the hidden shock $\varepsilon_{1,t}$. This is the natural dynamic, frequency-domain analogue of the condition in the static case, where we required the static $y_t$ to be perfectly informative about $\varepsilon_{1,t}$. If the macro aggregates are in fact not perfectly informative about the shock at any frequency, then the lower bound attributes too much of the (frequency-by-frequency) lack of co-movement between $y_t$ and $z_t$ to measurement error.

\paragraph{The identified set.}
The main theoretical result of this paper is that the above bounds $\alpha_{LB}^2,\alpha_{UB}^2$ are sharp.
\begin{prop} \label{thm:identif_alpha}
Let there be given a joint spectral density for $w_t = (y_t',\tilde{z}_t)'$, continuous and positive definite at every frequency, with $\tilde{z}_t$ unpredictable from $\lbrace w_\tau \rbrace_{-\infty<\tau<t}$. Choose any $\alpha \in (\alpha_{LB},\alpha_{UB}]$. Then there exists an SVMA-IV model as in \cref{asn:svma,asn:iv} with the given $\alpha$ such that the model-implied spectral density of $w_t$ matches the given spectral density.
\end{prop}
Recall that the previous discussion has already shown that any value of $\alpha^2 \notin [\alpha_{LB}^2,\alpha_{UB}^2]$ is impossible. The proposition strengthens this result to say that, given the second moments of the data, we cannot rule out any values of $\alpha^2$ in the interval $[\alpha_{LB}^2,\alpha_{UB}^2]$.\footnote{The proposition does not cover the knife-edge case $\alpha=\alpha_{LB}$ due to economically inessential technicalities.}

To interpret the identified set, we proceed as in the static model and express the interval in terms of the underlying model parameters. We focus on the identified set for $\frac{1}{\alpha^2}$, as this transformation is again the most relevant one for identifying the FVR and degree of invertibility/recoverability, as shown below. We can write the identified set for $1/\alpha^2$ as
\begin{equation} \label{eqn:identif_alpha_param}
\frac{1}{\alpha^2} \in \bigg [\underbrace{\frac{\alpha^2}{\alpha^2 + \sigma_v^2}}_{\text{instrument strength}} \times \; \; \; \frac{1}{\alpha^2}, \; \; \; \underbrace{\frac{1}{1 - 2\pi \inf_{\omega \in [0,\pi]}s_{\varepsilon_1 - \varepsilon_1^\dagger}(\omega)}}_{\text{informativeness of data for shock}} \times \; \; \; \frac{1}{\alpha^2} \bigg].
\end{equation}
As in the static case, the lower bound is larger (and closer to the true $\frac{1}{\alpha^2}$) when the instrument is stronger in the sense of a higher signal-to-noise ratio $\alpha^2/\sigma_v^2$. The upper bound is again smaller (and closer to the true $\frac{1}{\alpha^2}$) when the data are more informative about the shock of interest. The relevant notion of informativeness, however, is now more complicated than in the static case, for two reasons: First, we now exploit the explanatory power of all leads and lags of the macro aggregates when forming the projection $\varepsilon_{1,t}^\dagger \equiv E(\varepsilon_{1,t} \mid \lbrace y_\tau \rbrace_{-\infty<\tau<\infty})$; and second, we consider all frequencies of the data separately. The upper bound is close to the truth as long as the leads and lags of $y_t$ are highly informative about the frequency-$\overline{\omega}$ fluctuations of the shock at \emph{some} frequency $\overline{\omega}$ (e.g., in the long run $\overline{\omega} \approx 0$), in the sense that the spectral density of the projection residual $\varepsilon_{1,t}- \varepsilon_{1,t}^\dagger$ vanishes at this frequency. This does not require the macro variables to be informative about the shock at \emph{all} frequencies (e.g., in the short run $\overline{\omega} \approx \pi$). We illustrate this point in \cref{subsec:info_obs_dyn}.

Similar to the static case, the identified set for $\frac{1}{\alpha^2}$ does not collapse to a point unless the instrument is perfect \emph{and} there exists a frequency $\overline{\omega}$ for which the data are perfectly informative about the frequency-$\overline{\omega}$ cyclical component of the shock.

\paragraph{Practical upper bound on shock importance.}
In practice, we do not recommend exploiting the sharp lower bound on $\alpha$ for estimation and inference. The reason is that $\alpha_{LB}$ in equation \eqref{eqn:alpha_lb} equals the supremum of a function, which depends on the spectral density matrix of the data. Nonparametric estimation of the supremum of an unknown function is highly challenging given the moderate sample sizes available to applied macroeconomists \citep{Gafarov2018}. For this reason, our implementation in \cref{sec:how_to} instead uses the weaker bound
\begin{equation} \label{eqn:alpha_lb_lb}
\underline{\alpha}^2 \equiv \var(\tilde{z}_t^\dagger) = \int_0^{2\pi} s_{\tilde{z}^\dagger}(\omega)\,d\omega \leq 2 \pi \sup_{\omega \in [0,\pi]} s_{\tilde{z}^\dagger}(\omega) = \alpha_{LB}^2.
\end{equation}
Since $\underline{\alpha}^2$ is given by an integral of the spectrum as opposed to a supremum, its point estimator defined in \cref{sec:how_to} is consistent and asymptotically normal, as shown in \cref{sec:var_sieve}.\footnote{Methods from the moment inequality literature could be applied to develop confidence intervals that exploit our sharp lower bound $\alpha_{LB}^2$ \citep{Andrews2013, Andrews2017, Chernozhukov2013}. We leave this more complicated option to future work. Alternatively, if researchers have a strong \emph{a priori} reason to believe that the shock is likely to be particularly important at certain frequencies, then they may fix frequency bounds $[\omega_1, \omega_2]$ and compute the integral in \eqref{eqn:alpha_lb_lb} by integrating over this interval only.\label{fn:fix_interval_sd}}

Since $\var(\tilde{z}_t^\dagger) = \alpha^2 \var(\varepsilon_{1,t}^\dagger) = \alpha^2 \times R_\infty^2$, the weaker lower bound on $\alpha^2$ will nevertheless be close to the truth if the shock of interest is close to being recoverable ($R_\infty^2 \approx 1$), and thus in particular if the shock is close to being invertible ($R_0^2 \approx 1$). In \cref{sec:illustration_news} we show by example that the bound $\underline{\alpha}^2$ binds in a model with news shocks, which cannot be analyzed using conventional SVAR-IV methods that assume invertibility.

\subsection{Dynamic model: parameters of interest}
\label{sec:identif_dyn_interest}

Given the identified set for $\frac{1}{\alpha^2}$, it is now straight-forward to derive identified sets for variance decompositions as well as the degree of invertibility and recoverability.

\paragraph{Variance decompositions.}
The FVR satisfies
\begin{equation} \label{eqn:identif_fvr}
\mathit{FVR}_{i,\ell} = \frac{\sum_{m=0}^{\ell-1} \Theta_{i,1,m}^2}{\var(y_{i,t+\ell} \mid \lbrace y_\tau \rbrace_{-\infty<\tau\leq t})} = \frac{1}{\alpha^2} \times \frac{\sum_{m=0}^{\ell-1} \cov(y_{i,t},\tilde{z}_{t-m})^2}{\var(y_{i,t+\ell} \mid \lbrace y_\tau \rbrace_{-\infty<\tau\leq t})}.
\end{equation}
Hence, as in the static case, the identified set for $\mathit{FVR}_{i,\ell}$ equals the identified set for $\frac{1}{\alpha^2}$, scaled by the (point-identified) second fraction on the far right-hand side above. As discussed previously, and as in the static case, the lower bound for the FVR depends on the strength of the IV, and the upper bound on the FVR depends on the informativeness of the macro variables for the shock of interest. Adding more variables to the vector $y_t$ of endogenous observables always leads to a weakly narrower identified set (in percentage terms, since the parameter $\mathit{FVR}_{i,\ell}$ itself also changes when we change the vector $y_t$). Unlike in the static case, the upper bound in the dynamic case is generally below 1 even if we only observe a single macro time series ($n_y=1$), as shown by example in \cref{subsec:info_obs_dyn}.

\cref{sec:identif_vd} derives bounds on the other variance decomposition concepts (VD and FVD) introduced in \cref{sec:params}. Bounding the FVD in particular requires more work.

\paragraph{Degree of invertibility \& recoverability.}
The definition \eqref{eqn:R2ell} of $R_\ell^2$ implies
\begin{equation} \label{eqn:identif_R2}
R_\ell^2 = \frac{1}{\alpha^{2}} \times \var(E(\tilde{z}_t \mid \lbrace y_\tau \rbrace_{-\infty<\tau\leq t+\ell})).
\end{equation}
Since the variance on the right-hand side above is point-identified, the identified sets for the degree of invertibility ($\ell = 0$) and the degree of recoverability ($\ell=\infty$) follow immediately from the identified set for $\frac{1}{\alpha^2}$.

From the sharp bounds on $R_0^2$ and $R_\infty^2$, we can also derive testable conditions under which the distribution of the observable data is consistent with invertibility or recoverability.
\begin{prop} \label{thm:identif_R2}
Assume $\alpha_{LB}^2>0$. The identified set for $R_0^2$ contains 1 if and only if the instrument residual $\tilde{z}_t$ does not Granger cause the macro observables $y_t$. The identified set for $R_\infty^2$ contains 1 if and only if the projection $\tilde{z}_t^\dagger$ is serially uncorrelated.
\end{prop}
According to \cref{thm:identif_R2}, $\varepsilon_{1,t}$ is certain to be noninvertible if and only if $\tilde{z}_t$ Granger causes $y_t$ (which is equivalent with the condition that $z_t$ Granger causes $y_t$). This result will be the basis for the pre-test of invertibility in \cref{sec:how_to}. Note, however, that a finding of Granger non-causality need not imply that $R_0^2=1$; the identified set for $R_0^2$ always includes values below 1. \cref{thm:identif_R2} additionally implies that $\varepsilon_{1,t}$ is certain to be non-recoverable if and only if $\tilde{z}_t^\dagger$, defined in \eqref{eqn:z_tilde_dagger}, is serially correlated at some lag.\footnote{We leave the development of a practical statistical test of recoverability to future research.}

\paragraph{Absolute impulse responses.}
For completeness, we note that the identified set for the \emph{absolute} impulse response $\Theta_{i,1,\ell}$ is obtained by scaling the identified set for $\frac{1}{\alpha}$, cf. equation \eqref{eqn:irf}. This extends existing results on the point-identification of \emph{relative} impulse responses \citep{Stock2018}, as discussed at the end of \cref{sec:framework}.

\subsection{Dynamic model: point identification}
\label{sec:point_id}

As we have seen, without further restrictions, our various parameters of interest are only interval-identified, albeit with informative bounds. In this section we complement those results by stating a menu of sufficient conditions, each of which guarantees point identification of the FVR and historical decompositions.

\paragraph{Informative instruments.}
Point identification obtains if the researcher is willing to assume that the instrument is perfect, i.e., $\sigma_v=0$. In this case the lower bounds on the FVR and degree of invertibility/recoverability bind. Indeed, since the instrument equals the shock up to scale, $\tilde{z}_t=\alpha\varepsilon_{1,t}$, the FVR and historical decompositions are easily computed through regressions \citep{Jorda2005,Gorodnichenko2017}. Note that the assumption that the IV is perfect is not testable.

\paragraph{Informative macro aggregates.}
The second set of sufficient conditions relates to the informativeness of the macro aggregates $y_t$ for the hidden shock $\varepsilon_{1,t}$. In this category, our weakest condition for point identification is that the data $y_t$ is perfectly informative about $\varepsilon_{1,t}$ at \emph{some} frequency, i.e., the spectral density of the projection residual $\varepsilon_{1,t}-\varepsilon_{1,t}^\dagger$ vanishes at some frequency $\overline{\omega}$. Then $\alpha=\alpha_{LB}$, so the FVR and degree of invertibility/recoverability are identified. This assumption is not testable.

A stronger but more easily interpretable assumption is recoverability, i.e., $\varepsilon_{1,t}^\dagger \equiv E(\varepsilon_{1,t} \mid \lbrace y_\tau\rbrace_{-\infty<\tau<\infty})=\varepsilon_{1,t}$. This assumption is testable, cf. \cref{thm:identif_R2}. As explained in \cref{sec:params}, recoverability is restrictive, but it is a meaningfully weaker requirement than invertibility in many economic applications, such as in the news shock model in \cref{sec:illustration_news} below. In particular, it is satisfied whenever there are as many shocks as variables, $n_\varepsilon=n_y$. Under recoverability, the shock itself can be identified as $\varepsilon_{1,t} = \frac{1}{\alpha} \tilde{z}_t^\dagger$, so the historical decomposition $E(y_{i,t} \mid \lbrace \varepsilon_{1,\tau} \rbrace_{-\infty<\tau \leq t})=E(y_{i,t} \mid \lbrace \tilde{z}_\tau^\dagger \rbrace_{-\infty<\tau \leq t})$ is also identified.

\subsection{Extension: multiple instruments}
\label{sec:iv_multi}

To conclude, we briefly extend the analysis to a model with multiple IVs for the shock of interest ($n_z\geq 2$). This extended multiple-IV model is testable, unlike the single-IV model. As in the single-IV case, define the projection residual
\begin{equation} \label{eqn:z_tilde_multi}
\tilde{z}_t \equiv z_t - E(z_t \mid \lbrace y_\tau,z_\tau\rbrace_{-\infty<\tau<t}) = \alpha \lambda \varepsilon_{1,t} + \Sigma_v^{1/2} v_t.
\end{equation}
\cref{sec:iv_multi_details} shows that the testable implication of the multiple-IV model is that the cross-spectrum $s_{y\tilde{z}}(\omega)$ has a rank-1 factor structure. The validity of the multiple-IV model can be rejected if and only if this factor structure fails.

When the multiple-IV model is consistent with the distribution of the data, then the identification analysis can be reduced to the single-IV case in \cref{sec:identif_dyn_scale,sec:identif_dyn_interest,sec:point_id}. Specifically, \cref{sec:iv_multi_details} shows that (i) $\lambda$ is point-identified, and (ii) the identified sets for $\alpha$, variance decompositions, and the degree of invertibility are the same as the identified sets that exploit only the scalar instrument
\begin{equation} \label{eqn:iv_lincom}
\breve{z}_t \equiv \frac{1}{\lambda'\var(\tilde{z}_t)^{-1}\lambda}\lambda'\var(\tilde{z}_t)^{-1} \tilde{z}_t.
\end{equation}
Intuitively, $\breve{z}_t \propto E(\varepsilon_{1,t} \mid \tilde{z}_t)$. Because $\breve{z}_t$ is a linear combination of all $n_z$ instruments, the identified sets are narrower than if we had used any one instrument $z_{k,t}$ in isolation.

In \cref{sec:iv_multimulti} we also derive sharp bounds in the more general case of multiple instruments being correlated with \emph{multiple} structural shocks, as in \citet{Mertens2013}.


\section{Practical implementation}
\label{sec:how_to}
We now describe the practical implementation of our inference procedures for variance decompositions and our test of invertibility of the shock of interest. To keep the exposition self-contained, we review some of the conclusions from \cref{sec:identif}. For ease of notation we focus on the case with a single IV $z_t$ in this section. The generalization to multiple instruments is straight-forward, cf. \cref{sec:iv_multi}.

The key \cref{thm:identif_alpha} in \cref{sec:identif} showed that, without further assumptions, variance decompositions and the degree of invertibility are only \emph{partially identified}. That is, even if the sample size were infinite so we knew the autocovariance function of the observed data $W_t \equiv (y_t',z_t)'$ perfectly, we would not be able to exactly pinpoint the true values of these parameters. However, we were able to derive informative \emph{bounds} on the parameters of interest. The remainder of this section gives an overview of how to compute those bounds in practice, how to do inference on the identified set, and how the additional \emph{a priori} assumption of recoverability allows for consistent point estimation of all economic parameters of interest.

As mentioned in the introduction, a Matlab code suite that implements all steps below is available online.

\subsection{Preliminaries: approximating the autocovariance function}
Our bounds are simple functions of the autocovariances of the data $W_t \equiv (y_t',z_t)'$. The first step of our procedure is thus to estimate this autocovariance function. Though various estimators could in principle be used in conjunction with our identification results, we choose here to approximate the distribution of the observed data with a finite-order VAR (we discuss nonparametric consistency below). Note that this is an approximation of the \emph{reduced-form} dynamics of the data; we do not need to assume a \emph{structural} VAR model and the restrictive invertibility assumption that goes with it. Since our analysis in \cref{sec:identif} assumes stationarity, the data should be appropriately transformed and detrended prior to the analysis.\footnote{Because our analysis relies heavily on the spectral density matrix of the data, it is not straight-forward to extend our procedures to work directly with non-stationary data (without prior transformation/detrending). We leave this important topic to future research.}

As a first step, we select the VAR lag length $p$ by a standard information criterion, such as the Akaike Information Criterion (AIC). We then estimate a VAR($p$) model for the data $W_t=(y_t',z_t)'$ by OLS. Finally, we compute the VAR-implied estimates of the autocovariances and cross-covariances of $y_t$ and the projection residual $\tilde{z}_t \equiv z_t - E(z_t \mid \lbrace y_\tau,z_\tau \rbrace_{-\infty < \tau \leq t-1})$. Denote these estimates by $\widehat{\var}(\tilde{z}_t)$, $\widehat{\cov}(\tilde{z}_t, y_{t+h})$, and $\widehat{\cov}(y_t, y_{t-h})$ for $h=0,1,\dots$; see \cref{sec:inference_formulas} for explicit formulas.

\subsection{Pre-test for invertibility}

Though our identification bounds below are valid irrespective of the invertibility of the shocks, some researchers may wish to have available a convenient pre-test of the null hypothesis of invertibility. \cref{thm:identif_R2} showed that the distribution of the data is \emph{consistent} with the shock of interest $\varepsilon_{1,t}$ being invertible if and only if the IV $z_t$ does not Granger cause the vector $y_t$ of macro observables. Intuitively, if the shock is invertible, then lags of the macro observables $y_t$ capture all the forecasting power of lags of the shock $\varepsilon_{1,t}$; hence, lags of the IV $z_t$ (a noisy measure of $\varepsilon_{1,t}$) do not contribute anything to forecasting. We can test the null hypothesis of no Granger causality in the following standard way:
\begin{itemize}
\item \textbf{Reject the null hypothesis of invertibility} of $\varepsilon_{1,t}$ at the chosen significance level if the $n_y \times p$ VAR coefficients on all lags of $z_t$ in all the $y_t$ equations are jointly statistically significant (for example using a Wald test, cf. \citealp[ch. 2.5]{Kilian2017}).
\end{itemize}
Non-rejection should not be interpreted as strong evidence in favor of invertibility: Any valid test of invertibility necessarily has trivial power against some non-invertible alternatives, since it is possible for $z_t$ not to Granger cause $y_t$ even if the shock $\varepsilon_{1,t}$ is non-invertible.\footnote{\citet{Stock2018} develop an invertibility test which directs power against alternatives with impulse response functions that differ substantially from the invertible null. It is not immediately clear whether their test has power against all falsifiable non-invertible alternatives, as our proposed test does.}

\subsection{Estimating the identification bounds}
We now describe how to estimate our identification bounds for the Forecast Variance Ratio (FVR) and the degrees of invertibility and recoverability. Bounds on Variance Decompositions (VDs) and Forecast Variance Decompositions (FVDs) are provided in \cref{sec:identif_vd}.

The bounds all depend on the two scalar quantities
\[\hat{\underline{\alpha}}^2 \equiv \widehat{\var}(E(\tilde{z}_t \mid \lbrace y_\tau \rbrace_{-\infty<\tau<\infty})),\quad \hat{\bar{\alpha}}^2 \equiv \widehat{\var}(\tilde{z}_t),\]
which are lower and upper bounds for $\alpha^2$, cf. \cref{thm:identif_alpha} and the discussion surrounding \eqref{eqn:alpha_lb_lb}. An explicit formula for the somewhat non-standard projection variance $\hat{\underline{\alpha}}^2$ -- as well as other, similar projection variances mentioned below -- is given in \cref{sec:inference_formulas}.
\begin{itemize}
\item The estimated \textbf{bounds for the Forecast Variance Ratio} $\mathit{FVR}_{i,\ell}$ of variable $y_{i,t}$ at horizon $\ell$ are given by
\begin{equation} \label{eqn:implement_fvr}
\left[\frac{1}{\hat{\bar{\alpha}}^2} \times \frac{\sum_{m=0}^{\ell-1} \widehat{\cov}(y_{i,t},\tilde{z}_{t-m})^2}{\widehat{\var}(y_{i,t+\ell} \mid \lbrace y_\tau \rbrace_{-\infty <\tau\leq t})}\;,\; \frac{1}{\hat{\underline{\alpha}}^2} \times \frac{\sum_{m=0}^{\ell-1} \widehat{\cov}(y_{i,t},\tilde{z}_{t-m})^2}{\widehat{\var}(y_{i,t+\ell} \mid \lbrace y_\tau \rbrace_{-\infty<\tau\leq t})}\right],
\end{equation}
cf. equation \eqref{eqn:identif_fvr}. The interval is always non-empty and never collapses to a point. The true FVR is contained in this interval with high probability asymptotically, but the analysis does not allow us to say where in the interval the parameter lies without making further assumptions. The lower bound -- which upon inspection corresponds to pretending that the residualized IV $\tilde{z}_t$ is a perfect measure of $\varepsilon_{1,t}$ -- is closer to the true FVR when the IV is stronger (i.e., there is less measurement error), cf. equation \eqref{eqn:identif_alpha_param}. The upper bound instead does not depend on the amount of measurement error, and is closer to the true FVR when the macro variables $y_t$ are more informative about the hidden shock $\varepsilon_{1,t}$ (in the sense that the degree of recoverability $R_\infty^2$ is larger).

\item The estimated \textbf{bounds for the degree of invertibility} $R_0^2$ are given by
\[\left[\frac{1}{\hat{\bar{\alpha}}^2} \times \widehat{\var}(E(\tilde{z}_t \mid \lbrace y_\tau \rbrace_{-\infty< \tau\leq t})) \;,\; \frac{1}{\hat{\underline{\alpha}}^2} \times \widehat{\var}(E(\tilde{z}_t \mid \lbrace y_\tau \rbrace_{-\infty< \tau\leq t}))\right],\]
cf. equation \eqref{eqn:identif_R2}. The data are consistent with substantial non-invertibility of the shock $\varepsilon_{1,t}$ if the above interval contains values substantially below 1. By the definition of $\hat{\underline{\alpha}}^2$, this is the case if \emph{future} values of the macro observables help to predict the residualized IV $\tilde{z}_t$.

\item The estimated \textbf{bounds for the degree of recoverability} $R_\infty^2$ are given by
\[\left[\frac{\hat{\underline{\alpha}}^2}{\hat{\bar{\alpha}}^2}\;,\; 1\right],\]
cf. \eqref{eqn:identif_R2}. The data are consistent with substantial non-recoverability of the shock of interest $\varepsilon_{1,t}$ if the interval contains values substantially below 1. The reason why the upper bound above equals the trivial bound of 1 is that we do not exploit the \emph{sharp} upper bound, which is difficult to estimate in realistic sample sizes, as discussed at the end of \cref{sec:identif_dyn_scale}. The theoretical sharp upper bound was derived in \cref{thm:identif_alpha}.

\end{itemize}

\paragraph{Point identification/estimation under recoverability.}
Finally, our analysis in \cref{sec:point_id} showed that it is possible to point-identify many of the parameters of interest if the researcher is willing to impose additional \emph{a priori} assumptions. In particular, if we are willing to assume that the shock is recoverable -- i.e., $R_\infty^2=1$ -- then the \emph{upper bound} for $\mathit{FVR}_{i,\ell}$ in \eqref{eqn:implement_fvr} is a consistent estimator of the true FVR, as argued in \cref{sec:point_id}. As discussed in \cref{sec:params,sec:illustration}, recoverability is a mathematically and economically weaker assumption than the invertibility assumption required by conventional SVAR-IV analysis.

\subsection{Confidence intervals}

In \cref{sec:var_sieve} we prove that the above-mentioned bounds are jointly asymptotically normal under weak \emph{nonparametric} regularity conditions on the data generating process (DGP). We assume neither that the true DGP is a finite-order VAR, nor that the shocks are Gaussian. This argument requires the VAR lag length $p=p_T$ used for estimation to diverge with the sample size $T$ at an appropriate rate.

Since the bounds are asymptotically normal, we can use standard arguments to construct confidence sets \citep{Imbens2004}. Consider any one of the partially identified parameters discussed above and denote the estimates of its bounds by the generic notation $[\hat{\underline{\theta}},\hat{\bar{\theta}}]$. We then use a conventional bootstrap for VAR models \citep[ch. 12.2]{Kilian2017} to generate bootstrap samples of the bound estimates $\hat{\underline{\theta}}$ and $\hat{\bar{\theta}}$, and let $\hat{\underline{q}}_{\beta}$ and $\hat{\bar{q}}_{\beta}$ denote the bootstrap $\beta$-quantiles of the lower and upper bounds, respectively. Then the interval $[\hat{\underline{q}}_{\beta/2},\hat{\bar{q}}_{1-\beta/2}]$ is a valid $1-\beta$ confidence interval for the identified set of the parameter in question.\footnote{The validity requires that the VAR bootstrap procedure is consistent. For example, the bootstrap must take into account the conditional heteroskedasticity of the data. See \citet[ch. 12.2]{Kilian2017} for a menu of procedures, formal results, and regularity conditions.} That is, the probability that the confidence interval contains the \emph{entire} identified set is greater than or equal to $1-\beta$ asymptotically; in particular, this confidence interval is therefore also a valid confidence set for the parameter itself.\footnote{In principle one could construct narrower confidence intervals that only guarantee coverage of the parameter itself (not the identified set), as in \citet{Imbens2004} and \citet{Stoye2009}, and we do this in our Matlab code suite. However, the decrease in length appears to be minimal in realistic applications.} Under the additional point-identifying assumption that the shock is recoverable, the FVR is consistently estimated by the upper bound $\hat{\bar{\theta}}$, so we can construct a $1-\beta$ confidence interval as $[\hat{\bar{q}}_{\beta/2},\hat{\bar{q}}_{1-\beta/2}]$.

Because VAR inference is subject to well-known small-sample biases \citep{Kilian2017}, we recommend that the following alternative formulas be used. Let $\hat{\underline{\theta}}^*$ and $\hat{\bar{\theta}}^*$ denote the average bootstrap draws of $\hat{\underline{\theta}}$ and $\hat{\bar{\theta}}$. Then we report the bias-corrected point estimate $[2\hat{\underline{\theta}}-\hat{\underline{\theta}}^*,2\hat{\bar{\theta}}-\hat{\bar{\theta}}^*]$ of the bounds, as well as Hall's percentile confidence interval $[2\hat{\underline{\theta}}-\hat{\underline{q}}_{1-\beta/2}, 2\hat{\bar{\theta}} -\hat{\bar{q}}_{\beta/2}]$. Similar corrections can be applied in the case of point identification via recoverability.


\section{Application to monetary policy shocks}
\label{sec:mp_application}

To illustrate our method, we revisit an old question: the importance of monetary shocks for U.S. macro fluctuations. Our main result is that monetary shocks are of limited importance for post-1990 aggregate dynamics, especially for inflation. The application illustrates that our upper bound on variance decompositions can yield surprisingly sharp inference, despite the weakness of our identifying assumptions.

\paragraph{Background.}
\citet{Gertler2015} construct an external instrument for monetary shocks from high-frequency changes in asset prices in very short time windows around FOMC announcements, following earlier work by \citet{Kuttner2001}, \citet{Cochrane2002}, and \citet{Gurkaynak2005}. While \citet{Gertler2015} focus on estimation of relative IRFs, we will seek to quantify shock importance, taking the validity of their instrument as given.\footnote{\citet{Caldara2019} compute FVDs for a similar specification, assuming an SVAR model. Their estimates of the importance of monetary shocks for inflation are somewhat larger than our upper bounds.} This setting is ideal for illustrating the appeal of our method, for two reasons.

First, measurement error is likely to be substantial. Intuitively, while short time windows around FOMC meetings may be a clean way of isolating \emph{some} monetary shocks, all shocks occurring outside of that window are necessarily missed. Moreover, financial data are subject to noise due to market microstructure effects and uninformed traders. Treating the IV as the shock -- as in the method of \citet{Gorodnichenko2017}, which is equivalent to our lower bound -- will then understate the importance of monetary shocks due to attenuation bias.\footnote{Formally, let the total monetary shock consist of two independent components, $\varepsilon_{1,t} \equiv \bar{\varepsilon}_{1,t} + \tilde{\varepsilon}_{1,t}$, where $\bar{\varepsilon}_{1,t}$ captures those shocks that occur inside FOMC announcement windows. Assume $\lbrace\bar{\varepsilon}_{1,t}, \tilde{\varepsilon}_{1,t} \rbrace$ are independent of $\lbrace \varepsilon_{2,t},\dots,\varepsilon_{n_\varepsilon,t}\rbrace$. If $z_t = \bar{\varepsilon}_{1,t} + \bar{v}_t$, where the noise $\bar{v}_t$ is independent of $\lbrace \bar{\varepsilon}_{1,t},\tilde{\varepsilon}_{1,t}, \varepsilon_{2,t},\dots,\varepsilon_{n_\varepsilon,t}\rbrace$, then the IV moment conditions \eqref{eqn:iv_moment_conds} are satisfied. For the case $\bar{v}_t=0$, our results in \cref{sec:identif} imply that the \citet{Gorodnichenko2017} FVR estimator will be biased downward by a factor of $\var(\bar{\varepsilon}_{1,t}) \in [0,1]$.} Second, non-invertibility is a threat to SVAR-IV analysis. For example, \citet{Ramey2016}, citing the increasing prevalence of forward guidance in the conduct of U.S. monetary policy, cautions against the conventional SVAR-IV approach. In contrast, our partial identification approach does not require the shock to be invertible (or even recoverable).

\paragraph{Model.}
Our specification largely follows \citet{Gertler2015}, except that we do not impose a SVAR structure. We consider four endogenous macro variables $y_t$: output growth (log growth rate of industrial production), inflation (log growth rate of CPI inflation), the Federal Funds Rate (FFR), and the Excess Bond Premium of \citet{Gilchrist2012} as a measure of the non-default-related corporate bond spread. For robustness, we also try replacing the FFR with the 1-year Treasury rate, as in \citet{Gertler2015}. The external IV $z_t$ is constructed from changes in 3-month-ahead futures prices written on the FFR, where the changes are measured over short time windows around Federal Open Market Committee monetary policy announcement times.\footnote{See \citet{Gertler2015} for details on the construction of the IV and a discussion of the exclusion restriction. \citet{Nakamura2018} argue that the monetary shock identified using this IV partially captures revelation of the Federal Reserve's superior information about economic fundamentals. This is related to the idea in \citet{Campbell2012} that monetary policy communication can be both ``Delphic'' and ``Odyssean''. \cref{sec:iv_multimulti} shows that our FVR bounds can generally be interpreted as bounding the importance of the particular linear combination of shocks that tend to hit during FOMC announcements, e.g., a weighted sum of ``Delphic'' and ``Odyssean'' shocks.} Data are monthly from January 1990 to June 2012. The AIC selects $p=6$ lags in the reduced-form VAR. We use 1,000 bootstrap draws from a homoskedastic recursive residual VAR bootstrap.

\paragraph{Results.}
The data are consistent with substantial non-invertibility. \cref{tab:empir_R2} shows point estimates and 90\% confidence intervals for the identified sets of the degree of invertibility and the degree of recoverability, either using the FFR or the 1-year rate as the interest rate variable. When we use the FFR, we can reject invertibility at the 10\% level, since the confidence set for the degree of invertibility excludes 1. When we use the 1-year rate, we cannot outright reject invertibility, but the confidence set is still consistent with very low degrees of invertibility.\footnote{The p-values for the Granger causality pre-test of invertibility in \cref{sec:how_to} are 0.0001 (FFR) and 0.390 (1-year rate). Note that \citet{Stock2018} fail to reject invertibility in a somewhat different specification.} Since the data cannot rule out a low degree of invertibility in either case, we proceed with our invertibility-robust SVMA-IV analysis. The data are similarly consistent with a wide range of values for the degree of recoverability.

\begin{table}[p]
\centering
\textsc{Empirical application: Degree of invertibility/recoverability} \\[2ex]
\renewcommand{\arraystretch}{1.5}
\begin{tabular}{|l|l|c|c|}
\hline
\multicolumn{2}{|c|}{ } & FFR & 1-year rate \\
 \hline
$R_0^2$ & Bound estimates & $[ 0.196, 0.684]$ & $[0.118,    0.922]$ \\
 & 90\% conf. interval &  $[0.097, 0.877]$ & $[0.029,    1.000]$ \\
 \hline
$R_\infty^2$ & Bound estimates & $[0.282, 1.000]$ & $[0.119, 1.000]$ \\
 & 90\% conf. interval & $[0.190, 1.000]$ & $[0.028, 1.000]$\\
 \hline
\end{tabular}
\caption{Bounds on the degree of invertibility $R_0^2$ and the degree of recoverability $R_\infty^2$. Interest rate variable is either Federal Funds Rate (left) or 1-year Treasury rate (right). All numbers are bootstrap bias corrected.}
\label{tab:empir_R2}
\end{table}

\begin{figure}[p]
\centering
\textsc{Empirical application: Forecast variance ratios} \\[2ex]
\includegraphics[width=0.48\linewidth]{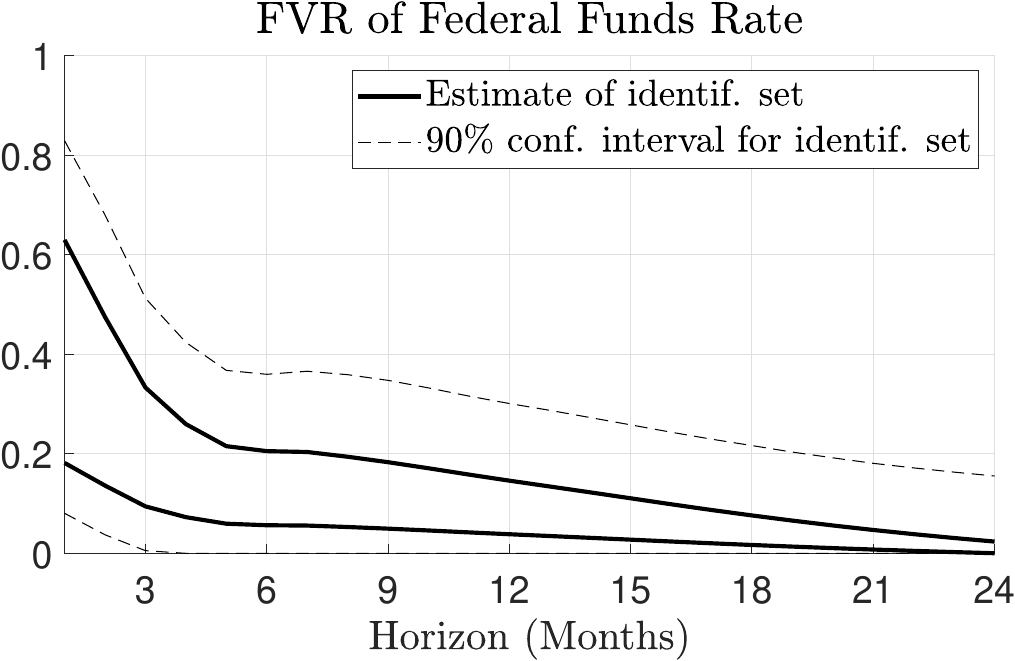} \includegraphics[width=0.48\linewidth]{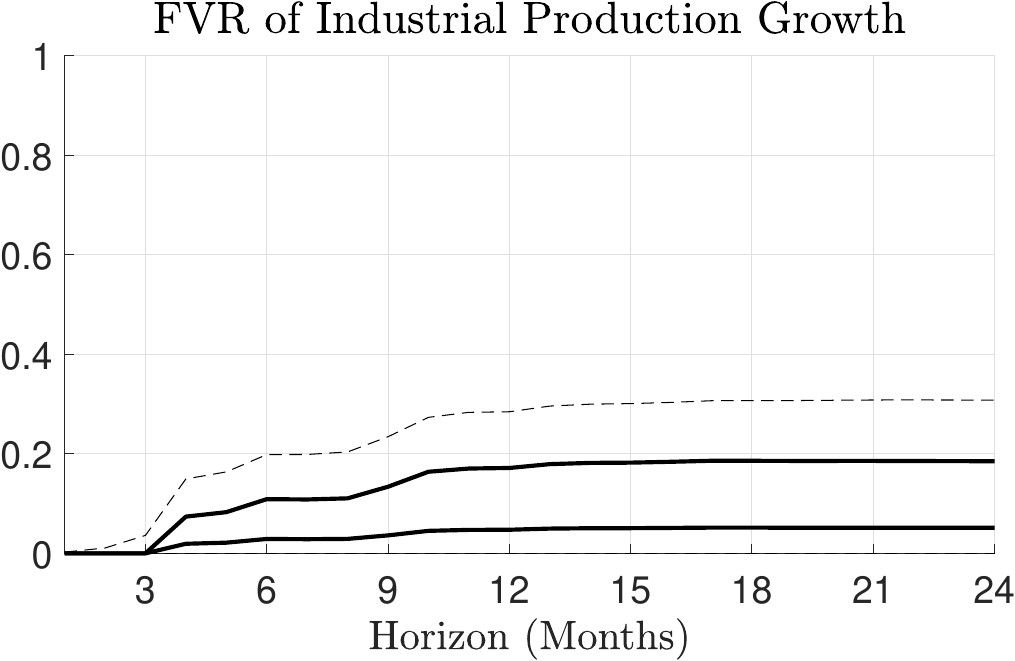} \\[2ex]
\includegraphics[width=0.48\linewidth]{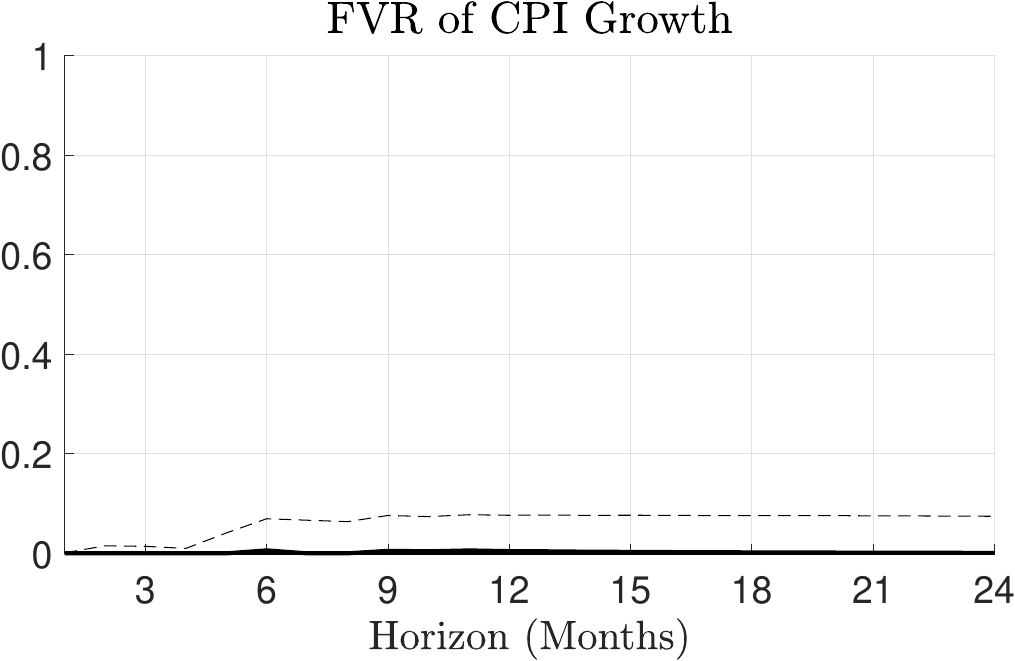} \includegraphics[width=0.48\linewidth]{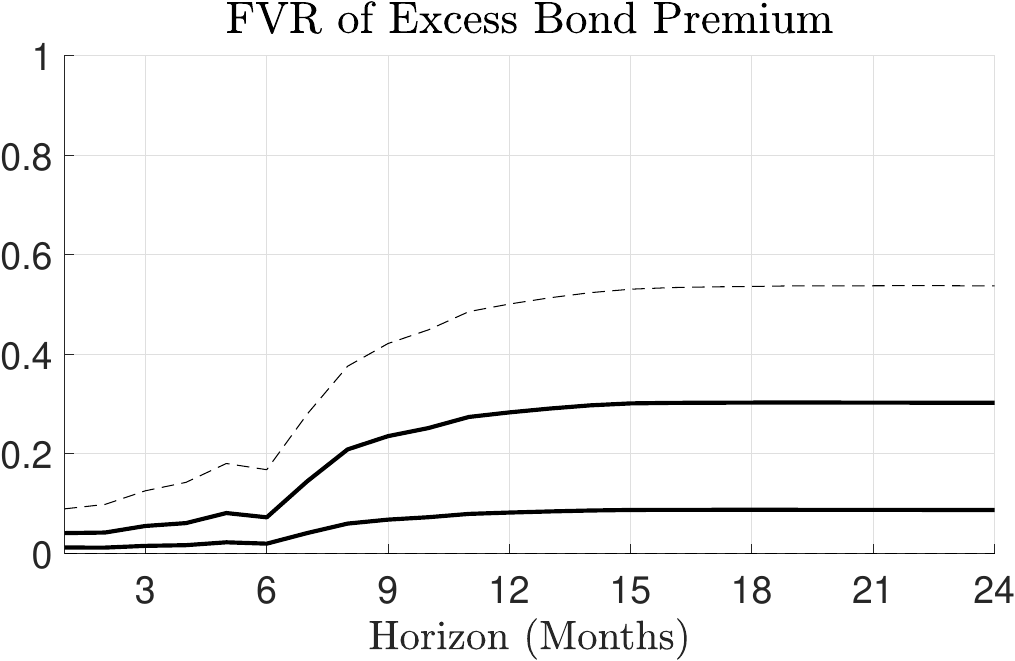}
\caption{Point estimates and 90\% confidence intervals for the identified sets of forecast variance ratios, across different variables and forecast horizons. For visual clarity, we force bias-corrected estimates/bounds to lie in $[0,1]$. The interest rate variable is the Federal Funds Rate.}
\label{fig:gk_fvr}
\end{figure}

\cref{fig:gk_fvr} shows partial identification robust confidence intervals for the forecast variance ratio of the four endogenous macro variables with respect to the monetary shock. We report point estimates and confidence intervals for the identified sets at each horizon separately. We focus here on the specification with the FFR instead of the 1-year rate, since our quantitative conclusions are if anything even starker with the latter observable. At all forecast horizons, the 90\% confidence intervals rule out FVRs above 31\% for output growth and 8\% for inflation. At forecast horizons up to 6 months, we can rule out that the monetary shock accounts for more than 19\% of the forecast variance of the Excess Bond Premium. However, we cannot rule out that the monetary shock is an important contributor to medium- or long-run forecasts of the bond premium. On the other hand, we cannot rule out that the monetary shock is completely unimportant either.

Our analysis reveals that the weak assumptions of the SVMA-IV model suffice to obtain tight upper bounds on the forecast variance contribution of monetary shocks for several variables, especially inflation. This is despite the finding by \citet{Stock2018} that standard errors for impulse response functions are large in this application. Many commentators have documented a recent divorce between inflation and output dynamics \citep{Hall2011}; our results document a similar divorce in dynamics \emph{conditional} on monetary policy shocks in post-1990 data. Although this finding echoes previous SVAR work \citep{Christiano1999,Ramey2016}, our identifying assumptions are weaker -- we merely impose validity of the IV.\footnote{\cref{sec:empir_add} reports variance decompositions obtained from a conventional SVAR-IV procedure. These results confirm the limited importance of the monetary shock, though under stronger identifying assumptions.} We conclude that, if inflation is a monetary phenomenon, it is so because of the systematic component of monetary policy, not because of erratic policy shocks.

\paragraph{Other application: Oil news shocks.}
In \cref{sec:oil} we show that our method also yields highly informative upper bounds on the importance of international oil supply news shocks for the U.S. and global business cycles. We use an IV constructed by \citet{Kaenzig2021} from OPEC announcements.\footnote{Our empirical specification otherwise differs somewhat from his because we work with stationarity-transformed variables and restrict the sample to the period where the IV is available.} We find the oil news shock to be highly non-invertible, causing conventional SVAR-IV analysis to reach several spurious conclusions.


\section{Analytical illustrations}
\label{sec:illustration}

In this section we consider three simple analytical examples that illustrate how our identification bounds depend on the characteristics of the data. \cref{sec:identif_dyn_scale} argued that the tightness of our \emph{lower} bound on variance decompositions depends solely on the strength of the IV. We here show that, under stylized but empirically motivated assumptions, the \emph{upper} bound can be expected to be highly informative, in the sense that it at worst mildly overstates the instrumented shock's contribution to macroeconomic fluctuations.

Throughout this section we assume the availability of a single IV $z_t = \alpha\varepsilon_{1,t} + \sigma_v v_t$, and then consider different illustrative toy models for the $y_t$ variables, as specified below. In \cref{sec:sw_illustration} we extend the analytical intuition below to the much richer quantitative DSGE model developed by \citet{Smets2007}.

\subsection{Information content of several observables}
\label{subsec:info_obs}

As our first example, consider the static model
\begin{equation*}
y_t = \Theta_0 \varepsilon_t
\end{equation*}
with $n_y=2$ observables: inflation ($y_{1,t}$) and the monetary policy interest rate ($y_{2,t}$). We think of $\varepsilon_{1,t}$ as a conventional monetary policy shock. If we only observed inflation $y_{1,t}$ in addition to an IV, we would not be able to rule out that the monetary shock drives all the variation in this single variable, as explained in \cref{sec:identif_stat}. How can adding the interest rate to the data set help tighten the identified set for the FVR of inflation?

To gain economic intuition, let the reduced-form moments of the data be given by
\begin{equation*}
\var(y_t) = \left(\begin{smallmatrix} 1 & \rho \\ \rho & 1 \end{smallmatrix}\right), \quad \text{and} \quad \cov(y_t, z_t) = \left(\begin{smallmatrix} 1 \\ -\zeta \end{smallmatrix}\right),
\end{equation*}
where $\zeta \geq 0$. We see that, even with the amount of measurement error unknown, the IV $z_t$ reveals the signs and the relative magnitudes of the co-movement in observables induced by monetary shocks: The shock $\varepsilon_{1,t}$ moves inflation and interest rates in opposite directions \citep{Uhlig2005}, while the unconditional correlation of interest rates and inflation is $\rho$.

We consider three instructive special cases. For the first two, we set $\zeta = 1$; applying our identification analysis, we then get the bounds
\begin{equation*}
\mathit{FVR}_{i,0} \leq \frac{1}{2} (1-\rho),\quad i=1,2.
\end{equation*}
Now suppose first that $\rho = -1$; that is, interest rates and inflation are not just perfectly negatively correlated conditional on monetary shocks, but also unconditionally. In that case the upper bound for the FVR of both variables equals 1: The data cannot rule out that the correlation of the IV with macro observables is imperfect purely because of measurement error. Second, suppose that $\rho = 1$; that is, interest rates and inflation are perfectly positively correlated in the data. Then our upper bound for the FVR suddenly equals zero: The monetary shock induces co-movement patterns that we never see in the data, so it cannot possibly explain any observed macro fluctuations. Third, if instead $\rho = 0$, then our upper bounds are, for any $\zeta \geq 0$,
\begin{equation*}
\mathit{FVR}_{1,0} \leq \frac{1}{1 + \zeta^2}, \quad \mathit{FVR}_{2,0} \leq \frac{\zeta^2}{1 + \zeta^2}.
\end{equation*}
Suppose that nominal rates respond much more to the monetary shock than inflation does, i.e., $\zeta \gg 1$. Then the upper bound on the inflation variance decomposition $\mathit{FVR}_{1,0}$ is very small; intuitively, since the IV reveals that the monetary shock moves interest rates by much more than inflation, but both have the same \emph{unconditional} variance, the monetary shock cannot possibly be an important driver of inflation. This third example rationalizes the findings in our application to monetary shocks in \cref{sec:mp_application}: The IV $z_t$ correlates much more with interest rates than with prices, yet prices are not commensurately less volatile than interest rates, so monetary shocks cannot account for much of the volatility in prices.

This example shows that our upper bound on the FVR is close to the true value if either the shock is very prominent (so that the bound of one is not far from the truth) or if the shock induces somehow atypical co-movements of the various observed macro aggregates. This second condition is equivalent to the shock being prominent for some \emph{linear combination} of the macro observables $y_t$, which is equivalent to the shock being nearly invertible in this static model. Thus, the preceding arguments agree with the analysis in \cref{sec:identif_stat}.

\subsection{Dynamic information content}
\label{subsec:info_obs_dyn}

Whereas the previous example illustrated how the availability of several macro time series sharpens identification in a static context, we now show how the dynamics of individual time series can do the same. Consider the univariate but dynamic model
\begin{equation*}
y_t = \sum_{j = 1}^{n_\varepsilon} \sum_{\ell=0}^\infty\rho_j^\ell \varepsilon_{j,t-\ell}
\end{equation*}
with $n_y = 1$. That is, we observe a single variable $y_t$ driven by $n_\varepsilon$ independent AR(1) processes. To fix ideas, we think of $\varepsilon_{1,t}$ as a technology shock and $y_t$ as aggregate output.

Now assume that long-run fluctuations in output $y_t$ are exclusively driven by the technology shock $\varepsilon_{1,t}$; that is, consider the limit $\rho_1 \rightarrow 1$, while fixing $|\rho_j|<1$ for all $j \geq 2$. In this case, the sharp lower bound on $\alpha^2$ converges to the truth:\footnote{Note that $s_{y\tilde{z}}(0) = \frac{\alpha}{2\pi}\sum_{\ell=0}^\infty \rho_1^\ell = \frac{\alpha}{2\pi} \times \frac{1}{1-\rho_1}$ and $s_y(0) = \frac{1}{2\pi}(\sum_{j=1}^{n_\varepsilon}\sum_{\ell=0}^\infty \rho_j^\ell)^2 = \frac{1}{2\pi}(\sum_{j=1}^{n_\varepsilon} \frac{1}{1-\rho_j})^2$. \cref{fn:bw} then implies $2\pi s_{\tilde{z}^\dagger}(0) = s_{y\tilde{z}}(0)^2/s_y(0) = ((1-\rho_1)s_{y\tilde{z}}(0))^2/((1-\rho_1)s_y(0))  \to \alpha^2$ as $\rho_1\to 1$.}
\[\lim_{\rho_1 \to 1} \alpha_{LB}^2 = \lim_{\rho_1 \to 1} 2 \pi \sup_{\omega \in [0,\pi]} s_{\tilde{z}^\dagger}(\omega) = \lim_{\rho_1 \to 1} 2 \pi s_{\tilde{z}^\dagger}(0) = \alpha^2.\]
Intuitively, at spectral frequency zero, all fluctuations in $y_t$ are driven by the technology shock. Loosely speaking, applying a low-pass filter to $y_t$ therefore isolates the fluctuations caused by the technology shock. Leads and lags of this low-pass filtered series are thus highly correlated with the IV, putting a lower bound on the signal-to-noise ratio in the IV. This is why the sharp upper bound for the FVR converges to the true value.  

The example reveals that cross-restrictions over time can be highly informative even if the shock of interest is neither invertible nor recoverable. Intuitively, for the sharp upper bound on the FVR to bind, our method only needs the shock to dominate at some frequency; the across-frequency restrictions then do the rest, exactly like the cross-variable restrictions in the static example above.\footnote{As mentioned in \cref{sec:identif_dyn_scale}, for finite-sample statistical reasons, we recommend the use of a weaker lower bound $\underline{\alpha}^2$ on $\alpha^2$ in place of the sharp bound $\alpha_{LB}^2$. This weaker lower bound does not converge to $\alpha^2$ as $\rho_1 \to 1$, unless $\varepsilon_{1,t}$ is recoverable. However, as discussed in \cref{fn:fix_interval_sd}, researchers may leverage a strong prior belief about the low-frequency importance of shocks by computing the integral in \eqref{eqn:alpha_lb_lb} for a pre-specified range of (low) frequencies.}

\subsection{Non-invertibility and news shocks}
\label{sec:illustration_news}

In the third example, we show how our method deals with non-invertible news shocks. First discussed in \citet{Pigou1927}, news shocks have recently received much attention as drivers of macroeconomic fluctuations \citep{Beaudry2006,Beaudry2014,Jaimovich2009,Schmitt2012}. Unfortunately, foresight of economic agents complicates conventional SVAR-based analysis since it induces equilibria with non-invertible MA representations \citep{Leeper2013}. In contrast, our methods are valid irrespective of invertibility.

To illustrate, consider a moving average model of order 1 with $n_y = n_\varepsilon = 2$:
\begin{equation*}
y_t = (1 + \zeta L) \Theta_0 \varepsilon_t,
\end{equation*}
where $\zeta > 1$. As is well known, this assumption implies that the moving average representation is non-invertible. We think of $\varepsilon_{1,t}$ as a monetary forward guidance shock: The shock moves inflation and nominal interest rates by more tomorrow (when the shock directly hits the monetary policy rule) than today (when the news is revealed).

The conventional SVAR-IV approach mis-measures the FVR because of non-invertibility. By standard arguments \citep[e.g.,][]{Leeper2013} the reduced-form VAR residuals equal
\begin{equation}
u_t \equiv y_t- E(y_t \mid \lbrace y_{\tau}\rbrace_{-\infty<\tau < t}) = \Theta_0 \varepsilon_{t} + (1-1/R_0^2) \sum_{\ell=1}^\infty \left( -\zeta\right)^{-\ell} \Theta_0 \varepsilon_{t-\ell}, \label{eq:u_SVAR}
\end{equation}
where the degree of invertibility equals
\begin{equation*}
R_0^2 = \zeta^{-2}.
\end{equation*}
Since SVAR procedures assume that the structural shocks $\varepsilon_t$ can be obtained as linear functions of the reduced-form residuals $u_t$, equation \eqref{eq:u_SVAR} shows that any SVAR analysis will conflate the explanatory power of the shock $\varepsilon_{1,t}$ with that of its lags. As a consequence, \cref{sec:svar_iv} shows that the SVAR-IV estimand of the FVR overstates the contribution of the shock to one-step-ahead forecasts:
\begin{equation*}
\mathit{FVR}_{1, 0}^{\mathit{SVAR-IV}} = \frac{1}{R_0^2} \times \mathit{FVR}_{1, 0} > \mathit{FVR}_{1,0}.
\end{equation*}
Clearly, the population bias of the SVAR-IV estimand worsens as the degree of invertibility $R_0^2$ decreases to 0 \citep[see also][]{Forni2018}. In the oil news shock application in \cref{sec:oil} we demonstrate that the SVAR-IV bias can be large in practice.

In contrast, our identification bounds are valid irrespective of invertibility, since we do not assume that $\varepsilon_{1,t}$ can be recovered as a function of only the contemporaneous VAR residuals $u_t$. In fact, in this model with as many observables as shocks, both shocks $\varepsilon_t=(\varepsilon_{1,t},\varepsilon_{2,t})'$ are recoverable.\footnote{In particular, $\varepsilon_{t} = -R_0^2 \Theta_0^{-1} u_t - \left(1 - R_0^2 \right)\sum_{\ell=1}^\infty \left( -\zeta\right)^{-\ell} \Theta_0^{-1} u_{t+\ell}$.} Hence, if we exploit this knowledge, we can even point-identify the shock as $\varepsilon_{1,t} \propto z_t^\dagger = E(z_t \mid \lbrace y_\tau \rbrace_{-\infty<\tau<\infty})$. The key is that our method can use the \emph{future} values of nominal rates and inflation, $y_\tau$, $\tau \geq t$, to recover the forward guidance shock $\varepsilon_{1,t}$ at time $t$. In so doing, it effectively realigns the information sets of the economic agents and of the econometrician, sidestepping the invertibility problem.

\section{Simulation study}
\label{sec:simulation}

We finish by showing that our inference procedures have good finite-sample performance in simulations. Our methods continue to work well in non-invertible models, unlike the conventional SVAR-IV procedure.

\paragraph{DGP.}
We adopt a variant of the DGP in \citet{Kilian2011} and assume that the macro aggregates $y_t$ follow a structural VARMA($p$,1) model:
\begin{equation*}
y_t = \textstyle \sum_{\ell=1}^p \Xi_\ell y_{t-\ell} + \Theta_0 (\varepsilon_t + \zeta \varepsilon_{t-1}).
\end{equation*}
We consider $n_y = 2$ macro variables, $p=1$ autoregressive lag (with one exception discussed below), and set $\Xi_1 = \left(\begin{smallmatrix} \rho_y & 0 \\ 0.5 & 0.5 \end{smallmatrix}\right)$. For the MA part, we consider $n_\varepsilon = 2$ shocks (which are thus both recoverable) and set $\Theta_0 = \text{chol} \left(\begin{smallmatrix} 1 & 0.8 \\ 0.8 & 1 \end{smallmatrix}\right),$ where ``chol'' denotes the lower triangular Cholesky decomposition. As in \cref{sec:illustration_news}, $\zeta$ is a scalar parameter that governs the degree of invertibility, with $\zeta>1$ implying non-invertibility. We add an external instrument $z_t$ for the shock of interest $\varepsilon_{1,t}$:
\begin{equation*}
z_t = \rho_z z_{t-1} + \rho_{zy} (y_{1,t-1} + y_{2,t-1}) + \varepsilon_{1,t} + \sigma_v v_t.
\end{equation*}
Notice that we have normalized $\alpha = 1$. Finally, the measurement error and structural shocks are i.i.d. Gaussian and orthogonal as in \cref{asn:shocks}.

We run Monte Carlo experiments for nine different parameterizations of the above DGP. Specifically, we consider various deviations from a baseline parametrization. In our benchmark, we set $\rho_y = 0.5$, $\rho_z = \rho_{zy} = 0$, $\zeta = 0$, $\sigma_v = 1$, and sample size $T = 250$. We then consider variations with more autoregressive persistence (either $\rho_y = 0.9$, or $\rho_z = 0.8$ and $\rho_{zy} = 0.3$), an invertible MA component ($\zeta=0.5$), a non-invertible MA component ($\zeta = 2$), a weaker instrument ($\sigma_v = 2$), and different sample sizes ($T = 100$, $T = 500$). Finally, we allow for richer dynamics, with $p=4$ and $\Xi_j = \frac{1}{j^2} \Xi_1$ for $j = 2,3,4$.

\paragraph{Results.}
Our parameters of interest are the degree of invertibility $R_0^2$ and the FVR for variable $y_{2,t}$ at horizons $1$ and $4$. We conduct $5,000$ Monte Carlo repetitions per DGP, and construct confidence intervals at the 90\% level using $1,000$ bootstrap draws per simulation. We use a homoskedastic recursive residual bootstrap. The reduced-form VAR lag length is selected using AIC, and we use Hall's percentile bootstrap confidence interval, cf. \cref{sec:how_to}.

\begin{sidewaystable}[p]
\centering
\textsc{Monte Carlo study: Coverage rates of confidence intervals} \\[2ex]
\small
\renewcommand{\arraystretch}{1.5}
\begin{tabular}{|l|C{0.8cm}C{1.2cm}C{1.2cm}|C{1.2cm}C{1.2cm}|C{1.2cm}C{1.2cm}C{1.2cm}|C{1.2cm}C{1.2cm}C{1.2cm}|}
\hline 
 & \multicolumn{3}{c|}{True parameter} & \multicolumn{2}{c|}{Coverage $R_0^2$} & \multicolumn{3}{c|}{Coverage $\mathit{FVR}_{2,1}$} & \multicolumn{3}{c|}{Coverage $\mathit{FVR}_{2,4}$} \\
\cline{2-12}
\multicolumn{1}{|c|}{Experiment} & $R_0^2$ & $\mathit{FVR}_{2,1}$ & $\mathit{FVR}_{2,4}$ & Set & Param & Set & Param & SVAR & Set & Param & SVAR \\
\hline
\multicolumn{1}{|c|}{Baseline} &	1	&	0.64	&	0.819	&	0.935	&	0.999	&	0.898	&	0.955	&	0.867	&	0.880	&	0.948	&	0.838 \\
\multicolumn{1}{|c|}{$\rho_y = 0.9$} &	1	&	0.64	&	0.862	&	0.929	&	1.000	&	0.896	&	0.953	&	0.876	&	0.889	&	0.952	&	0.824 \\
\multicolumn{1}{|c|}{$\rho_z = 0.8$, $\rho_{zy} = 0.3$} &	1	&	0.64	&	0.819	&	0.933	&	1.000	&	0.904	&	0.960	&	0.873	&	0.891	&	0.959	&	0.850 \\
\multicolumn{1}{|c|}{$\zeta = 0.5$} &	1	&	0.64	&	0.839	&	0.937	&	0.998	&	0.892	&	0.944	&	0.870	&	0.853	&	0.928	&	0.828 \\
\multicolumn{1}{|c|}{$\zeta = 2$} &	0.25	&	0.16	&	0.806	&	0.879	&	0.910	&	0.857	&	0.881	&	0.137	&	0.861	&	0.938	&	0.674 \\
\multicolumn{1}{|c|}{$\sigma_v = 2$} &	1	&	0.64	&	0.819	&	0.949	&	0.997	&	0.889	&	0.924	&	0.824	&	0.865	&	0.903	&	0.770 \\
\multicolumn{1}{|c|}{$T = 100$} &	1	&	0.64	&	0.819	&	0.907	&	0.998	&	0.871	&	0.943	&	0.835	&	0.840	&	0.929	&	0.797 \\
\multicolumn{1}{|c|}{$T = 500$} & 1	&	0.64	&	0.819	&	0.931	&	1.000	&	0.900	&	0.954	&	0.881	&	0.884	&	0.947	&	0.869 \\
\multicolumn{1}{|c|}{$p=4$} & 1	&	0.64	&	0.855	&	0.937	&	0.999	&	0.897	&	0.945	&	0.868	&	0.829	&	0.868	&	0.846 \\
\hline
\end{tabular}
\caption{True parameter values and coverage rates of 90\% confidence intervals, constructed as in \cref{sec:how_to}, using 1,000 bootstrap iterations for each Monte Carlo experiment, and 5,000 Monte Carlo experiments per DGP. The DGPs (along rows) are described in the text. ``Set'': probability that SVMA-IV confidence interval covers entire identified set. ``Param'': probability that SVMA-IV confidence interval covers parameter. ``SVAR'': probability that conventional SVAR-IV confidence interval covers parameter.} \label{tab:sim}
\end{sidewaystable}

\cref{tab:sim} shows that the partial identification robust SVMA-IV confidence sets defined in \cref{sec:how_to} achieve coverage rates close to or exceeding the desired level of 90\% throughout. We report coverage rates for both the population identified sets (columns ``Set'') and for the underlying parameters (columns ``Param''). The coverage rate for the parameter is never below 86.8\% in any case. The coverage rate for the identified set is mostly close to 90\% and at worst 82.9\% in our experiments. We also report the coverage rates of conventional SVAR-IV bootstrap confidence intervals for the FVR (columns ``SVAR''). The coverage distortions of our SVMA-IV procedures are almost always smaller than those of the SVAR-IV procedure. Most notably, our procedures have acceptable coverage even in the non-invertible case ($\zeta=2$), whereas the SVAR-IV procedure under-covers severely in this case.\footnote{We acknowledge, however, that in DGPs with only mild non-invertibility, SVAR-IV procedures may be preferable to our more robust SVMA-IV procedure, since the former procedure has fewer parameters to estimate and will be only mildly biased (cf. \cref{sec:svar_iv}).}

We make the following additional remarks. First, coverage deteriorates slightly with noisier/weaker instruments ($\sigma_v=2$), as expected. Our inference methods are not robust to arbitrarily weak instruments ($\sigma_v \to \infty$); we leave this issue to future work. Second, we face some well-known parameter-at-the-boundary issues. For most experiments, $R_0^2 = 1$. This explains the over-coverage of confidence intervals for this parameter and, less so, for the overall identified set. Similar problems would arise if the true FVR were close to $0$. Third, for more persistent DGPs, the AIC tends to select an insufficient number of lags, resulting in moderate under-coverage, in particular for the FVRs at horizon 4. For example, in the experiment with $p = 4$ autoregressive lags, the AIC selects an average lag length of $2.2$.

\section{Conclusion}
\label{sec:concl}

Applied macroeconomists have recently turned to external sources of exogenous variation to identify dynamic causal effects. Though such external instruments or proxies are frequently used to estimate impulse responses, existing methods did not allow researchers to quantify the contribution of individual shocks to business-cycle fluctuations -- a question of first-order interest in traditional business-cycle analysis. We fill this gap by providing identification results and inference techniques for variance decompositions, historical decompositions, and the degree of invertibility. Our methods require neither the absence of measurement error in the external instrument, nor the often dubious assumption that the instrumented shock is invertible (as assumed in conventional SVAR analysis). We prove that the importance of the instrumented shock is generally interval-identified. Point identification can be achieved if the shock is known to be recoverable -- a substantively weaker assumption than invertibility. We provide a software package that implements all steps of our inference procedures. Applying our method to U.S. data, we are able to establish a tight upper bound on the importance of monetary shocks for recent inflation dynamics, despite our weak identifying assumptions.

\clearpage
\appendix
\section{Appendix}
\label{sec:appendix}

\subsection{Formulas for estimation and inference}
\label{sec:inference_formulas}
Here we provide the remaining formulas needed for the inference procedures in \cref{sec:how_to}. 

Let $\hat{A}_1,\dots,\hat{A}_p$ denote the $(n_y+1)\times(n_y+1)$ coefficient matrix estimates for the VAR in $W_t=(y_t',z_t)'$. Let $\hat{\Sigma}$ denote the residual sample variance-covariance matrix. Let $\hat{\Sigma}^{1/2}$ denote any square matrix such that $\hat{\Sigma}^{1/2}\hat{\Sigma}^{1/2\prime} = \hat{\Sigma}$, e.g., the Cholesky factor. Compute the moving average coefficients $\hat{B}(L) \equiv (I_{n_y+1} - \sum_{\ell=1}^p \hat{A}_\ell L^\ell)^{-1}\hat{\Sigma}^{1/2}$ using the familiar recursion
\[\hat{B}_0 = \hat{\Sigma}^{1/2},\quad \hat{B}_h = \textstyle \sum_{\ell=1}^{\min\lbrace h,p\rbrace} \hat{A}_\ell \hat{B}_{h-\ell},\; h \geq 1.\]
Denote the top $n_y$ rows of $\hat{B}_h$ by $\hat{B}_{y,h}$ and the bottom row by $\hat{B}_{z,h}$. Then
\[\widehat{\var}(\tilde{z}_t) \equiv \hat{B}_{z,0}\hat{B}_{z,0}',\quad \widehat{\cov}(\tilde{z}_t, y_{t+h}) \equiv \begin{cases} \hat{B}_{z,0} \hat{B}_{y,h}' & \text{if $h \geq 0$}, \\ 0_{1 \times n_y} & \text{otherwise}, \end{cases}\]
\[\widehat{\cov}(y_t, y_{t-h}) \equiv \textstyle \sum_{\ell = 0}^{\infty} \hat{B}_{y,\ell} \hat{B}_{y,\ell+h}' \text{ for } h \geq 0.\]
In practice, we truncate the infinite sum above at a large value of $\ell$.

Define now the projection variances\footnote{These could alternatively be computed using the Kalman filter, but there appears to be little difference in numerical accuracy or speed relative to the formulas stated here.}
\begin{align*}
\widehat{\var}(E(\tilde{z}_t \mid \lbrace y_\tau \rbrace_{-\infty<\tau<\infty})) &\equiv \hat{\Sigma}_{\tilde{z}, y, (M,M)} \hat{\Sigma}_{y,(M,M)}^{-1} \hat{\Sigma}_{\tilde{z},y,(M,M)}', \\
\widehat{\var}(y_{i,t+\ell} \mid \lbrace y_\tau \rbrace_{-\infty< \tau\leq t}) &\equiv \widehat{\var}(y_{i,t}) - (\widehat{\cov}(y_{i,t+\ell},y_t),\dots,\widehat{\cov}(y_{i,t+\ell},y_{t-M}))\widehat{\Sigma}_{y,(M,0)}^{-1} \\
&  \qquad\qquad\qquad \times (\widehat{\cov}(y_{i,t+\ell},y_t),\dots,\widehat{\cov}(y_{i,t+\ell},y_{t-M}))', \\
\widehat{\var}(E(\tilde{z}_t \mid \lbrace y_\tau \rbrace_{-\infty< \tau\leq t})) &\equiv (\widehat{\cov}(\tilde{z}_t,y_t),0_{1 \times n_yM}) \hat{\Sigma}_{y,(M,0)}^{-1} (\widehat{\cov}(\tilde{z}_t,y_t),0_{1 \times n_yM})',
\end{align*}
where $\hat{\Sigma}_{\tilde{z},y,(M,M)}$ is the estimated covariance vector of $\tilde{z}_t$ and $(y_{t+M}', \dots, y_t', \dots, y_{t-M}')'$ obtained by stacking the estimates $\widehat{\cov}(\tilde{z}_t, y_{t+h})$ defined above, $\hat{\Sigma}_{y,(M,M)}$ is similarly the estimated variance-covariance matrix of $(y_{t+M}',\dots,y_{t}',\dots,y_{t-M}')'$, and $\hat{\Sigma}_{y,(M,0)}$ is the estimated variance-covariance matrix of $(y_t',y_{t-1}',\dots,y_{t-M}')'$.

In these formulas, the integer $M$ is a numerical truncation parameter. For example, we estimate $\var(E(\tilde{z}_t \mid \lbrace y_\tau \rbrace_{-\infty<\tau<\infty}))$ using an estimate of the truncated conditional variance $\var(E(\tilde{z}_t \mid \lbrace y_\tau \rbrace_{t-M\leq\tau\leq t+M}))$. $M$ should exceed at least 50 to yield an accurate approximation. We recommend checking that the numerical results do not change much when $M$ is increased, since the effects of truncation will depend on the persistence of the data.

\subsection{Proofs of main results}
\label{sec:proofs}

\subsubsection{Auxiliary lemma}
\label{sec:proofs_lemma}
\begin{lem} \label{thm:semidef}
Let $B$ be an $n \times n$ Hermitian positive definite complex-valued matrix and $b$ an $n$-dimensional complex-valued column vector. Let $x$ be a nonnegative real scalar. Then $B - x^{-1} bb^*$ is positive (semi)definite if and only if $x >\!(\geq)\, b^*B^{-1}b$.
\end{lem}
\noindent Please find the proof in \cref{sec:proof_semidef}.

\subsubsection{Proof of \texorpdfstring{\cref{thm:identif_alpha}}{Proposition \ref{thm:identif_alpha}}}
Let $\alpha$ and the spectrum $s_w(\omega)$ be given. Define the $n_y$-dimensional vectors
\[\overline{\Theta}_{\bullet, 1,\ell} = \alpha^{-1}\cov(y_t,\tilde{z}_{t-\ell}),\quad \ell \geq 0,\]
and the corresponding vector lag polynomial
\[\overline{\Theta}_{\bullet, 1}(L) = \sum_{\ell=0}^\infty \overline{\Theta}_{\bullet, 1,\ell}L^\ell.\]
Since $\alpha^2 \leq \alpha_{UB}^2$, we may define $\overline{\sigma}_v = \sqrt{\var(\tilde{z}_t)-\alpha^2}$. Since $\alpha^2 > \alpha_{LB}^2$, \cref{thm:semidef} implies that
\[s_y(\omega) - \frac{2\pi}{\alpha^2}s_{y\tilde{z}}(\omega)s_{y\tilde{z}}(\omega)^* = s_y(\omega) - \frac{1}{2\pi}\overline{\Theta}_{\bullet, 1}(e^{-i\omega})\overline{\Theta}_{\bullet, 1}(e^{-i\omega})^*\]
is positive definite for every $\omega \in [0,2\pi]$. Hence, the Wold decomposition theorem \citep[Thm. $2^{\prime \prime}$, p. 158]{Hannan1970} implies that there exists an $n_y \times n_y$ matrix lag polynomial $\tilde{\Theta}(L) = \sum_{\ell=0}^\infty \tilde{\Theta}_\ell L^\ell$ such that\footnote{We can rule out a deterministic term in the Wold decomposition because a continuous and positive definite spectral density satisfies the full-rank condition of \citet[p. 162]{Hannan1970}.}
\[s_y(\omega) - \frac{1}{2\pi}\overline{\Theta}_{\bullet, 1}(e^{-i\omega})\overline{\Theta}_{\bullet, 1}(e^{-i\omega})^* = \frac{1}{2\pi}\tilde{\Theta}(e^{-i\omega})\tilde{\Theta}(e^{-i\omega})^*,\quad \omega \in [0,2\pi].\]
Thus, the following model for $w_t = (y_t',\tilde{z}_t)'$ generates the desired spectrum $s_w(\omega)$:
\begin{align*}
y_t &= \overline{\Theta}_{\bullet, 1}(L) \overline{\varepsilon}_{1,t} + \tilde{\Theta}(L)\tilde{\varepsilon}_t, \\
\tilde{z}_t &= \alpha \overline{\varepsilon}_{1,t} + \overline{\sigma}_v\overline{v}_t, \\
(\overline{\varepsilon}_{1,t},\tilde{\varepsilon}_t',\overline{v}_t)' &\stackrel{i.i.d.}{\sim} N(0, I_{n_y+2}).
\end{align*}
Note that the construction requires only $n_\varepsilon=n_y+1$ shocks, $\overline{\varepsilon}_{1,t} \in \mathbb{R}$ and $\tilde{\varepsilon}_t \in \mathbb{R}^{n_y}$. \qed

\subsubsection{Proof of \texorpdfstring{\cref{thm:identif_R2}}{Proposition \ref{thm:identif_R2}}}
\paragraph{Identified set for $R_0^2$.}
If the identified set contains 1, then there must exist an $\overline{\alpha} \in [\alpha_{LB},\alpha_{UB}]$ and i.i.d., independent standard Gaussian processes $\overline{\varepsilon}_{1,t}$ and $\overline{v}_t$ such that (i) $\tilde{z}_t = \overline{\alpha}\times \overline{\varepsilon}_{1,t} + \overline{v}_t$, (ii) $\overline{v}_t$ is uncorrelated with $y_t$ at all leads and lags, and (iii) $\overline{\varepsilon}_{1,t}$ lies in the closed linear span of $\lbrace y_\tau \rbrace_{-\infty<\tau \leq t}$. This immediately implies the ``only if'' statement.

For the ``if'' part, assume $\tilde{z}_t$ does not Granger cause $y_t$. By the equivalence of Sims and Granger causality, $\tilde{z}_t^\dagger = E(\tilde{z}_t \mid \lbrace y_\tau \rbrace_{-\infty<\tau<\infty})= E(\tilde{z}_t \mid \lbrace y_\tau \rbrace_{-\infty<\tau\leq t})$. Note that the latter best linear predictor is white noise since, for any $\ell \geq 1$,
\begin{align*}
\cov\big(E(\tilde{z}_t \mid \lbrace y_\tau \rbrace_{-\infty<\tau \leq t}), y_{t-\ell}\big) &= \cov(\tilde{z}_t, y_{t-\ell}) - \cov\big(\tilde{z}_t - E(\tilde{z}_t \mid \lbrace y_\tau \rbrace_{-\infty<\tau\leq t}), y_{t-\ell}\big) \\
&= 0 - 0,
\end{align*}
using the fact that $\tilde{z}_t$ is a projection residual. In conclusion, the best linear predictor $\tilde{z}_t^\dagger$ of $\tilde{z}_t$ given $\lbrace y_\tau \rbrace_{-\infty<\tau<\infty}$ depends only on $\lbrace y_\tau \rbrace_{-\infty<\tau \leq t}$ and it has a constant spectrum. From the expression for $\alpha_{LB}^2$, we get that $\alpha_{LB}^2 = \var(E(\tilde{z}_t \mid \lbrace y_\tau \rbrace_{-\infty<\tau \leq t}))$. Hence, expression \eqref{eqn:identif_R2} implies that the upper bound of the identified set for $R_0^2$ equals 1.

\paragraph{Identified set for $R_\infty^2$.}
The upper bound of the identified set for $R_\infty^2$ equals 1 if and only if $2\pi \sup_{\omega \in [0,\pi]} s_{\tilde{z}^\dagger}(\omega)=\var(E(\tilde{z}_t \mid \lbrace y_\tau \rbrace_{-\infty<\tau < \infty}))$. The right-hand side of this equation equals $\var(\tilde{z}_t^\dagger)= \int_0^{2\pi} s_{\tilde{z}^\dagger}(\omega)\, d\omega$. But $\sup_{\omega \in [0,\pi]} s_{\tilde{z}^\dagger}(\omega) = \frac{1}{2\pi}\int_0^{2\pi} s_{\tilde{z}^\dagger}(\omega)\, d\omega$ if and only if $s_{\tilde{z}^\dagger}(\omega)$ is constant in $\omega$ almost everywhere, i.e., $\tilde{z}_t^\dagger$ is white noise. \qed

\clearpage

\phantomsection
\addcontentsline{toc}{section}{References}
\bibliography{decomp_iv_ref}

\end{document}


\title{Online Appendix for: \texorpdfstring{\\}{ } Instrumental Variable Identification of \texorpdfstring{\\}{ } Dynamic Variance Decompositions}
\author{Mikkel Plagborg-M{\o}ller \and Christian K. Wolf}
\date{\today}
\maketitle

\noindent
This online appendix contains supplemental material for the article ``Instrumental Variable Identification of Dynamic Variance Decompositions''. We provide (i) bounds on other notions of variance decompositions, (ii) extensions of the identification analysis to multiple instruments correlated with a single or multiple shocks, (iii) characterizations of the bias of SVAR-IV (or ``proxy SVAR'') procedures under noninvertibility, (iv) an illustration of our method using a quantitative structural macro model, (v) supplementary results for the monetary shock application, (vi) a second set of empirical results on the importance of oil news shocks, and (vii) asymptotic theory on the nonparametric validity of our sieve VAR inference strategy. The end of this appendix contains proofs and auxiliary lemmas.

\bigskip

\noindent \textbf{Any references to equations, figures, tables, assumptions, propositions, lemmas, or sections that are not preceded by ``B.'' refer to the main article.}

\bigskip

{
\small
\tableofcontents
}

\clearpage

\section{Identification and estimation of other variance decomposition concepts}
\label{sec:identif_vd}

Our main analysis focuses on forecast variance ratios as a measure of shock importance, defined in \cref{sec:params}. This appendix defines two additional concepts -- forecast variance decompositions (FVD) and unconditional frequency-specific variance decompositions (VD) -- and discusses the identification and estimation of both.

\paragraph{Definitions.}
The \emph{forecast variance decomposition} (FVD) for variable $i$ at horizon $\ell$ is defined as
\begin{equation} \label{eqn:fvd}
\mathit{FVD}_{i,\ell} \equiv 1 - \frac{\var(y_{i,t+\ell} \mid \lbrace \varepsilon_\tau \rbrace_{-\infty<\tau\leq t}, \lbrace \varepsilon_{1,\tau} \rbrace_{t<\tau<\infty})}{\var(y_{i,t+\ell} \mid \lbrace \varepsilon_\tau \rbrace_{-\infty<\tau\leq t})} = \frac{\sum_{m=0}^{\ell-1} \Theta_{i,1,m}^2}{\sum_{j=1}^{n_\varepsilon} \sum_{m=0}^{\ell-1} \Theta_{i,j,m}^2}.
\end{equation}
The FVD measures the reduction in forecast variance that arises from learning the path of future realizations of the shock of interest, supposing that we already had the history of past structural shocks $\varepsilon_t$ available when forming our forecast. Because the econometrician generally does not observe the structural shocks directly, the FVD is best thought of as reflecting forecasts of economic agents who observe the underlying shocks. The FVD always lies between 0 and 1, purely reflects fundamental forecasting uncertainty, and equals 1 if the first shock is the only shock driving variable $i$ in equation \eqref{eqn:svma}. The software package Dynare reports FVDs after having estimated a DSGE model.

While the FVR and FVD concepts generally differ, they coincide in the case where all shocks are invertible, since in that case the information set $\lbrace y_\tau \rbrace_{-\infty<\tau\leq t}$ equals the information set $\lbrace \varepsilon_\tau \rbrace_{-\infty<\tau\leq t}$. This explains why the SVAR literature has not made the distinction between the two concepts.\footnote{\citet{Forni2018} point out the bias caused by noninvertibility when estimating the FVD using SVARs.}

Our second additional concept is the frequency-specific \emph{unconditional} variance decomposition (VD) of \citet[Sec. 3.4]{Forni2018}. The VD for variable $i$ over the frequency band $[\omega_1,\omega_2]$ is given by
\begin{equation} \label{eqn:vd}
\mathit{VD}_i(\omega_1,\omega_2) \equiv\frac{\int_{\omega_1}^{\omega_2} |\Theta_{i,1}(e^{-i \omega})|^2 \,d\omega}{\sum_{j=1}^{n_\varepsilon} \int_{\omega_1}^{\omega_2} |\Theta_{i,j}(e^{-i \omega})|^2 \,d\omega},\quad 0\leq \omega_1<\omega_2 \leq \pi,
\end{equation}
where $\Theta_{i,j}(L)$ is the $(i,j)$ element of the lag polynomial $\Theta(L)$. $\mathit{VD}_i(\omega_1,\omega_2)$ is the percentage reduction in the variance of $y_{i,t}$ -- after passing the data through a bandpass filter that retains only cyclical frequencies $[\omega_1,\omega_2]$ -- caused by entirely ``shutting off'' the shock of interest $\varepsilon_{1,t}$. The software package Dynare automatically reports $\mathit{VD}_i(0,\pi)$ after solving a DSGE model.

\paragraph{Identification and estimation: VD.}
Identification of the VD is completely analogous to our analysis of the FVR. By definition,
\[\mathit{VD}_i(\omega_1,\omega_2) = \frac{1}{\alpha^2} \times \frac{\int_{\omega_1}^{\omega_2} |s_{y_i\tilde{z}}(\omega)|^2 \,d\omega}{ \int_{\omega_1}^{\omega_2} s_{y_i}(\omega) \,d\omega},\]
where $s_{y_i\tilde{z}}(\omega)=\alpha \Theta_{i,1}(e^{-i\omega})$ is the $i$-th element of $s_{y\tilde{z}}(\omega)$, cf. equation \eqref{eqn:vd}. Since the last fraction on the right-hand side is point-identified, our identified set for $\alpha^2$ immediately maps into an identified for the VD.

We estimate the bounds as
\begin{equation} \label{eqn:implement_vd}
\left[\frac{1}{\hat{\bar{\alpha}}^2} \times \frac{\int_{\omega_1}^{\omega_2} |\hat{s}_{y_i\tilde{z}}(\omega)|^2 \,d\omega}{ \int_{\omega_1}^{\omega_2} \hat{s}_{y_i}(\omega) \,d\omega}\;,\; \frac{1}{\hat{\underline{\alpha}}^2} \times \frac{\int_{\omega_1}^{\omega_2} |\hat{s}_{y_i\tilde{z}}(\omega)|^2 \,d\omega}{ \int_{\omega_1}^{\omega_2} \hat{s}_{y_i}(\omega) \,d\omega}\right].
\end{equation}
The integrals are computed numerically. The spectral densities required to compute \eqref{eqn:implement_vd} are functions of the estimated reduced-form VAR parameters (see \cref{sec:inference_formulas}). Specifically,
\begin{eqnarray*}
\hat{s}_y(\omega) &=& \frac{1}{2\pi} \hat{B}_y(e^{-i\omega}) \hat{B}_y(e^{-i\omega})^*, \\
\hat{s}_{y\tilde{z}}(\omega) &=& \frac{1}{2\pi} \sum_{\ell = 0}^{\infty} \hat{\Sigma}_{y,\tilde{z},\ell} e^{-i\omega \ell},
\end{eqnarray*}
with
\begin{equation*}
\hat{B}_y(e^{-i\omega}) \equiv \sum_{\ell=0}^\infty \hat{B}_{y,\ell}e^{-i\omega\ell}, \quad \hat{\Sigma}_{y,\tilde{z},\ell} \equiv \widehat{\cov}(y_t,\tilde{z}_{t-\ell}) = \hat{B}_{y,\ell} \hat{B}_{\tilde{z}}'.
\end{equation*}
In practice, we truncate the infinite sums at a large lag.

\paragraph{Identification and estimation: FVD.}
Bounding the FVD requires more work. Intuitively, the reason that identification of the FVD is more challenging than for the FVR is that, even if we knew $\alpha$, the IV $z_t$ provides no information about the other structural shocks $\varepsilon_{j,t}$, $j \neq 1$. This matters because the definition \eqref{eqn:fvd} of the FVD, unlike that of the FVR, conditions on knowing all past shocks, rather than all past macro observables. \cref{thm:identif_fvd} formally characterizes the resulting identified set.

\begin{prop} \label{thm:identif_fvd}
Let there be given a joint spectral density for $w_t = (y_t',\tilde{z}_t)'$ satisfying the assumptions in \cref{thm:identif_alpha}. Given knowledge of $\alpha \in (\alpha_{LB}, \alpha_{UB}]$, the largest possible value of the forecast variance decomposition $\mathit{FVD}_{i,\ell}$ is $1$ (the trivial bound), while the smallest possible value is given by
\begin{equation} \label{eqn:FVD_lb}
\frac{\sum_{m=0}^{\ell-1} \cov(y_{i,t},\tilde{z}_{t-m})^2}{\sum_{m=0}^{\ell-1} \cov(y_{i,t},\tilde{z}_{t-m})^2 + \alpha^2\var(\tilde{y}_{i,t+\ell}^{(\alpha)} \mid \lbrace \tilde{y}_{\tau}^{(\alpha)} \rbrace_{-\infty<\tau\leq t})}.
\end{equation}
Here $\tilde{y}_t^{(\alpha)}=(\tilde{y}_{1,t}^{(\alpha)},\dots,\tilde{y}_{n_y,t}^{(\alpha)})'$ denotes a stationary Gaussian time series with spectral density $s_{\tilde{y}^{(\alpha)}}(\omega) = s_y(\omega) - \frac{2\pi}{\alpha^2}s_{y\tilde{z}}(\omega)s_{y\tilde{z}}(\omega)^*$, $\omega \in [0,2\pi]$. Expression \eqref{eqn:FVD_lb} is monotonically decreasing in $\alpha$, so the overall lower bound on $\mathit{FVD}_{i,\ell}$ is attained by $\alpha = \alpha_{UB}$; in this boundary case we can represent $\tilde{y}_t^{(\alpha_{UB})}= y_t - E(y_t \mid \lbrace \tilde{z}_\tau \rbrace_{-\infty<\tau\leq t})$.
\end{prop}

The upper bound on the FVD always equals the trivial bound of 1, for any $\ell \geq 1$. This upper bound is achieved by a model in which all shocks, except the first one, only affect $y_t$ after an $\ell$-period delay. The lower bound in contrast is nontrivial and informative. The argument is as follows: Even if $\alpha$ is known, the denominator $\var(y_{i,t+\ell} \mid \lbrace \varepsilon_\tau \rbrace_{-\infty<\tau\leq t})$ of the FVD is not identified due to the lack of information about shocks other than the first. Although we can upper-bound this conditional variance by the denominator of the FVR, this upper bound is not sharp. Instead, to maximize the denominator, as much forecasting noise as possible should be of the pure forecasting variety, and not related to noninvertibility. For all shocks except for $\varepsilon_{1,t}$, this is achievable through a Wold decomposition construction \citep[Thm. $2^{\prime \prime}$, p. 158]{Hannan1970}. Given $\alpha$, we know the contribution of the first shock to $y_t$; the residual after removing this contribution has the distribution of $\tilde{y}_t^{(\alpha)}$, as defined in the proposition. If $\alpha$ is not known, the smallest possible value of the lower bound \eqref{eqn:FVD_lb} is attained at the largest possible value of $\alpha$, namely $\alpha_{UB}$, for which $\varepsilon_{1,t}$ contributes the least to forecasts of $y_t$.

We estimate the bounds in \cref{thm:identif_fvd} as
\begin{equation} \label{eqn:implement_fvd}
\left[\frac{\sum_{m=0}^{\ell-1} \widehat{\cov}(y_{i,t},\tilde{z}_{t-m})^2}{\sum_{m=0}^{\ell-1} \widehat{\cov}(y_{i,t},\tilde{z}_{t-m})^2 + \hat{\bar{\alpha}}^2\widehat{\var}(\tilde{y}_{i,t+\ell}^{(\bar{\alpha})} \mid \lbrace \tilde{y}_{\tau}^{(\bar{\alpha})} \rbrace_{-\infty<\tau\leq t})}\;,\; 1\right].
\end{equation}
To approximate the conditional variance in the denominator, we proceed as in \cref{sec:inference_formulas}. First, we replace the infinite conditioning set with the finite set $\lbrace \tilde{y}_{\tau}^{(\bar{\alpha})} \rbrace_{\tau-M \leq \tau \leq t}$. Second, we compute the conditional variance using the standard projection formula, where the autocovariances of the process $\lbrace \tilde{y}_{\tau}^{(\bar{\alpha})} \rbrace$ are estimated as
\[\widehat{\cov}(\tilde{y}_{t+\ell}^{(\bar{\alpha})},\tilde{y}_{t}^{(\bar{\alpha})}) = \widehat{\cov}(y_{t+\ell},y_t) - \frac{1}{\hat{\bar{\alpha}}^2} \textstyle \sum_{m=0}^\infty \widehat{\cov}(y_t,\tilde{z}_{t-m-\ell})\widehat{\cov}(y_t,\tilde{z}_{t-m})'.\]
In practice, we truncate the infinite sum at a large lag.

\clearpage

\section{Multiple instruments correlated with one shock}
\label{sec:iv_multi_details}

Here we show that the multiple-IV model in \cref{asn:svma,asn:iv,asn:shocks} is testable, but if it is consistent with the data, then identification analysis can be reduced to the single-IV case.

Define the IV residual vector $\tilde{z}_t$ as in equation \eqref{eqn:z_tilde_multi}. The multiple-IV model in \cref{asn:svma,asn:iv} implies the following cross-spectrum between $y_t$ and $\tilde{z}_t$:
\begin{equation} \label{eqn:cross_spectrum_iv_multi}
s_{y\tilde{z}}(\omega) = \frac{\alpha}{2\pi}\Theta(e^{-i\omega})e_1 \lambda',\quad \omega \in [0,2\pi].
\end{equation}
Thus, the cross-spectrum has rank-1 factor structure: It equals a nonconstant column vector times a constant row vector. This testable property turns out to be exactly what characterizes the multiple-IV model.
\begin{prop} \label{thm:iv_multi}
Let a spectrum $s_w(\omega)$ for $w_t = (y_t',\tilde{z}_t')'$ be given, satisfying the assumptions of \cref{thm:identif_alpha}. There exists a model of the form in \cref{asn:svma,asn:iv} which generates the spectrum $s_w(\omega)$ if and only if there exist $n_y$-dimensional real vectors $\zeta_\ell$, $\ell \geq 0$, and an $n_z$-dimensional constant real vector $\eta$ of unit length such that
\begin{equation} \label{eqn:factor}
s_{y\tilde{z}}(\omega) = \zeta(e^{-i\omega})\eta',\quad \omega \in [0,2\pi],
\end{equation}
where $\zeta(L)=\sum_{\ell=0}^\infty \zeta_\ell L^\ell$.
\end{prop}
Assuming henceforth that the factor structure obtains, we now show that identification in the multiple-IV model reduces to the single-IV case. It is convenient first to reparametrize the model slightly, by setting $\Sigma_v = \Sigma_{\tilde{z}} - \alpha^2\lambda\lambda'$ and treating $\Sigma_{\tilde{z}}$ as a basic model parameter instead of $\Sigma_v$. We then impose the requirement that $\Sigma_{\tilde{z}} - \alpha^2\lambda\lambda'$ be positive semidefinite. Clearly, $\Sigma_{\tilde{z}}=\var(\tilde{z}_t)$ is point-identified. Next, note from \eqref{eqn:cross_spectrum_iv_multi} that $\lambda$ is point-identified and equal to the $\eta$ vector in equation \eqref{eqn:factor}. This is because any rank-1 factorization of a matrix is identified up to sign and scale, and we have normalized $\eta$ to have length 1. Let $\Xi$ be any $(n_z-1) \times n_z$ matrix such that $\Xi\Sigma_{\tilde{z}}^{-1/2}\lambda = 0$. Define the $n_z \times n_z$ matrix
\[Q \equiv \left( \begin{array}{c}
\frac{1}{\lambda' \Sigma_{\tilde{z}}^{-1}\lambda}\lambda' \Sigma_{\tilde{z}}^{-1} \\
\Xi \Sigma_{\tilde{z}}^{-1/2}
\end{array} \right).\]
Since $Q$ is point-identified (given a choice of $\Xi$), it is without loss of generality to perform identification analysis based on the linearly transformed IV residuals
\[Q\tilde{z}_t = \left( \begin{array}{c}
\alpha \\
0 \\
\vdots \\
0
\end{array} \right)\varepsilon_{1,t} + \tilde{v}_t,\quad \tilde{v}_t \sim N\left(0, \left( \begin{array}{cc}
\frac{1}{\lambda' \Sigma_{\tilde{z}}^{-1}\lambda} - \alpha^2 & 0 \\
0  & \Xi\Xi'
\end{array} \right) \right).\]
Notice, however, that $\alpha$ only enters into the equation for the first element of $Q\tilde{z}_t$, and the $(n_z-1)$ last elements of $Q\tilde{z}_t$ are independent of the first element (and independent of $y_t$ at all leads and lags). Hence, it is without loss of generality to limit attention to the first element of $Q\tilde{z}_t$ when performing identification analysis for $\Theta_{i,j,\ell}$ and $\alpha$. The first element of $Q\tilde{z}_t$ equals $\breve{z}_t$ as defined in equation \eqref{eqn:iv_lincom} in the main text.\footnote{The above display implies that we must have $\alpha^2 \leq (\lambda' \var(\tilde{z}_t)^{-1}\lambda)^{-1}$, which is precisely what the upper bound for $\alpha^2$ yields when applied to $\breve{z}_t$.}

Additional restrictions on the IVs can ensure point identification. In particular, if $n_z \geq 2$ and the researcher is willing to restrict $\Sigma_v$ to be diagonal, then $\alpha$ is point-identified from any off-diagonal element of $\var(\tilde{z}_t) = \Sigma_v + \alpha^2 \lambda \lambda'$, since $\lambda$ is point-identified.

\clearpage

\section{Instruments correlated with multiple shocks}
\label{sec:iv_multimulti}

In this section, we ask how much can be said about forecast variance ratios if the researcher is only willing to assume that the observed set of external instruments $z_t$ is correlated with at most $n_{\varepsilon_x}$ shocks, collected in the vector $\varepsilon_{x,t}$. Hence, in this section we do not impose the exclusion restriction that only the first shock $\varepsilon_{1,t}$ be correlated with the IV(s).

\paragraph{Extended model and FVR.}
Without loss of generality, suppose the $n_z$ IVs are correlated with the first $n_{\varepsilon_x}$ of the $n_\varepsilon$ shocks. Denote this sub-vector of shocks by $\varepsilon_{x,t}$. For now, $n_{\varepsilon_x}$ need not be known to the econometrician. We define the extended SVMA-IV model as
\begin{align}
y_t &= \Theta(L)\varepsilon_t,\quad \Theta(L) \equiv \sum_{\ell=0}^\infty \Theta_{\ell} L^\ell, \label{eqn:multimulti_svma} \\
z_t &= \sum_{\ell=1}^\infty (\Psi_\ell z_{t-\ell} + \Lambda_\ell y_{t-\ell} ) + \underbrace{\Gamma \varepsilon_{x,t} + \Sigma_v^{1/2} v_t}_{\tilde{z}_t}, \label{eqn:multimulti_iv}
\end{align}
where $\Gamma$ is $n_z \times n_{\varepsilon_x}$. We continue to impose i.i.d. normality of the shocks, cf. \cref{asn:shocks}. Our object of interest is the forecast variance ratio with respect to the $n_z$ particular linear combinations of shocks that enter into the IV equations, $\Gamma \varepsilon_{x,t}$:
\begin{align}
\mathit{FVR}_{i,\ell} &\equiv 1 - \frac{\var(y_{i,t+\ell} \mid \lbrace y_\tau \rbrace_{-\infty<\tau\leq t}, \lbrace \Gamma \varepsilon_{x,\tau} \rbrace_{t<\tau<\infty})}{\var(y_{i,t+\ell} \mid \lbrace y_\tau \rbrace_{-\infty<\tau\leq t})} \nonumber \\
&= \frac{\sum_{m=0}^{\ell - 1} \cov(y_{it}, \tilde{z}_{t-m}) (\Gamma \Gamma')^{-1} \cov(y_{it}, \tilde{z}_{t-m})'}{\var(y_{i,t+\ell} \mid \lbrace y_\tau \rbrace_{-\infty<\tau\leq t})}. \label{eqn:multimulti_fvr}
\end{align}
In the following we provide upper and lower bounds on this object. Given $\Gamma\Gamma'$, the FVR is point-identified, so we need to derive the identified set for $\Gamma\Gamma'$. At the end of this section we discuss how the FVR with respect to $\Gamma\varepsilon_{x,t}$ relates to other objects of interest.

Similar to \cref{sec:iv_multi_details}, the testable restriction of the model \eqref{eqn:multimulti_svma}--\eqref{eqn:multimulti_iv} is that the joint spectrum of $y_t$ and $\tilde{z}_t$ has a rank-$n_{\varepsilon_x}$ factor structure. If this assumption is not rejected, we can reduce the instrument vector to dimension $\min(n_z, n_{\varepsilon_x})$ without affecting the identification of $\mathit{FVR}_{i,\ell}$.\footnote{The argument is very similar to the one-shock case in \cref{sec:iv_multi_details} and is available upon request.} In particular, we may assume that $\Gamma$ has full row rank, which we do from now on, thus justifying the second equality in \eqref{eqn:multimulti_fvr}.

\paragraph{Identified set for $\Gamma$.}
Define $\Sigma_{\tilde{z}} \equiv \var(\tilde{z}_t)$. Proceeding similarly to the proof of \cref{thm:identif_alpha}, we can show that a given $\Gamma$ is consistent with the joint spectral density of the data if and only if $\Gamma \Gamma'$ has full row rank,
\begin{equation}
\Sigma_{\tilde{z}} - \Gamma \Gamma' \geq 0, \label{eqn:multimulti_IS1}
\end{equation}
and
\begin{equation}
\Gamma \Gamma' - 2 \pi s_{\tilde{z}^\dagger}(\omega) \geq 0, \quad \forall \;\omega \in [0,\pi], \label{eqn:multimulti_IS2}
\end{equation}
where $s_{\tilde{z}^\dagger}(\omega) = s_{y\tilde{z}}(\omega)^* s_y(\omega)^{-1} s_{y\tilde{z}}(\omega)$ and we use the notation $A \geq B$ if $A-B$ is Hermitian positive semi-definite (and similarly for $\leq$). Sharp bounds on $\mathit{FVR}_{i,\ell}$ thus follow from minimizing/maximizing \eqref{eqn:multimulti_fvr} over the space of $n_z \times n_z$ symmetric positive definite matrices $\Gamma\Gamma'$ subject to constraints \eqref{eqn:multimulti_IS1}--\eqref{eqn:multimulti_IS2}.

\paragraph{Lower bound on FVR.}
We now establish a sharp lower bound on the numerator in the definition \eqref{eqn:multimulti_fvr} of the FVR (the denominator is point-identified). Observe that
\begin{align*}
&\sum_{m=0}^{\ell - 1} \cov(y_{it}, \tilde{z}_{t-m}) (\Gamma \Gamma')^{-1} \cov(y_{it}, \tilde{z}_{t-m})' \\
&= \sum_{m=0}^{\ell - 1} \cov(y_{it}, \tilde{z}_{t-m}) \Sigma_{\tilde{z}}^{-1} \cov(y_{it}, \tilde{z}_{t-m})' + \sum_{m=0}^{\ell - 1} \cov(y_{it}, \tilde{z}_{t-m}) \lbrace (\Gamma\Gamma')^{-1}-\Sigma_{\tilde{z}}^{-1}\rbrace \cov(y_{it}, \tilde{z}_{t-m})' \\
&\geq \sum_{m=0}^{\ell - 1} \cov(y_{it}, \tilde{z}_{t-m}) \Sigma_{\tilde{z}}^{-1} \cov(y_{it}, \tilde{z}_{t-m})',
\end{align*}
where the inequality uses the constraint \eqref{eqn:multimulti_IS1}. The above lower bound is sharp: It is attained in a model where $\Sigma_v=0_{n_z \times n_z}$ and $\Gamma\Gamma'=\Sigma_{\tilde{z}}$, i.e., when all IVs are perfect.\footnote{Note that $\Gamma\Gamma'=\Sigma_{\tilde{z}}=2\pi s_{\tilde{z}}(\omega)$ satisfies constraint \eqref{eqn:multimulti_IS2} by the Schur complement formula and the positive semidefiniteness of the spectrum of $(y_t',\tilde{z}_t')'$.}

\paragraph{Upper bound on FVR.}
While we have not been able to derive a closed-form expression for the sharp upper bound on the FVR, it is straight-forward to numerically compute it. Let $\mathcal{S}_n$ denote the space of $n\times n$ real symmetric positive definite matrices, and let $\tr(A)$ denote the trace of a matrix $A$. The sharp upper bound on the numerator in the definition \eqref{eqn:multimulti_fvr} of the FVR is given by the value of the program
\begin{align}
&\max_{X \in \mathcal{S}_{n_z}} \tr(X C) + \tr(AC) \label{eqn:multimulti_semidefobj} \\
&\; \phantom{\text{s.t.}}\quad X \leq B(\omega), \quad \omega \in [0,\pi]. \nonumber
\end{align}
Here $X$ is a stand-in for $(\Gamma\Gamma')^{-1} - \Sigma_{\tilde{z}}^{-1}$, $C \equiv \sum_{m=0}^{\ell - 1} \cov(y_{it}, \tilde{z}_{t-m})'\cov(y_{it}, \tilde{z}_{t-m})$, $A \equiv \Sigma_{\tilde{z}}^{-1}$, and $B(\omega) \equiv \frac{1}{2\pi}s_{\tilde{z}^\dagger}(\omega)^{-1} - \Sigma_{\tilde{z}}^{-1}$. We can solve the above program to arbitrary accuracy by casting it as a (convex) semi-definite program with a finite number of constraints. Partition the interval $[0,\pi]$ into $N$ equal-length pieces, and consider the relaxed constraint set
\begin{equation}
X \leq \tilde{B}_m, \quad m \in \lbrace 1, 2, \dots, N\rbrace, \label{eqn:multimulti_semidef_c1}
\end{equation}
where $\tilde{B}_m \equiv \frac{N}{\pi} \times \int_{(m-1) \frac{\pi}{N}}^{m \frac{\pi}{N}} B(\omega)\, d\omega$. As $N \rightarrow \infty$, this constraint set approximates that of the original problem arbitrarily well, but for any finite $N$ the value of the discretized program provides an upper bound on the numerator in \eqref{eqn:multimulti_fvr}. Efficient numerical algorithms to compute the solution to semidefinite programs of the form \eqref{eqn:multimulti_semidefobj}--\eqref{eqn:multimulti_semidef_c1} are available in Matlab and other environments.\footnote{See for example \url{http://cvxr.com/cvx/doc/sdp.html}. To transform our constraints into ones involving real matrices, note that a Hermitian matrix with real part $A$ and imaginary part $B$ is positive semi-definite if and only if the real symmetric matrix $\left(\begin{smallmatrix} A & B' \\ B & A \end{smallmatrix}\right)$ is positive semi-definite.}

Alternatively, we can derive non-sharp upper bounds on the FVR numerator \eqref{eqn:multimulti_fvr} in closed form. For example, one conservative upper bound is obtained by maximizing $\tr(XC)+\tr(AC)$ subject to $X + \tilde{\Sigma}_{\tilde{z}}^{-1} \leq \left ( \int_{-\pi}^\pi s_{\tilde{z}^\dagger}(\omega)\,d\omega \right)^{-1} = \var(\tilde{z}_t^\dagger)^{-1}$. This yields the upper bound
\[\sum_{m=0}^{\ell - 1} \cov(y_{it}, \tilde{z}_{t-m}) \var(\tilde{z}_t^\dagger)^{-1} \cov(y_{it}, \tilde{z}_{t-m})',\]
which binds if the shocks $\varepsilon_{x,t}$ are all recoverable, but is otherwise not sharp. A less conservative -- but still generally suboptimal -- upper bound is given by
\[\tr(AC) + \sum_{m=1}^{n_z} \inf_{\omega \in [0,\pi]} \check{B}_{mm}(\omega),\]
where $\check{B}_{mm}(\omega)$ is the $(m,m)$ element of $\check{B}(\omega) \equiv C^{1/2\prime} B(\omega)C^{1/2}$, and $C=C^{1/2}C^{1/2\prime}$. This latter upper bound is sharp when $n_z=1$, in which case the lower and upper bounds in this section reduce to the FVR bound expressions derived in \cref{sec:identif}.

\paragraph{Interpretation.}
We highlight two special cases where the FVR with respect to $\Gamma\varepsilon_{x,t}$ (which we partially identified above) is of interest.

First, as in \citet{Mertens2013}, one may assume that the $n_z$ instruments are correlated with the same number $n_{\varepsilon_x} = n_z$ of structural shocks. In that case $\Gamma$ is square and nonsingular, so the FVR with respect to $\Gamma\varepsilon_{x,t}$ is the same as the FVR with respect to the shocks $\varepsilon_{x,t}$ themselves. Moreover, if we further assume that all included shocks $\varepsilon_{x,t}$ are recoverable, then $\tilde{z}_t^\dagger \equiv E(\tilde{z}_t \mid \lbrace y_\tau \rbrace_{-\infty<\tau<\infty}) = \Gamma \varepsilon_{x,t}$, so the historical decomposition of $y_t$ with respect to $\varepsilon_{x,t}$ is point-identified as $E(y_t \mid \lbrace \varepsilon_{x,t} \rbrace_{-\infty<\tau \leq t}) = E(y_t \mid \lbrace \tilde{z}_\tau^\dagger \rbrace_{-\infty<\tau \leq t})$.

Second, consider the case with a single IV but possibly several included shocks, $n_{\varepsilon_x}>1=n_z$. The above analysis shows that, even though the IV exclusion restrictions in the baseline model \eqref{eqn:iv} fail, the data are informative about the FVR with respect to the particular linear combination $\Gamma\varepsilon_{x,t}$ of shocks that enters the IV equation. The FVR with respect to this particular linear combination of shocks is evidently a lower bound for the FVR with respect to the full vector $\varepsilon_{x,t}$ of shocks that are correlated with the IV.

\clearpage

\section{Invertibility and SVAR-IV}
\label{sec:svar_iv}

In this section we characterize the bias of SVAR-IV methods when shocks may be noninvertible. Throughout we assume the validity of the SVMA-IV model in \cref{asn:svma,asn:iv,asn:shocks}. Our analysis builds on results by \cite{Lippi1994} and \citet{Forni2018}, who do not consider identification using external instruments.

The SVAR-IV (or ``proxy SVAR'') strategy identifies structural shocks by using the external IV to rotate the forecast errors from a reduced-form VAR \citep{Stock2008,Stock2012,Mertens2013,Gertler2015,Ramey2016}. For analytical clarity, we work with a VAR($\infty$) model with forecast errors $u_t \equiv y_t - E(y_t \mid \lbrace y_\tau \rbrace_{-\infty<\tau < t})$. Suppose the single residualized IV $\tilde{z}_t = \alpha\varepsilon_{1,t} + \sigma_v v_t$ is used by the econometrician. Under the SVAR-IV assumption of $n_{\varepsilon}=n_y$ and invertibility of all shocks, we would have $u_t = \Theta_0 \varepsilon_t$, with $\Theta_0$ square and nonsingular. Then the shock of interest would be identified as $\varepsilon_{1,t} = \gamma'u_t$, where $\gamma \equiv (\Sigma_{u\tilde{z}}'\Sigma_u^{-1}\Sigma_{u\tilde{z}})^{-1/2}\Sigma_u^{-1}\Sigma_{u\tilde{z}}$, $\Sigma_{u\tilde{z}} \equiv \cov(u_t,\tilde{z}_t)$, and $\Sigma_u \equiv \var(u_t)$.

We now ask what happens to the outputs of the SVAR-IV procedure if the invertibility assumption does not hold and $n_\varepsilon \geq n_y$.

\begin{prop}
\label{prop:proxy_SVAR}
Assume the SVMA-IV model in \cref{asn:svma,asn:iv,asn:shocks}. The shock that is (mis)identified by SVAR-IV is given by
\begin{equation} \label{eqn:svar_iv_shock}
\tilde{\varepsilon}_{1,t} \equiv \gamma'u_t = \sum_{j=1}^{n_\varepsilon} \sum_{\ell = 0}^\infty a_{j,\ell} \varepsilon_{j,t-\ell},
\end{equation}
where the scalar coefficients $\lbrace a_{j,\ell}\rbrace$ satisfy $\sum_{j=1}^{n_\varepsilon} \sum_{\ell = 0}^\infty a_{j,\ell}^2 = 1$ and $a_{1,0} = \sqrt{R_0^2}$. The associated SVAR-IV impulse responses are given by\footnote{In any SVAR($\infty$) model, the impulse responses implied by the model must equal the local projections of the outcomes on the identified shock(s). This follows from the Wold representation.}
\[\tilde{\Theta}_{\bullet, 1, \ell} \equiv \cov(y_t,\tilde{\varepsilon}_{1,t-\ell}) = \sum_{j=1}^{n_\varepsilon} \sum_{m=0}^\infty  a_{j,m} \Theta_{\bullet, 1, \ell+m},\quad \ell=0,1,2,\dots,\]
and the impact impulse responses satisfy
\[\tilde{\Theta}_{\bullet, 1, 0} = \frac{1}{\sqrt{R_0^2}} \Theta_{\bullet, 1, 0}.\]
\end{prop}
Under noninvertibility, SVAR-IV mis-identifies the shock as a distributed lag of all the shocks in the underlying model, with the coefficient on the true shock of interest $\varepsilon_{1,t}$ equal to $\sqrt{R_0^2}$ (the square root of the degree of invertibility, cf. \cref{sec:model}). This causes impulse responses to be conflated across horizons and shocks. At the impact horizon, SVAR-IV overstates the magnitudes of the true impulse responses $\Theta_{\bullet,1,0}$ (to a one standard deviation shock) by a factor of $1/\sqrt{R_0^2}$. Thus, the SVAR-IV-implied one-step-ahead forecast variance decompositions for the first shock overstate the true one-step-ahead FVRs (as defined in \cref{sec:model}) by a factor of $1/R_0^2$. The bias of SVAR-IV-implied \emph{multi-step} forecast variance decompositions depends in more complicated ways on the sequence of true impulse responses.

In summary, while SVAR-IV analysis solves the familiar ``rotation problem'' in SVAR analysis, it does not solve the invertibility problem. The issue is not that the IV selects a suboptimal linear combination $\gamma$ of the forecast residuals $u_t$ under noninvertibility, since it can be verified that $\gamma'u_t \propto E(\varepsilon_{1,t} \mid u_t)$ regardless of invertibility.\footnote{In particular, no other linear combination $\gamma$ can yield a representation \eqref{eqn:svar_iv_shock} where the weight $a_{1,0}$ exceeds $\sqrt{R_0^2}$ (subject to $\var(\tilde{\varepsilon}_{1,t})=1$). Thus, the IV handles the identification problem as well as possible subject to the constraints imposed by the (erroneous) invertibility assumption. As discussed in \cref{sec:model}, \emph{dynamic rotations} circumvent this issue by obtaining the shock $\tilde{\varepsilon}_{1,t}$ as a function of current \emph{and future} reduced-form residuals $\lbrace u_\tau \rbrace_{\tau \geq t}$. An argument similar to that in the proof of \cref{prop:proxy_SVAR} shows that, with such dynamic rotations, the weight on the true shock of interest is bounded above by $\sqrt{R_\infty^2}$. Dynamic rotations can thus solve the identification problem if and only if the shock of interest is recoverable.} Rather, SVAR methods fail because they assume that the time-$t$ forecast residuals suffice to recover $\varepsilon_{1,t}$ \citep{Lippi1994}. Only under invertibility (i.e., $R_0^2=1)$ do we have $a_{j,\ell}=0$ for all $(j,\ell) \neq (1,0)$, so that the SVAR-IV shock $\tilde{\varepsilon}_{1,t}$ equals the true shock $\varepsilon_{1,t}$. The higher the degree of invertibility $R_0^2$, the smaller is the extent of the SVAR-IV bias, as discussed by \citet{Sims2006}, \citet{Forni2018}, and \citet{Wolf2018}. An explicit illustration of SVAR-IV mis-identification is provided in \cref{sec:sw_illustration}.

\clearpage

\section{Illustration in a structural macro model}
\label{sec:sw_illustration}

In \cref{sec:illustration} we use several simple analytical examples to illustrate how our upper bound works. In this section we complement those simple examples with a quantitative exercise.

The nature of our exercise is as follows. We consider an econometrician observing (i) a small set of macroeconomic aggregates generated from the model of \citet{Smets2007} and (ii) noisy measures of some of the model's true underlying structural shocks (i.e., valid external instruments). For clarity, we abstract from any sampling uncertainty and assume that the econometrician observes an infinite amount of data, so the joint spectral density of observed macro aggregates and external IVs is perfectly known to her. Given this spectral density, she uses our bounds to draw conclusions about variance decompositions and the degree of invertibility, without exploiting the underlying structure of the model. Overall, the point of this exercise is to show that our conclusions on likely tightness of the upper bound are not an artifact of the particular stylized environments considered in \cref{sec:illustration}, but similarly obtain in quantitatively relevant, dynamic structural macro models, for exactly the same economic reasons.

\subsection{Preliminaries}

We employ the \citet{Smets2007} model. Throughout, we parametrize the model according to the posterior mode estimates of \citet{Smets2007}.\footnote{Our implementation of the Smets-Wouters model is based on Dynare replication code kindly provided by Johannes Pfeifer. The code is available at \url{https://sites.google.com/site/pfeiferecon/dynare}.} Following the canonical trivariate VAR in the empirical literature on monetary policy shock transmission, our baseline specification assumes the econometrician observes aggregate output, inflation, and the short-term policy interest rate; we consider additional observables below. These macro aggregates are all stationary in the model, so they should be viewed as deviations from trend. The model features seven unobserved shocks, so not all shocks can be invertible in the baseline specification.

The econometrician observes a single external instrument $z_t$ for the shock of interest $\varepsilon_{1,t}$:
\begin{equation*}
z_t = \alpha\varepsilon_{1,t} + \sigma_v v_t.
\end{equation*}
We normalize $\alpha = 1$ throughout and compute identified sets for two different degrees of informativeness of the external instrument, $\frac{1}{1+\sigma_v^2} \in \lbrace 0.25, 0.5\rbrace$. We do not attach any specific economic interpretation to the IV in the context of the model.

We separately consider three different shocks of interest: a monetary shock, a technology shock, and a forward guidance shock. The conventional monetary shock as well as the technology shock are already included in the original model of \citeauthor{Smets2007}. For the forward guidance shock, we depart from their model by assuming that monetary policy shocks are known two quarters in advance; that is, we change the model's Taylor rule to
\begin{equation*}
r_t = \rho_r r_{t-1} + (1-\rho_r) \times \left ( \phi_\pi \pi_t + \phi_y \hat{y}_t + \phi_{dy} (\hat{y}_t - \hat{y}_{t-1}) \right ) + \varepsilon_{t-2}^m,
\end{equation*}
where $r_t$ denotes the nominal interest rate, $\pi_t$ denotes the inflation, $\hat{y}_t$ is the output gap, and $\varepsilon_t$ is the monetary shock.\footnote{This is the notion of forward guidance discussed, for example, in \citet{DelNegro2012}.} Overall, these three shocks are chosen in line with our simple analytical illustrations in \cref{sec:illustration}, and identification will be subject to the same economic intuition as the small-scale examples discussed there.

We emphasize that our results in the remainder of this section should \emph{not} be taken to imply that conventional monetary shocks are robustly near-invertible, or that forward guidance and technology shocks are never invertible -- clearly, our statements are always conditional on a certain set of observables. Instead, the only purpose of this section is to document that the simple economic intuition of \cref{sec:illustration} still plays out in a quantitative macro model with rich dynamics. To further clarify this point, we finish this section by briefly discussing how our results change with alternative sets of observables.

\subsection{Information content of several observables}
\label{subsec:SW_MP}

We first consider identification of monetary policy shocks, i.e., shocks to the serially correlated disturbance in the model's Taylor rule.

The monetary shock is nearly invertible. In the model of \citeauthor{Smets2007}, monetary policy shocks are the only shock to contemporaneously move inflation and nominal rates in opposite directions \citep{Uhlig2005}. Given this unique conditional co-movement, the intuition offered in \cref{subsec:info_obs} suggests that the degree of invertibility should be high, and indeed it equals $R_0^2 = 0.8702$, in spite of the limited importance of monetary shocks for aggregate fluctuations in this model \citep{Wolf2018}. Looking forward in time does not sharpen identification much further ($R_\infty^2 = 0.8763$), and neither does looking across the spectrum frequency-by-frequency ($\alpha_{LB}^2 = 0.8947$).\footnote{Formally, the scaled spectral density $2\pi s_{\varepsilon_1^\dagger}(\cdot)$ of the best two-sided linear predictor of the monetary shock is nearly flat at around 0.9.}

\begin{figure}[tp]
\centering
{\textsc{Monetary shock: Identified set of FVRs}} \\
\vspace{0.5\baselineskip}
\includegraphics[width=\textwidth]{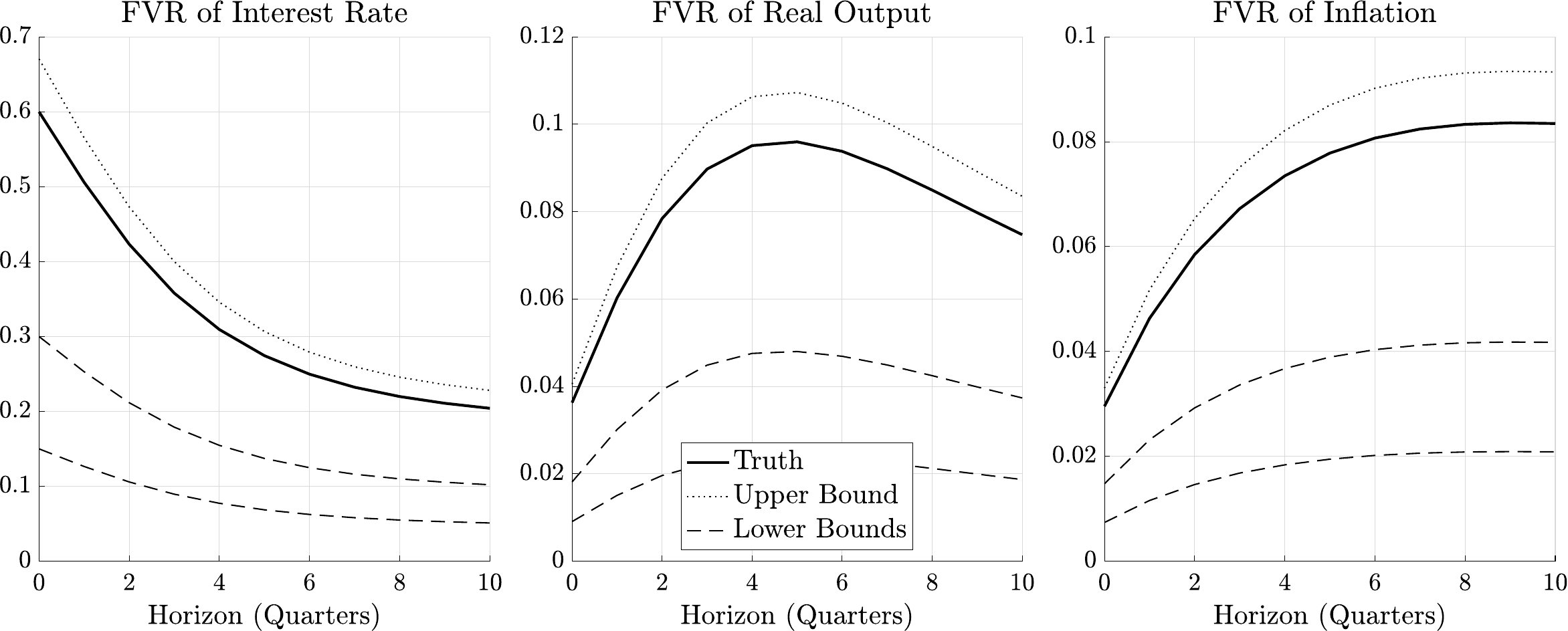}
\caption{Horizon-by-horizon identified sets for FVRs up to 10 quarters. The two lower bounds correspond to an IV with $\frac{1}{1+\sigma_v^2} = 0.25$ (lower dashed line) and an IV with $\frac{1}{1+\sigma_v^2} = 0.5$.}
\label{fig:SW_MPShock_FVR}
\end{figure}

\cref{fig:SW_MPShock_FVR} shows that the upper bounds on the forecast variance ratios are close to the true values. By construction, the upper and lower bounds are proportional to the true FVRs. The lower bound scales one-for-one with instrument informativeness, while the upper bound scales one-for-one with the maximal informativeness of the data for the shock across frequencies. Thus, near-invertibility immediately implies that the upper bounds are throughout close to the true FVR. In contrast, the informativeness of the lower bounds depends entirely on the strength of the IV.

\subsection{Dynamic information content}

Next, we consider the identification of technology shocks, i.e., innovations to the autoregressive process of total factor productivity. This type of shock illustrates how our sharp upper bound leverages information across frequencies.

In our baseline trivariate specification, the macro aggregates are informative about only the longest cycles of the technology shock. \cref{fig:SW_TechShock_SD} reports the spectral density of the best two-sided linear predictor of the technology shock. Strikingly, this spectral density is small at business-cycle frequencies, but close to 1 for long-run fluctuations, with a peak of $\alpha_{LB}^2 = 0.9084$. Intuitively, as in our simple example in \cref{subsec:info_obs_dyn}, technology shocks here account for most of the long-run fluctuations, and so macro aggregates are highly informative about the IV at low frequencies; as a result, our sharp upper bound on shock importance will again be tight. In contrast, averaging across frequencies, the shock is neither near-invertible ($R_0^2 = 0.1977$) nor near-recoverable ($R_\infty^2 = 0.2166$).\footnote{The issue is that, at short horizons, other shocks -- notably the price and wage mark-up shocks -- also push inflation and output in opposite directions. However, the technology shock becomes nearly invertible with a judicious choice of further observables; in particular, including either the level of TFP or hours worked leads to a nearly invertible representation.}

\begin{figure}[tp]
\centering
{\textsc{Technology shock: Spectral density of best 2-sided linear predictor}} \\ 
\vspace{0.5\baselineskip}
\includegraphics[width=0.6\textwidth]{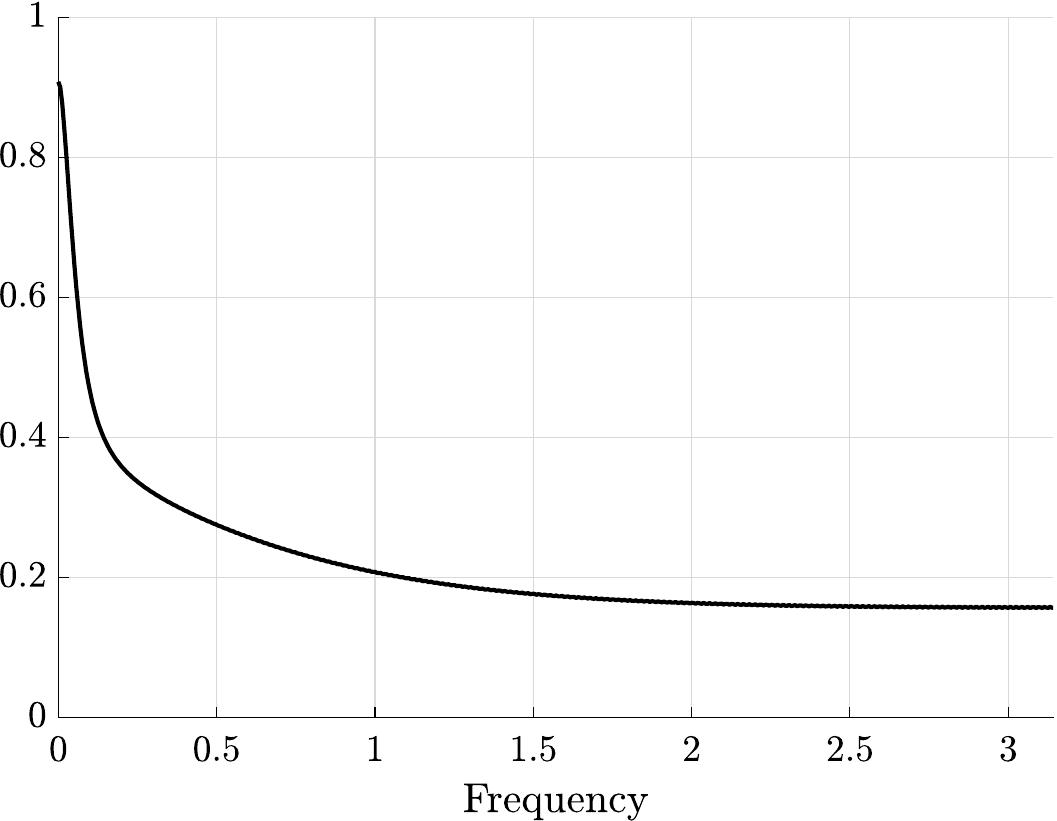}
\caption{Scaled spectral density $2\pi s_{\varepsilon_1^\dagger}(\cdot)$ of the best two-sided linear predictor of the technology shock. A frequency $\omega$ corresponds to a cycle of length $\frac{2\pi}{\omega}$ quarters.}
\label{fig:SW_TechShock_SD}
\end{figure}

Consequently, \cref{fig:SW_TechShock_FVR} shows that the sharp upper bounds on FVRs are tight. As always, the tightness of the lower bound is entirely governed by the strength of the IV.

\begin{figure}[tp]
\centering
{\textsc{Technology shock: Identified set of FVRs}} \\
\vspace{0.5\baselineskip}
\includegraphics[width=\textwidth]{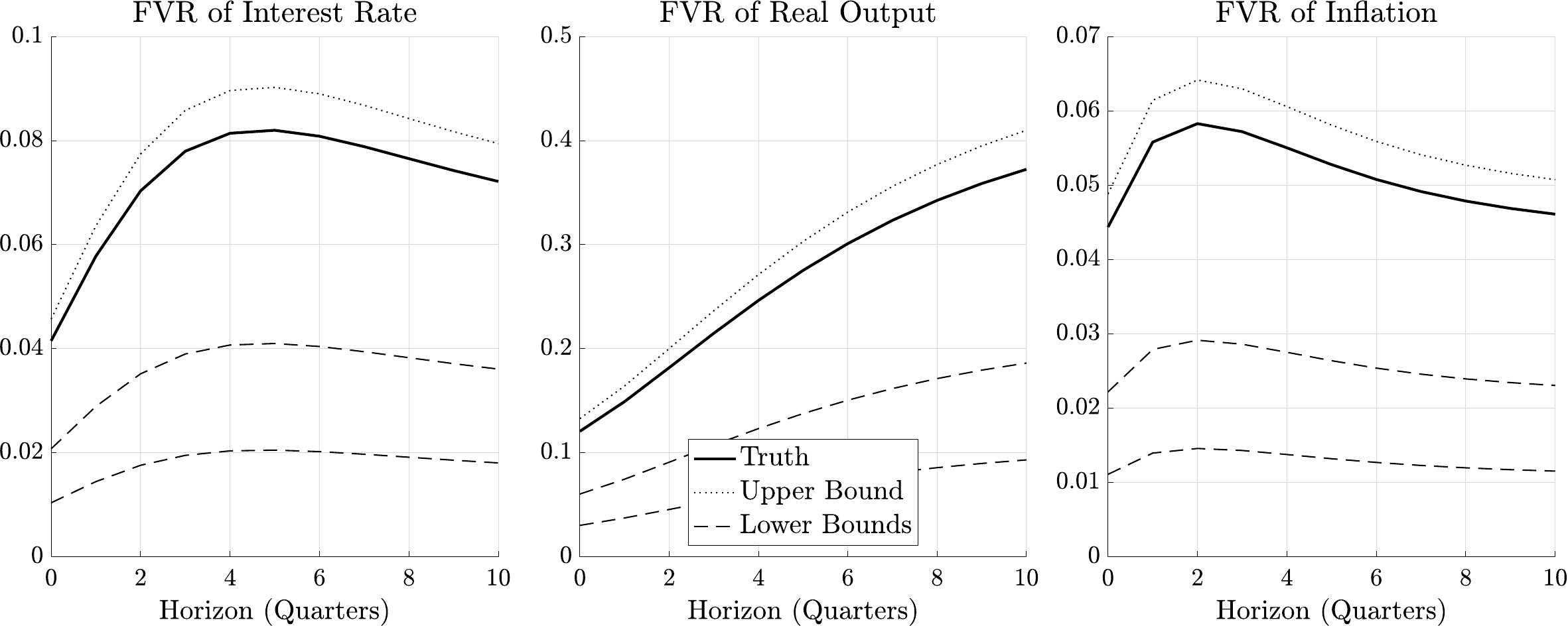}
\caption{Horizon-by-horizon identified sets for FVRs up to 10 quarters.  The two lower bounds correspond to an IV with $\frac{1}{1+\sigma_v^2} = 0.25$ (lower dashed line) and an IV with $\frac{1}{1+\sigma_v^2} = 0.5$.}
\label{fig:SW_TechShock_FVR}
\end{figure}

\subsection{Non-invertibility and news shocks}

For our third example, we modify the model to include forward guidance shocks, a type of news shock. As discussed above, a forward guidance shock is identical to a monetary shock, except that it is anticipated two quarters in advance by economic agents. This third example illustrates the robustness of our method to non-invertibility induced by news shocks.

The forward guidance shock is highly non-invertible, though it is nearly recoverable. \cref{fig:SW_FG_R2} shows the degree of invertibility $R_\ell^2$ up to time $t+\ell$ of the forward guidance shock. The figure considers horizons from $\ell = 0$ (the degree of invertibility) up to $\ell = 10$ (close to the degree of recoverability). Contemporaneous informativeness is limited, with $R_0^2 = 0.0768$; intuitively, when the forward guidance is announced, all macro aggregates move in the same direction, suggesting to the econometrician that the economy was probably buffeted by an ordinary demand shock.\footnote{This belief is further reinforced by the movement of nominal interest rates: Since output and inflation have both increased (for an expansionary forward guidance shock), nominal interest rates initially increase (as dictated by the Taylor rule), before declining two periods later.} At $\ell = 2$, however, the corresponding $R^2$ jumps to $R_2^2 = 0.8724$, reaching $R_\infty^2 = 0.8807$ as the overall degree of recoverability. The economic intuition is again simple: Two quarters from now, when the anticipated innovation finally directly hits the Taylor rule, the interest rate response suddenly switches sign, sending a strong signal that in fact a monetary policy shock -- and not some other kind of demand shock -- had occurred. By the same logic as in \cref{subsec:SW_MP}, the forward guidance shock is nearly recoverable.

\begin{figure}[tp]
\centering
{\textsc{Forward guidance shock: Degree of invertibility at time $t+\ell$}} \\ 
\vspace{0.5\baselineskip}
\includegraphics[width=0.7\textwidth]{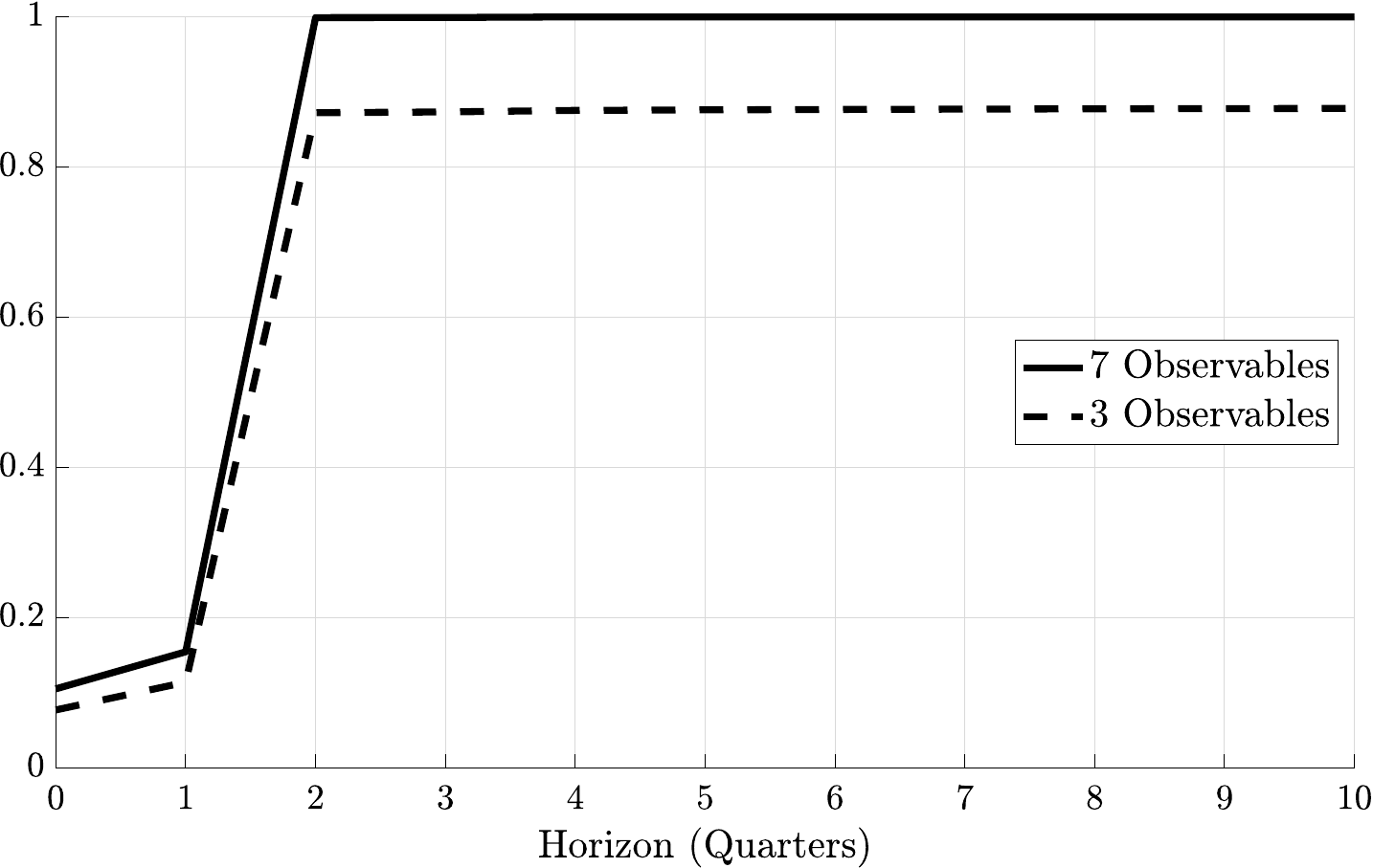}
\caption{Population $R_\ell^2$ for the forward guidance shock, with three observables (output, inflation, interest rate) and seven observables (the full set in \citealp{Smets2007}).}
\label{fig:SW_FG_R2}
\end{figure}

Because of non-invertibility, \cref{fig:SW_FG_SVARIV_FVR} shows that the conventional SVAR-IV approach dramatically overstates the forward guidance FVRs. In particular, as revealed by our analysis in \cref{sec:svar_iv}, the (impact) FVR is biased upward by a factor of $1/R_0^2 \approx 13$ (!).\footnote{Of course, for suitably chosen sets of observables (e.g., expectations about future interest rates), the non-invertibility problem would disappear \citep{Leeper2013}. Our method is robust in the sense that it does not require such a judicious choice of further observables -- it works even under extreme non-invertibility.}

\begin{figure}[tp]
\centering
{\textsc{Forward guidance shock: SVAR-IV FVRs}} \\ 
\vspace{0.5\baselineskip}
\includegraphics[width=\textwidth]{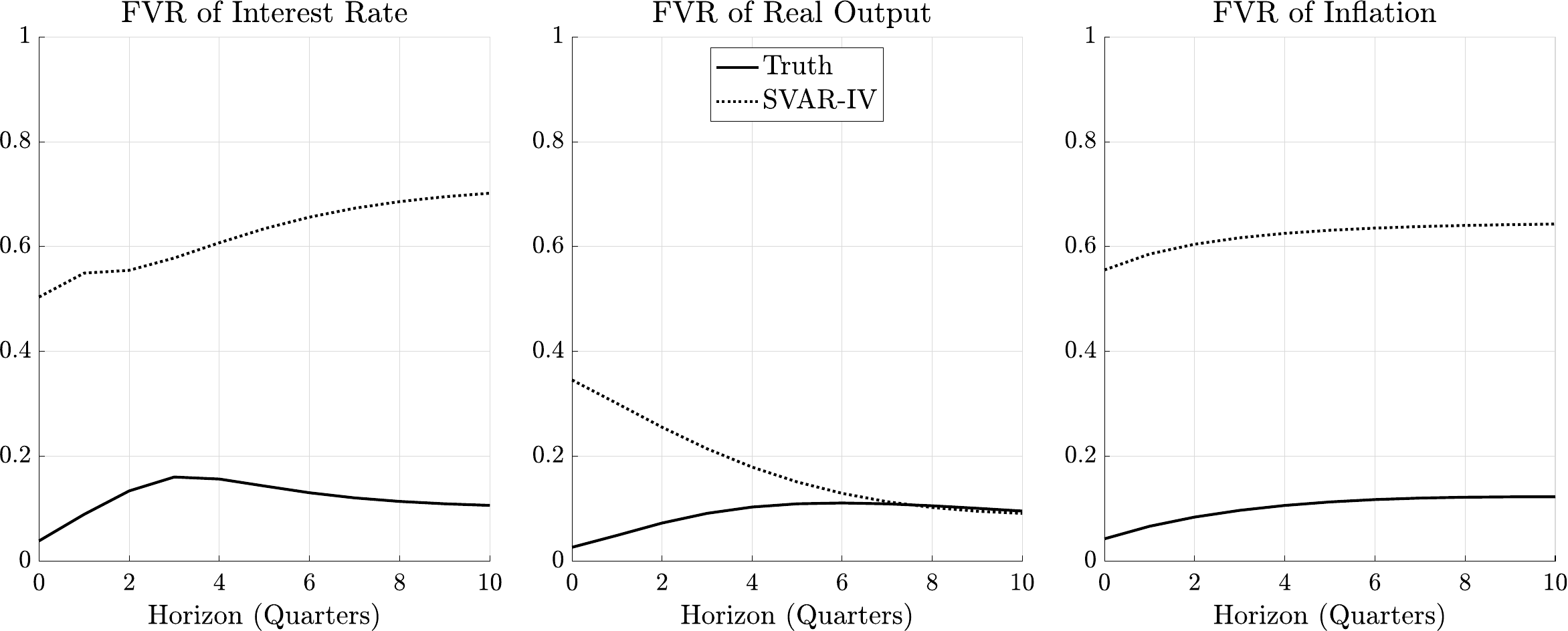}
\caption{FVRs for a forward guidance shock in the Smets-Wouters model, true values and SVAR-IV-estimated values (population limit). Baseline set of three observables.}
\label{fig:SW_FG_SVARIV_FVR}
\end{figure}

\begin{figure}[tp]
\centering
{\textsc{Forward guidance shock: Identified set of FVRs}} \\
\vspace{0.5\baselineskip}
\includegraphics[width=\textwidth]{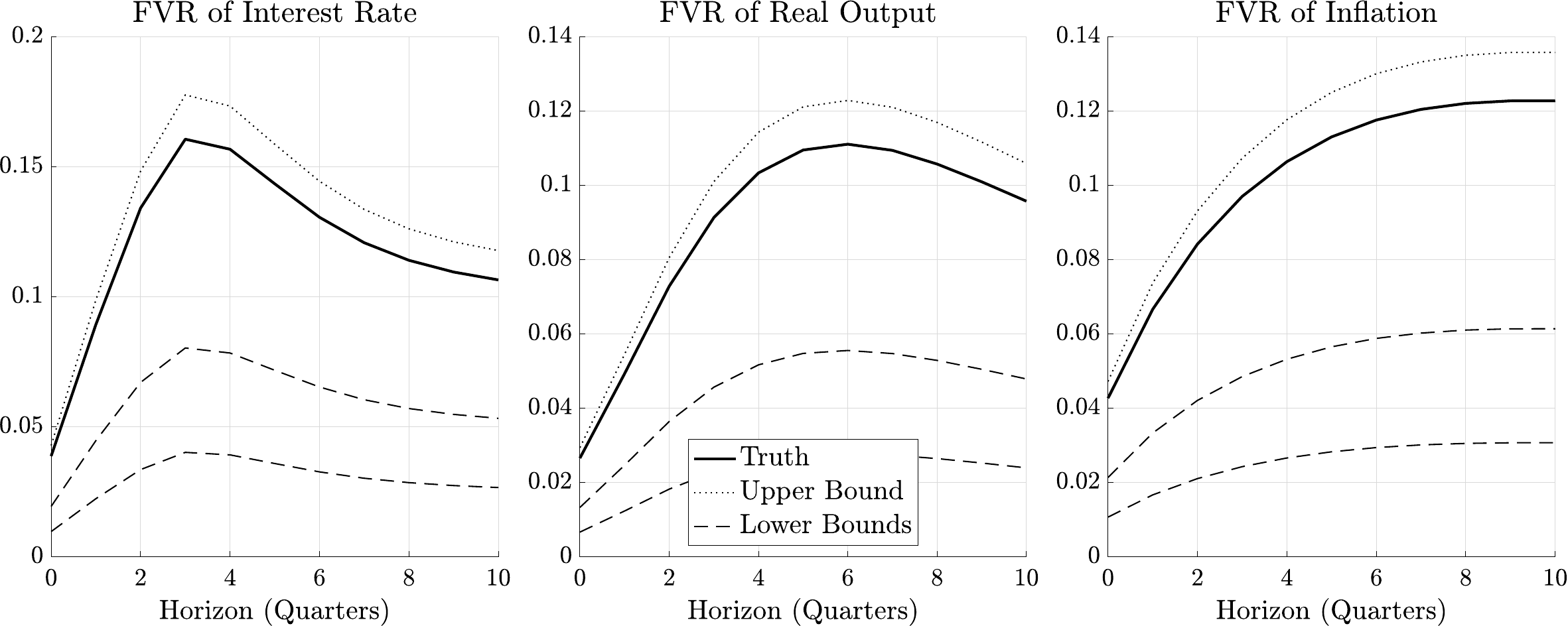}
\caption{Horizon-by-horizon identified sets for FVRs up to 10 quarters.  The two lower bounds correspond to an IV with $\frac{1}{1+\sigma_v^2} = 0.25$ (lower dashed line) and an IV with $\frac{1}{1+\sigma_v^2} = 0.5$.}
\label{fig:SW_FG_FVR}
\end{figure}

\cref{fig:SW_FG_FVR} shows that our method instead achieves a tight upper bound on the FVR, irrespective of the degree of invertibility. Since the forward guidance shock is nearly recoverable, the upper bounds of our identified sets for the different FVRs are again close to the truth, similar to the conventional (near-invertible) monetary shock studied in \cref{subsec:SW_MP}.

\subsection{Other observables}

The results in the preceding sections are designed to illustrate the economic logic of our method. They should not, however, be interpreted as offering generally valid conclusions on the invertibility (or lack thereof) of different structural shocks -- such statements are invariably sensitive to the choice of observables. To further emphasize this point, we in \cref{tab:sw_R2} compute the degrees of invertibility and recoverability for each shock, for different sets of macro observables.

The degrees of invertibility and recoverability are by definition increasing in the number of macroeconomic observables. For the baseline monetary shock, the degree of invertibility is high as soon as the researcher observes both nominal interest and inflation; with the full set of observables, the shock becomes perfectly invertible. For the technology shock, the degree of invertibility jumps to almost 1 as soon as hours worked become observable; intuitively, this is so because, at the posterior mode of the Smets-Wouters model, most \emph{high- and low-frequency} variation of hours worked is driven by technology shocks. Finally, because none of the observables included in the estimation exercise of \citeauthor{Smets2007} are forward-looking measures of nominal interest rates, the forward guidance shock remains highly non-invertible regardless of the choice of observables.

\begin{table}[tp]
\centering
\textsc{Structural illustration: Degree of invertibility/recoverability} \\[2ex]
\renewcommand{\arraystretch}{1.5}
\begin{tabular}{|l|C{1.4cm}C{1.4cm}|C{1.4cm}C{1.4cm}|C{1.4cm}C{1.4cm}|}
\hline 
 & \multicolumn{2}{c}{Monetary shock} & \multicolumn{2}{|c|}{Technology shock} & \multicolumn{2}{|c}{Forw. guid. shock} \\
\cline{2-7}
Macro observables & $R_0^2$ & $R_\infty^2$ & $R_0^2$ & $R_\infty^2$ & $R_0^2$ & $R_\infty^2$ \\
\hline
Baseline & 0.8702 & 0.8763 & 0.1977 & 0.2166 & 0.0768 & 0.8807 \\
+ investm. + consum. & 0.9415 & 0.9507 & 0.2128 & 0.2384 & 0.0980 & 0.9492 \\
+ hours & 0.9272 & 0.9286 & 0.9799 & 0.9816 & 0.0774 & 0.9331 \\
All observables & 1 & 1 & 1 & 1 & 0.1049 & 1 \\
\hline
\end{tabular}
\caption{Degree of invertibility $R_0^2$ and degree of recoverability $R_\infty^2$ in Smets-Wouters model, given three different sets of macro observables $y_t$. ``Baseline'' is the 3-variable specification with output, inflation, and short-term interest rate. The second and third rows add either (i) investment and consumption or (ii) hours to the baseline observables. The last row has the full set of observables considered in \citet{Smets2007}.}
\label{tab:sw_R2}
\end{table}

\clearpage

\section{Monetary policy shock application: further results}
\label{sec:empir_add}

Complementing the empirical results in \cref{sec:mp_application}, \cref{fig:gk_fvr_svar} shows the forecast variance ratios/decompositions estimated by a conventional SVAR-IV procedure, with bootstrap confidence intervals. The conclusions about the irrelevance of monetary shocks for output growth and inflation are even starker in this figure than in the main paper. Our bounds estimates in \cref{sec:mp_application} show that the irrelevance of monetary shocks is not merely an artifact of assuming invertibility, but is instead a robust implication of the empirical covariances of the macro aggregates $y_t$ with the monetary shock instrument $z_t$ of \citet{Gertler2015}.

\begin{figure}[tp]
\centering
\textsc{Empirical application: Forecast variance ratios, SVAR-IV} \\[2ex]
\includegraphics[width=0.48\linewidth]{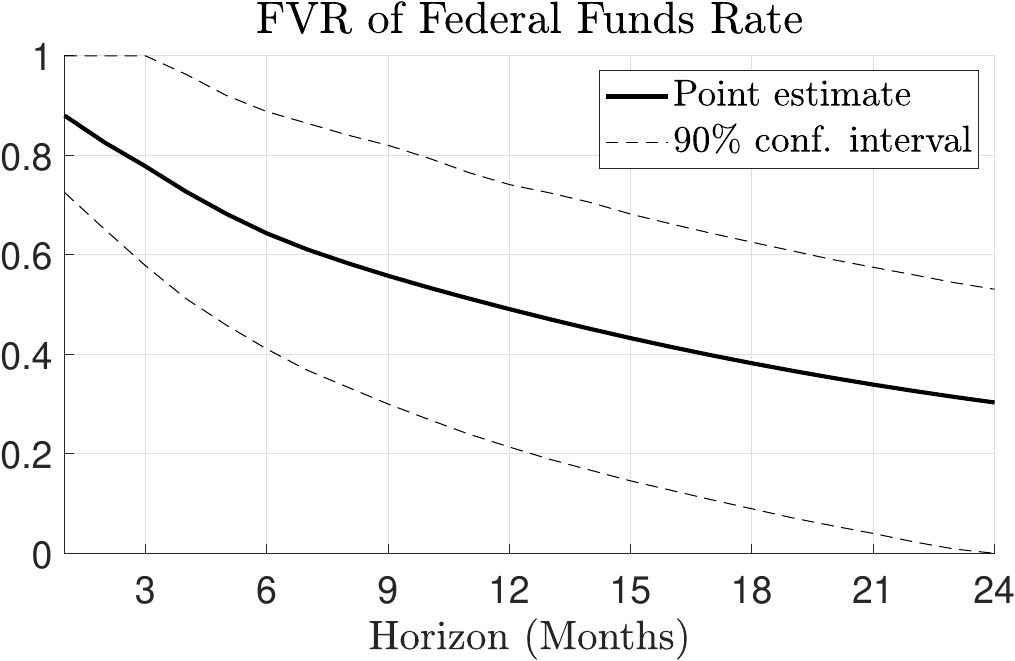} \includegraphics[width=0.48\linewidth]{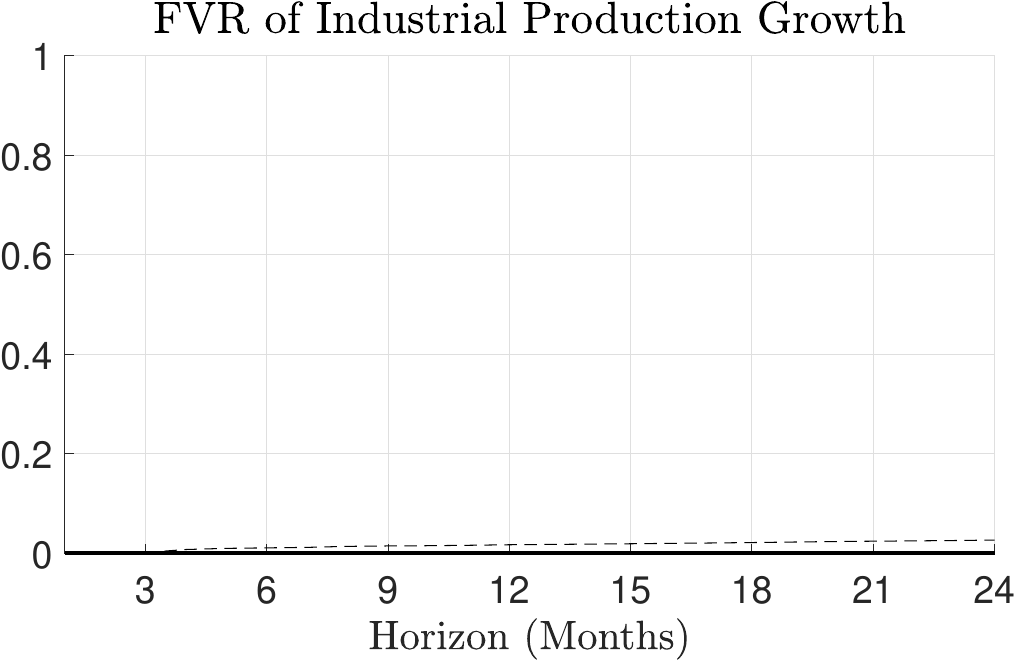} \\[2ex]
\includegraphics[width=0.48\linewidth]{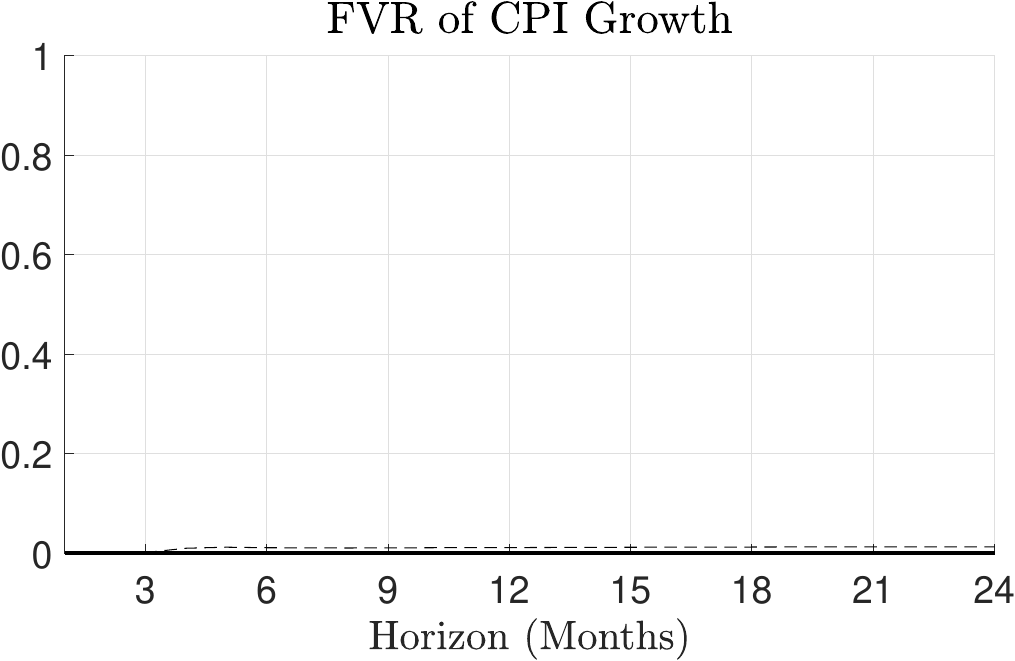} \includegraphics[width=0.48\linewidth]{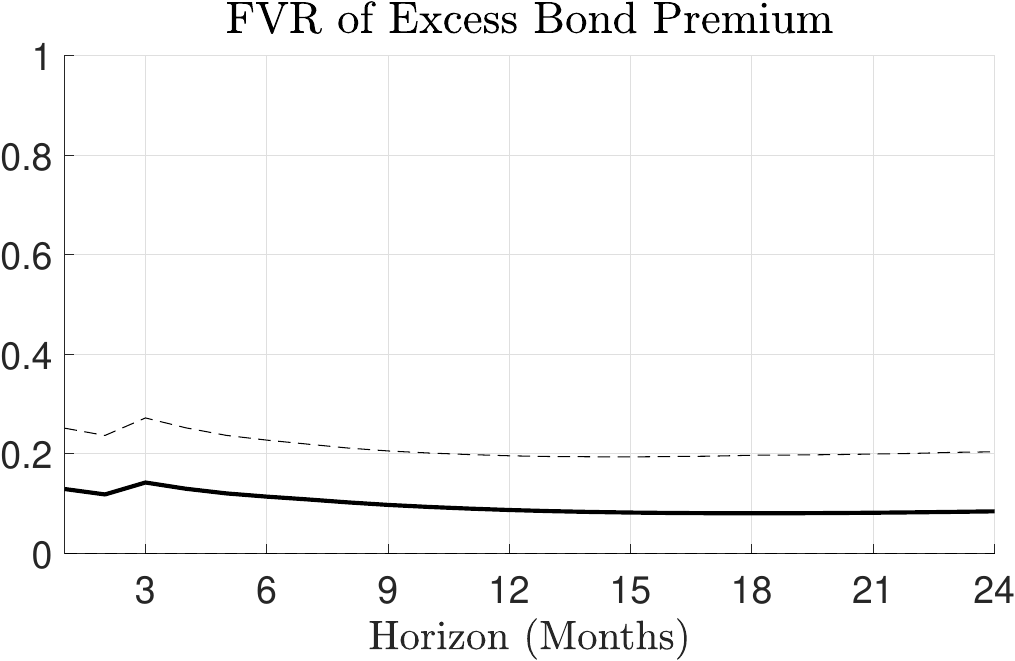}
\caption{Point estimates and 90\% confidence intervals for the identified sets of forecast variance ratios/decompositions estimated by SVAR-IV, across different variables and forecast horizons. For visual clarity, we force bias-corrected estimates/bounds to lie in $[0,1]$. The interest rate variable is the Federal Funds Rate.}
\label{fig:gk_fvr_svar}
\end{figure}

\clearpage

\section{Oil news shock application}
\label{sec:oil}

As a second application of our method, we study the importance of news about oil supply for aggregate business-cycle fluctuations. We use changes in oil supply expectations as an IV for oil news shocks and obtain two main results. First, we find that such oil supply news shocks are highly non-invertible, invalidating the standard SVAR-IV approach. Second, we show that our invertibility-robust method again yields sharp inference: Over the past three decades, oil news shocks have mattered little for U.S. inflation and played almost no role in U.S. output fluctuations.

\paragraph{Background.}
\citet{Kaenzig2021} constructs a measure of shocks to oil supply expectations from high-frequency changes in asset prices around press releases of the Organization of the Petroleum Exporting Countries (OPEC). Since news shocks are often non-invertible, as discussed in \cref{sec:illustration_news}, this setting is a second promising laboratory to showcase the appeal of our SVMA-IV approach.\footnote{We thank an anonymous referee for suggesting this application of our method.}

\paragraph{Model.}
We consider the same set of endogenous macro observables $y_t$ as in \citet[Section III.E]{Kaenzig2021}: log real oil prices, log oil production, log oil inventories, log world industrial production (IP), log broad U.S. nominal effective exchange rate index (NEER), log U.S. industrial production (IP), log U.S. consumer price index (CPI), federal funds rate (FFR), log VXO uncertainty index, and log U.S. terms of trade (TOT). Since our method requires stationary data, we transform the following variables to log growth rates (rather than log levels): oil production, oil inventories, world and U.S. industrial production, the nominal effective exchange rate, and U.S. CPI. Data are monthly from April 1983 to December 2017, as dictated by the availability of the IV. We include $p = 12$ lags and a constant in the reduced-form VAR (following \citeauthor{Kaenzig2021}), and use 1,000 bootstrap draws from a homoskedastic recursive residual VAR bootstrap.

We estimate the importance of oil shocks using both our invertibility-robust SVMA-IV framework as well as the conventional SVAR-IV approach. Note that, due to the differences in sample, data transformations, and bootstrap procedure, our SVAR-IV results are not directly comparable to those reported in \citet[][Table 2]{Kaenzig2021}.

\paragraph{Results.}
Our first main finding is that invertibility is strongly rejected: The 90\% confidence interval for the identified set of the degree of invertibility is $[0, 0.190]$, and the p-value for the Granger causality pre-test of invertibility in \cref{sec:how_to} is $0.0064$. We emphasize that invertibility is strongly rejected even though the VAR includes several forward-looking financial variables (e.g., oil prices, exchange rates, and the VIX). While these variables are known to respond quickly to oil supply news, such fast responses are apparently not sufficient to ensure invertibility; intuitively, financial variables are likely to respond quickly to many other nuisance shocks as well, which all things equal lowers the degree of invertibility of our shock of interest. We conclude that the analysis of oil supply news shocks calls for invertibility-robust approaches, such as our SVMA-IV procedure.

\cref{fig:kaenzig_svmaiv_fvr_1,fig:kaenzig_svmaiv_fvr_2} show our partial identification robust confidence intervals for the forecast variance ratio of the endogenous macro variables with respect to the oil news shocks. We report point estimates and confidence intervals for the identified sets at each horizon separately. As in \citet{Kaenzig2021}, we find that the data are consistent with the oil news shock explaining a sizable portion (but not the majority) of short-run oil price fluctuations. In spite of our weak identifying assumptions, we can furthermore conclude that -- on our post-1983 sample -- oil supply news have played a limited role in U.S. consumer price fluctuations, and have been essentially irrelevant for world and (in particular) U.S. real economic activity. Consistent with these results, oil supply news are estimated to have an at most moderate short-term effect on the U.S. terms of trade.

For comparison, \cref{fig:kaenzig_svariv_fvr_1,fig:kaenzig_svariv_fvr_2} show the analogous confidence intervals obtained from a conventional SVAR-IV procedure. The impact FVR estimates are substantially over-estimated relative to the invertibility-robust results, consistent with our finding of substantial non-invertibility and with \cref{prop:proxy_SVAR}. Our invertibility-robust results show that the following conclusions from the non-robust SVAR-IV analysis are spurious: (i) Oil supply news shocks are the overwhelming driver of oil price fluctuations in the short run; (ii) they are important drivers of medium-run fluctuations in U.S. inflation and real activity; and (iii) they are the dominant driver of the U.S. terms of trade at essentially all horizons.

\begin{figure}[p]
\centering
\textsc{Oil news shock: SVMA-IV FVRs, Global Variables} \\[4ex]
\includegraphics[width=0.48\linewidth]{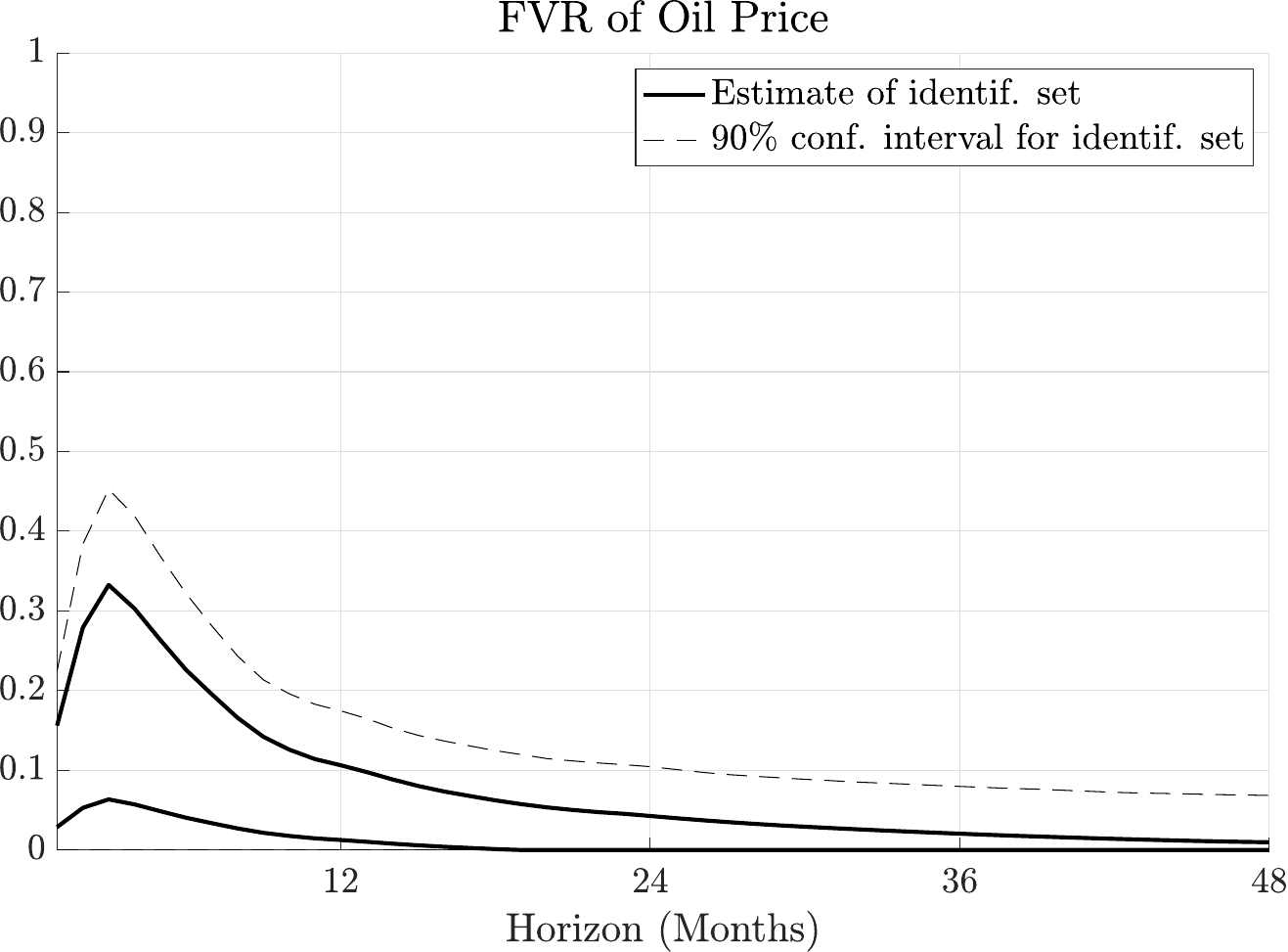} \includegraphics[width=0.48\linewidth]{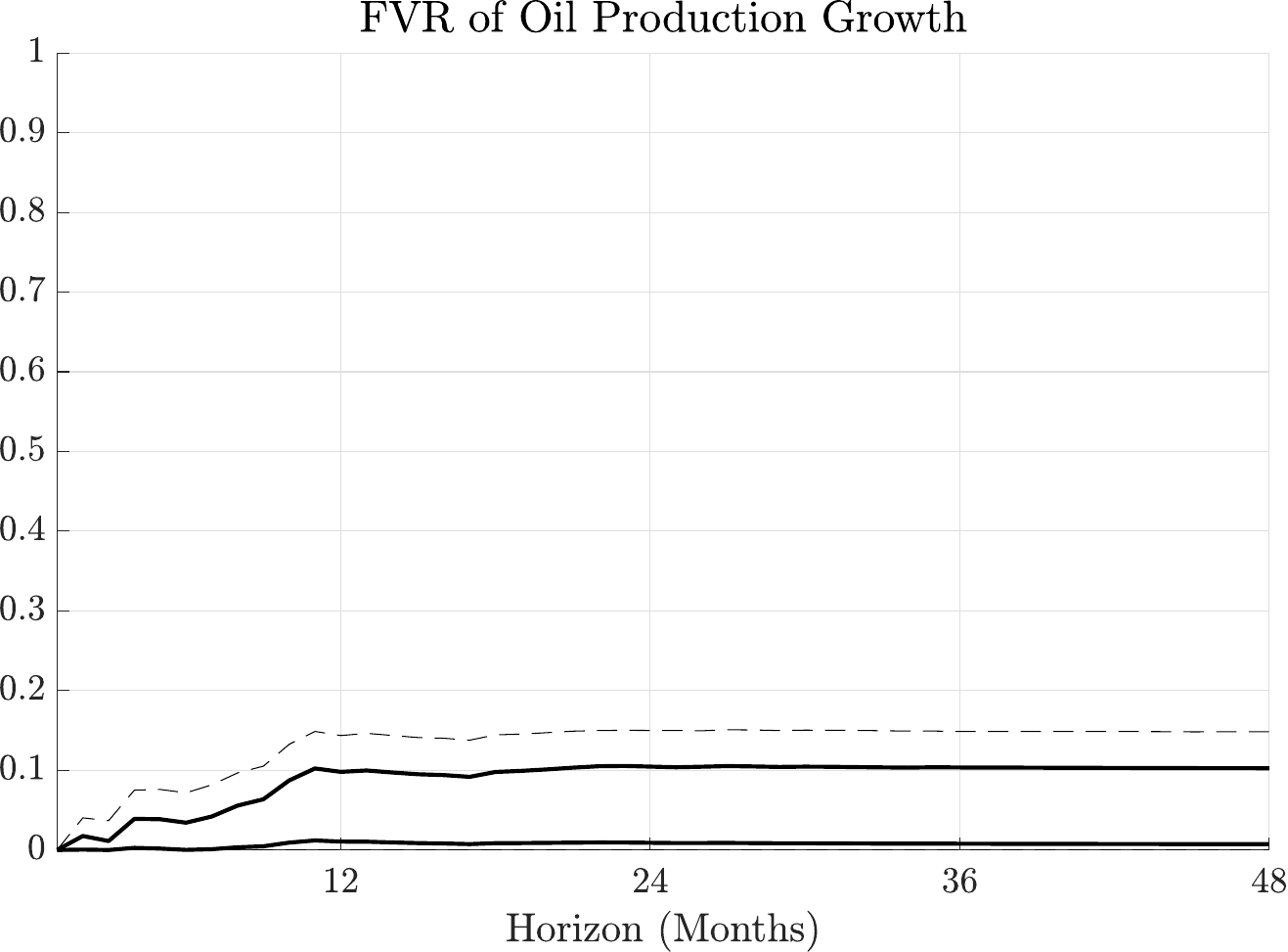} \\[2ex]
\includegraphics[width=0.48\linewidth]{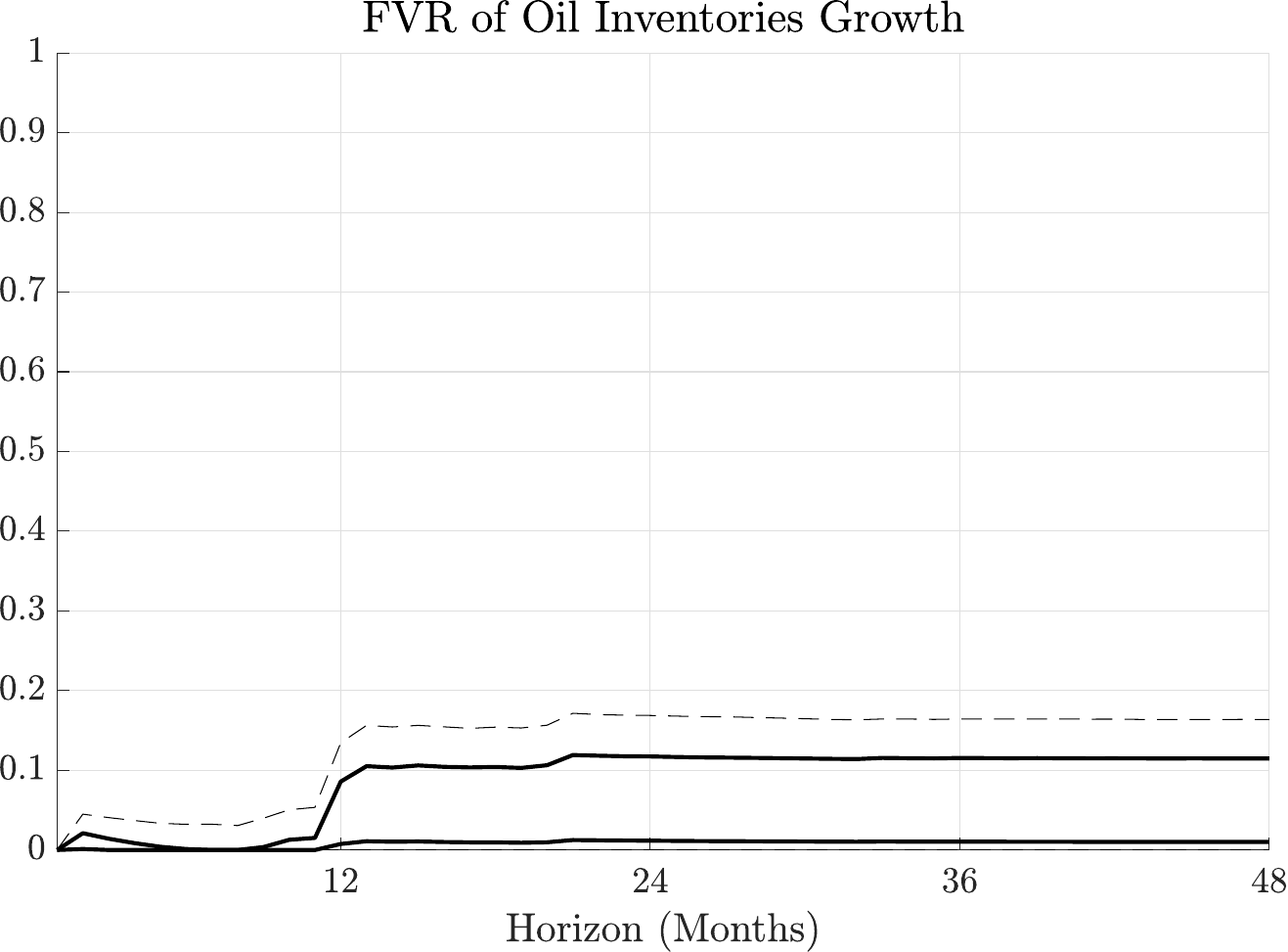} \includegraphics[width=0.48\linewidth]{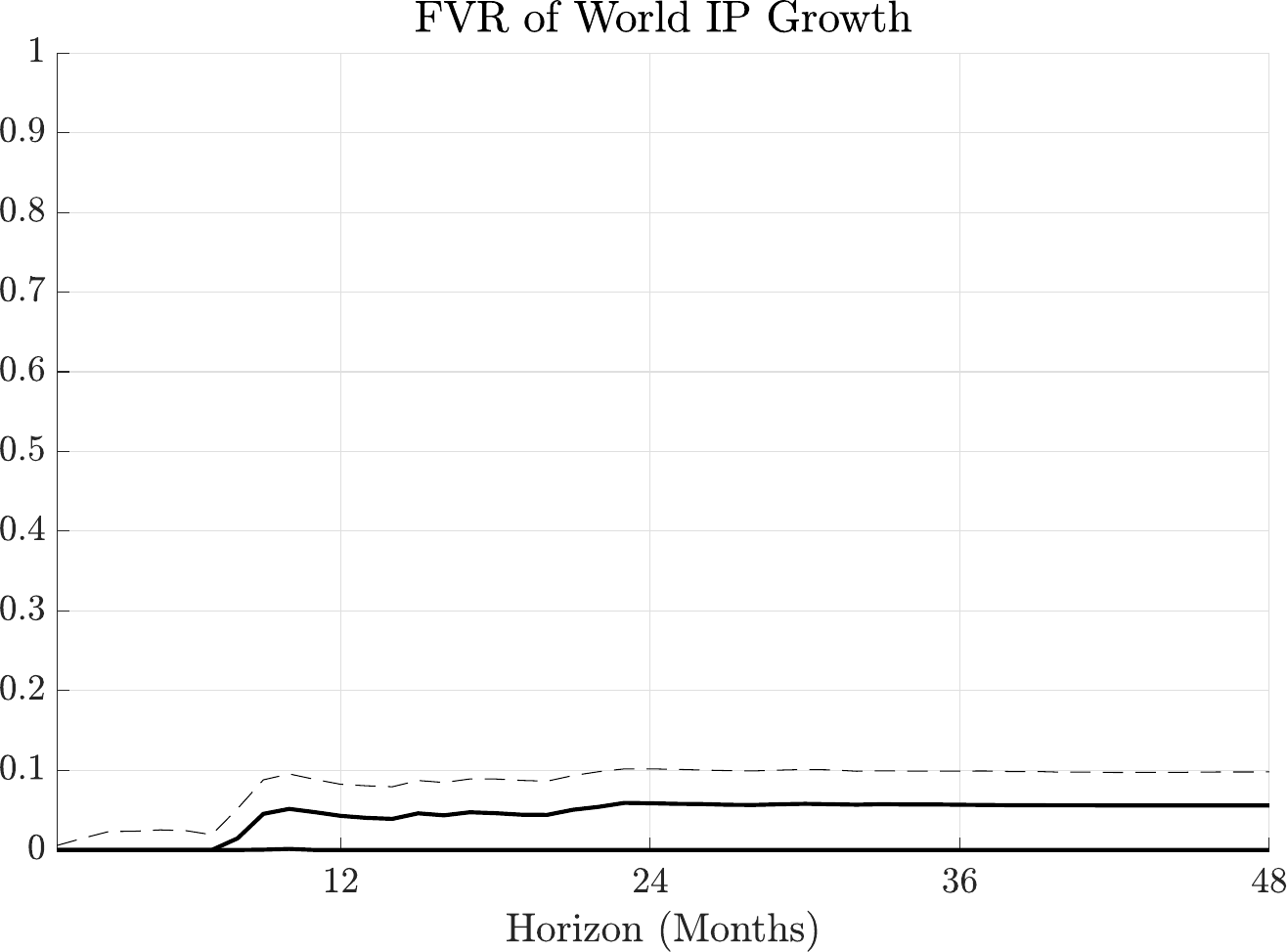} \\[2ex]
\caption{Point estimates and 90\% confidence intervals for the identified set of oil news shock FVRs, across different variables and forecast horizons, produced using our invertibility-robust SVMA-IV approach. For visual clarity, we force bias-corrected estimates/bounds to lie in $[0, 1]$.}
\label{fig:kaenzig_svmaiv_fvr_1}
\end{figure}

\begin{figure}[p]
\centering
\textsc{Oil news shock: SVMA-IV FVRs, U.S. Variables} \\[2ex]
\includegraphics[width=0.48\linewidth]{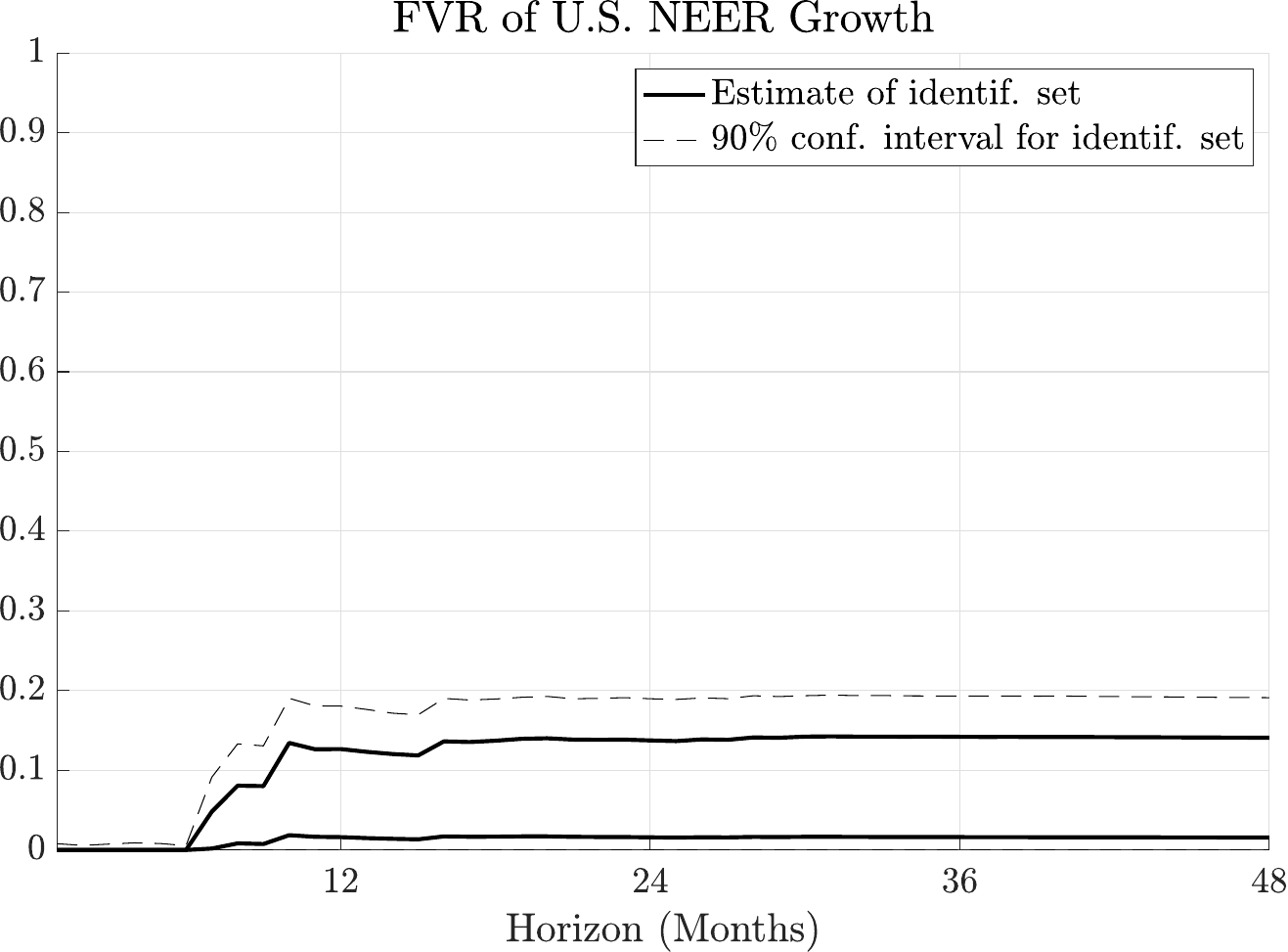} \includegraphics[width=0.48\linewidth]{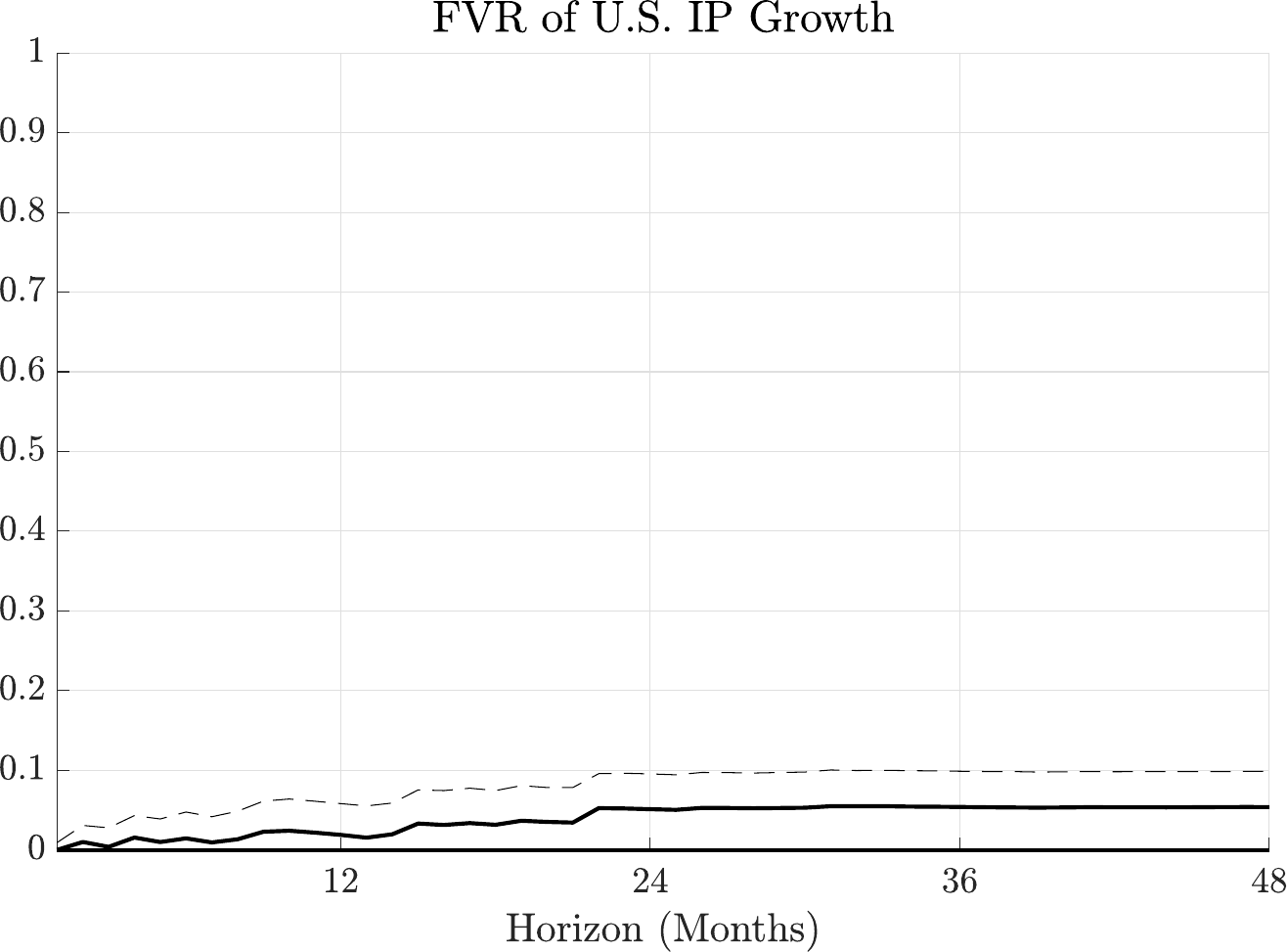} \\[2ex]
\includegraphics[width=0.48\linewidth]{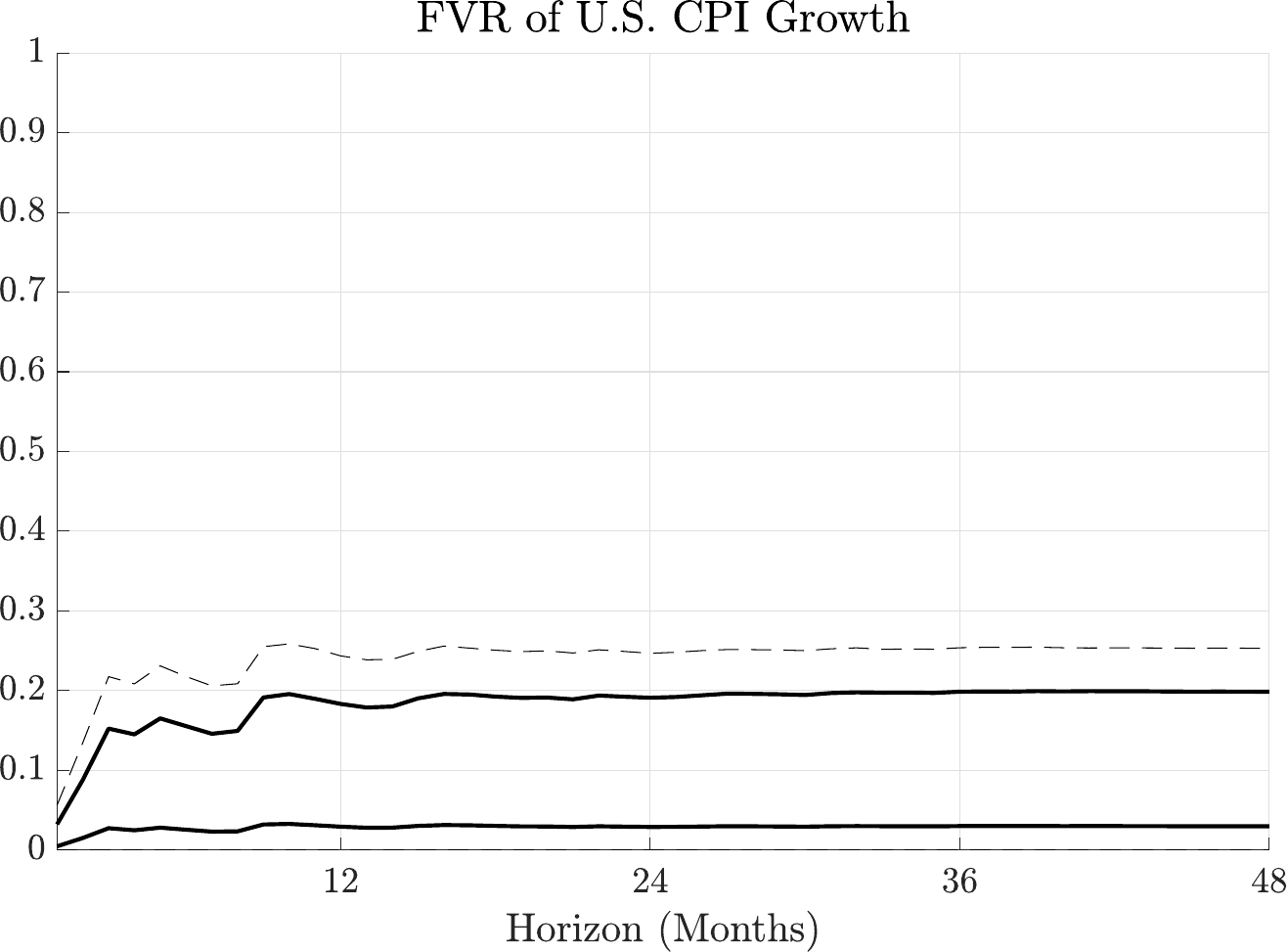} \includegraphics[width=0.48\linewidth]{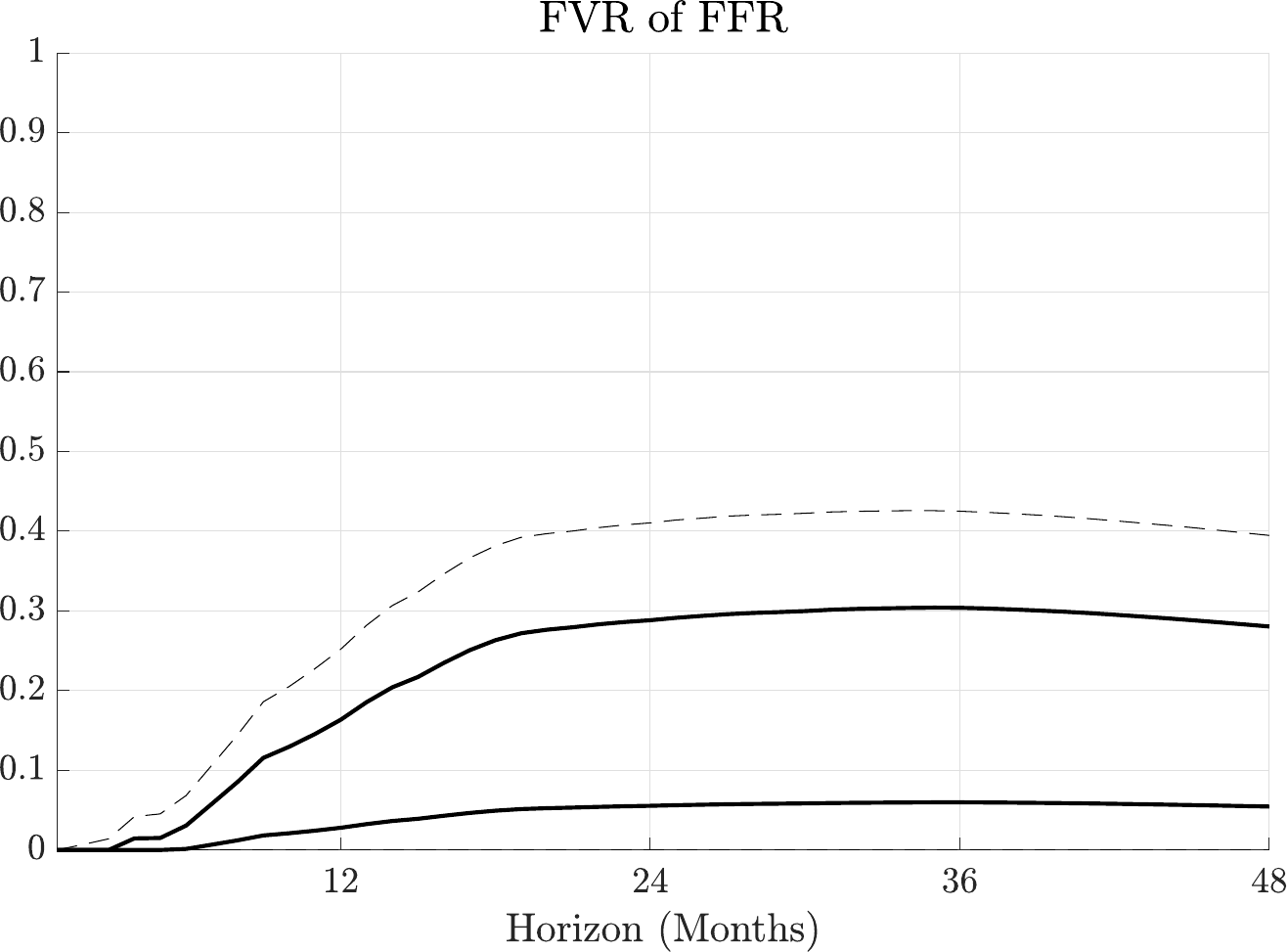} \\[2ex]
\includegraphics[width=0.48\linewidth]{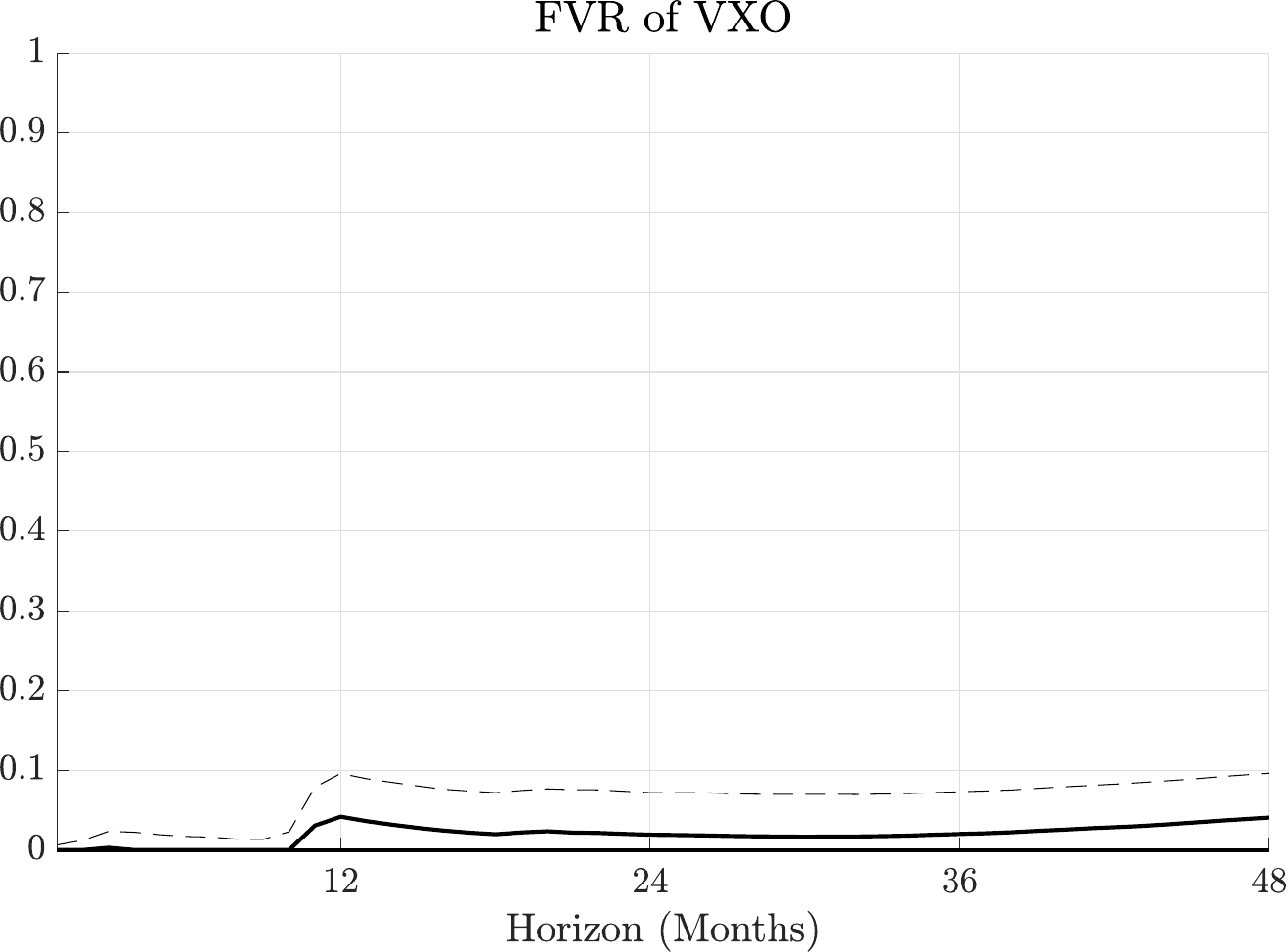} \includegraphics[width=0.48\linewidth]{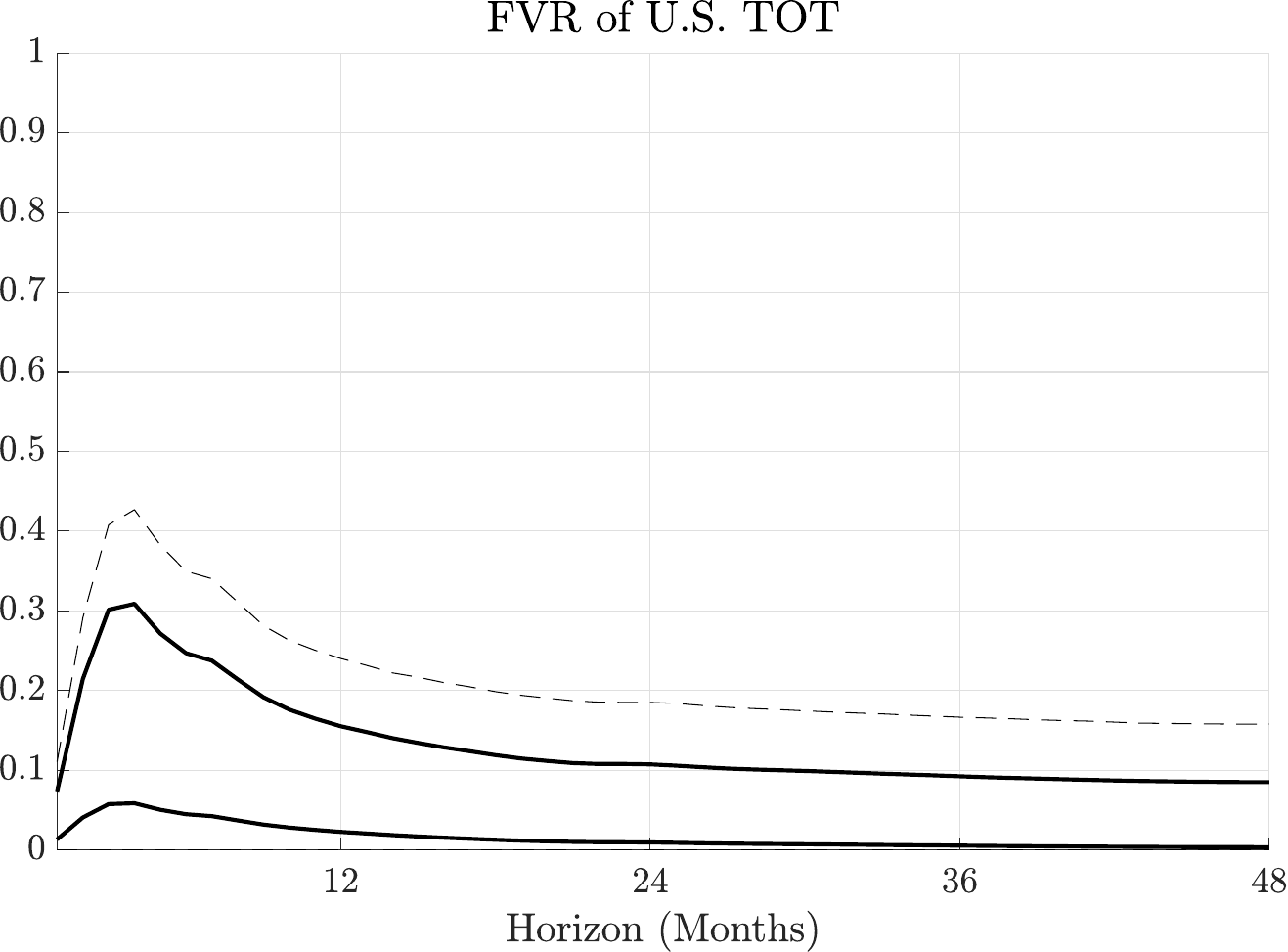} \\[2ex]
\caption{Point estimates and 90\% confidence intervals for the identified set of oil news shock FVRs, across different variables and forecast horizons, produced using our invertibility-robust SVMA-IV approach. For visual clarity, we force bias-corrected estimates/bounds to lie in $[0, 1]$.}
\label{fig:kaenzig_svmaiv_fvr_2}
\end{figure}

\begin{figure}[p]
\centering
\textsc{Oil news shock: SVAR-IV FVRs, Global Variables} \\[2ex]
\includegraphics[width=0.48\linewidth]{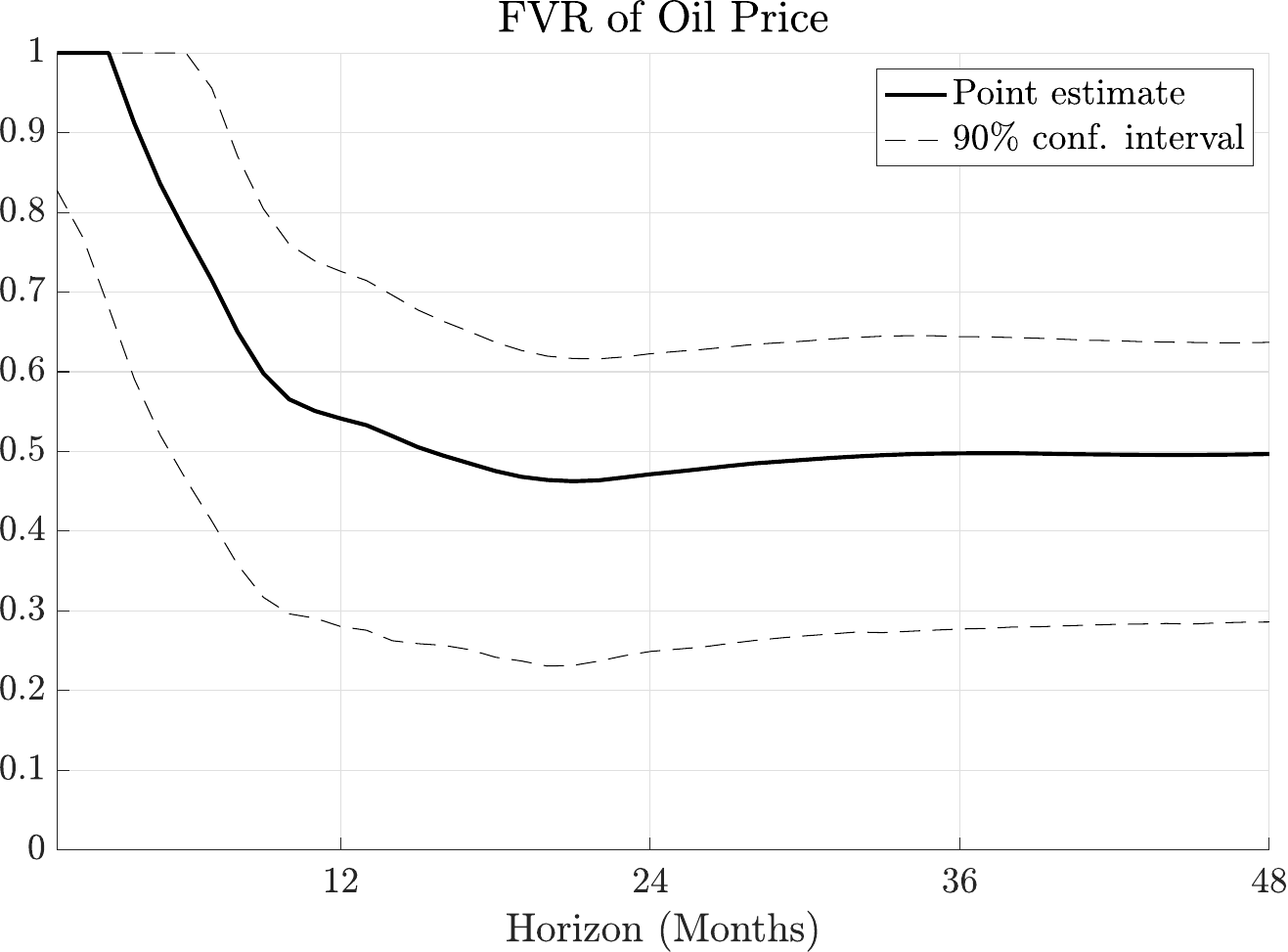} \includegraphics[width=0.48\linewidth]{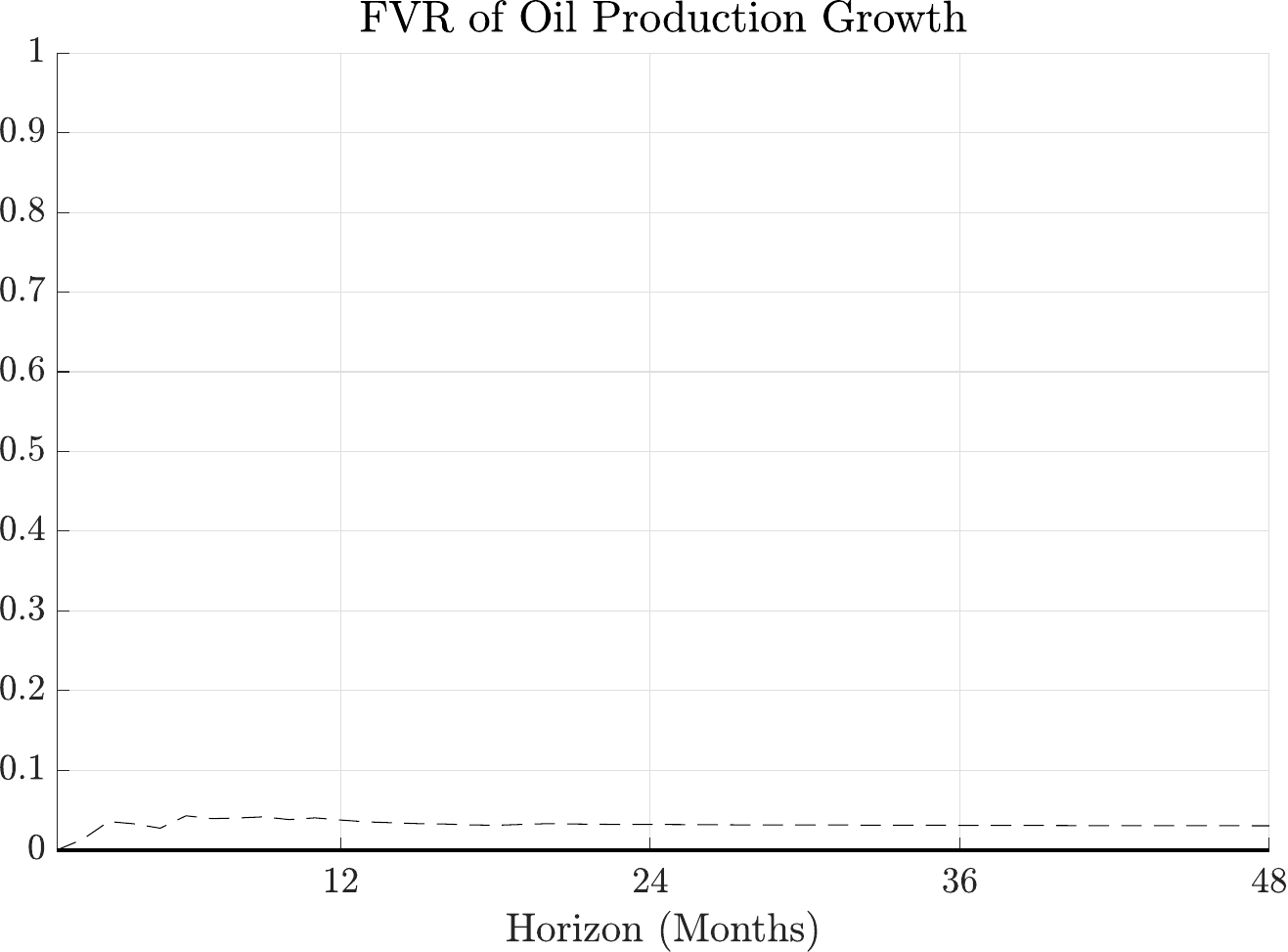} \\[2ex]
\includegraphics[width=0.48\linewidth]{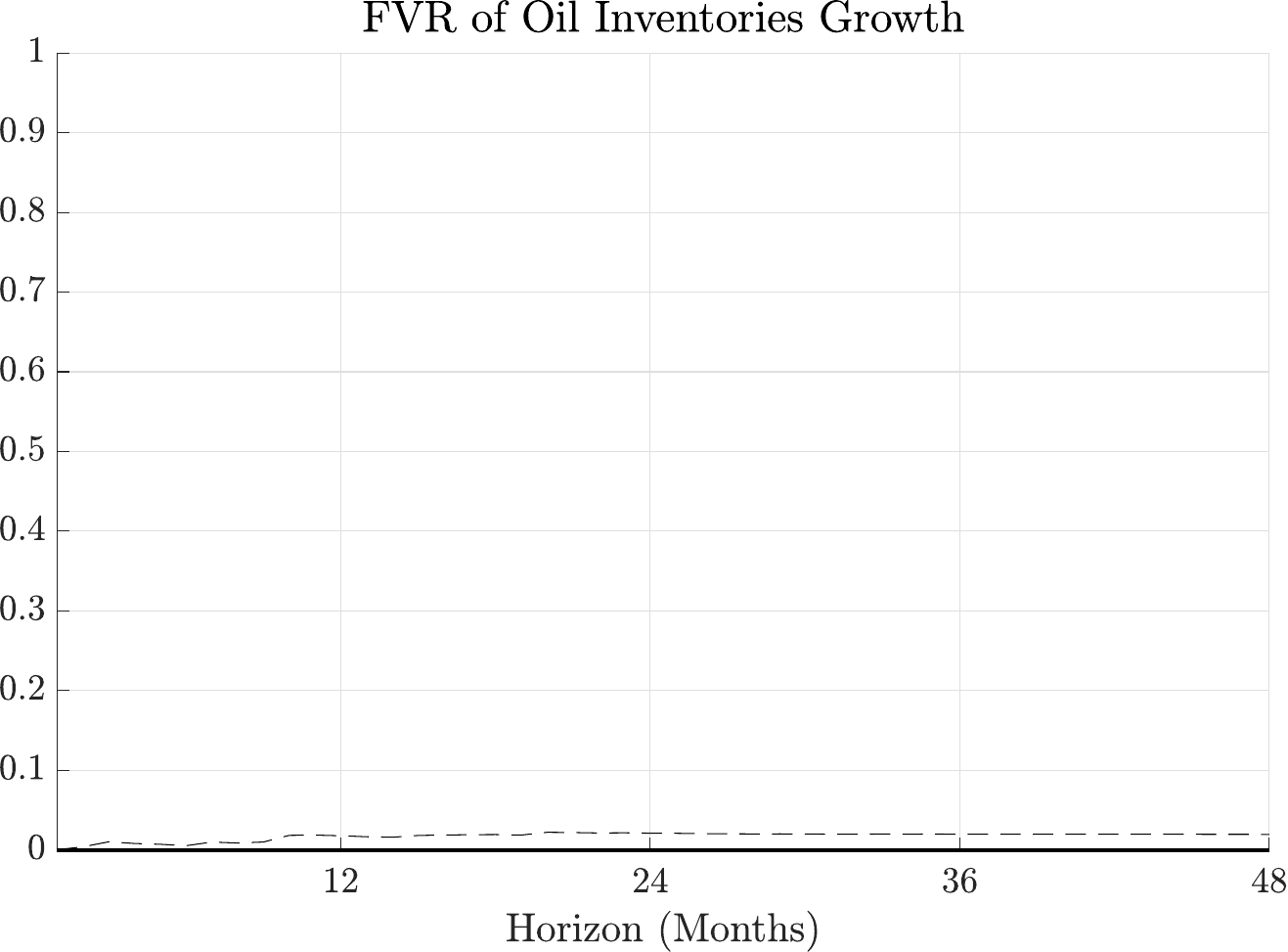} \includegraphics[width=0.48\linewidth]{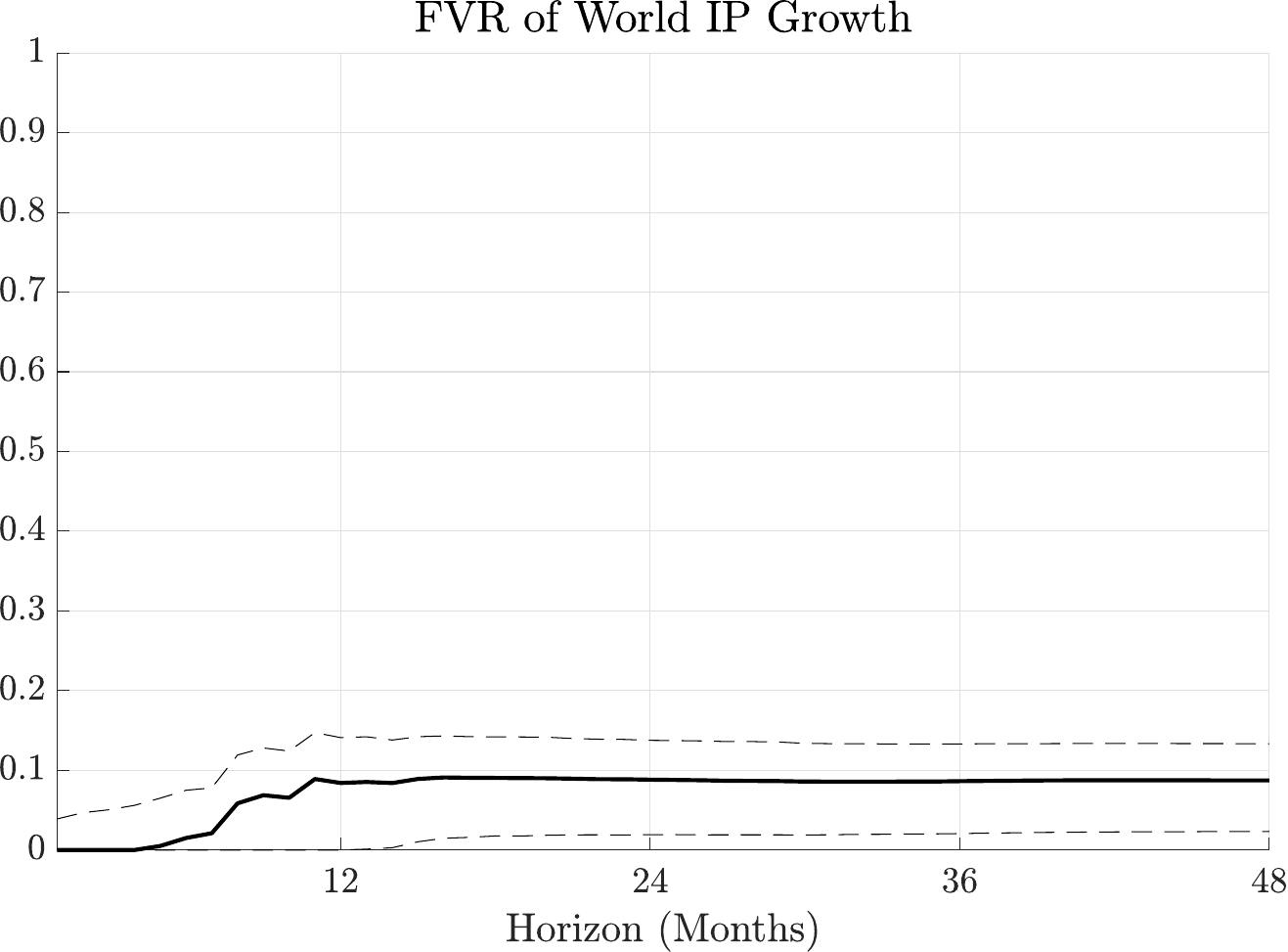} \\[2ex]
\caption{Point estimates and 90\% confidence intervals for oil news shock FVRs, across different variables and forecast horizons, produced using a conventional SVAR-IV approach. For visual clarity, we force bias-corrected estimates/bounds to lie in $[0, 1]$.}
\label{fig:kaenzig_svariv_fvr_1}
\end{figure}

\begin{figure}[p]
\centering
\textsc{Oil news shock: SVAR-IV FVRs, U.S. Variables} \\[2ex]
\includegraphics[width=0.48\linewidth]{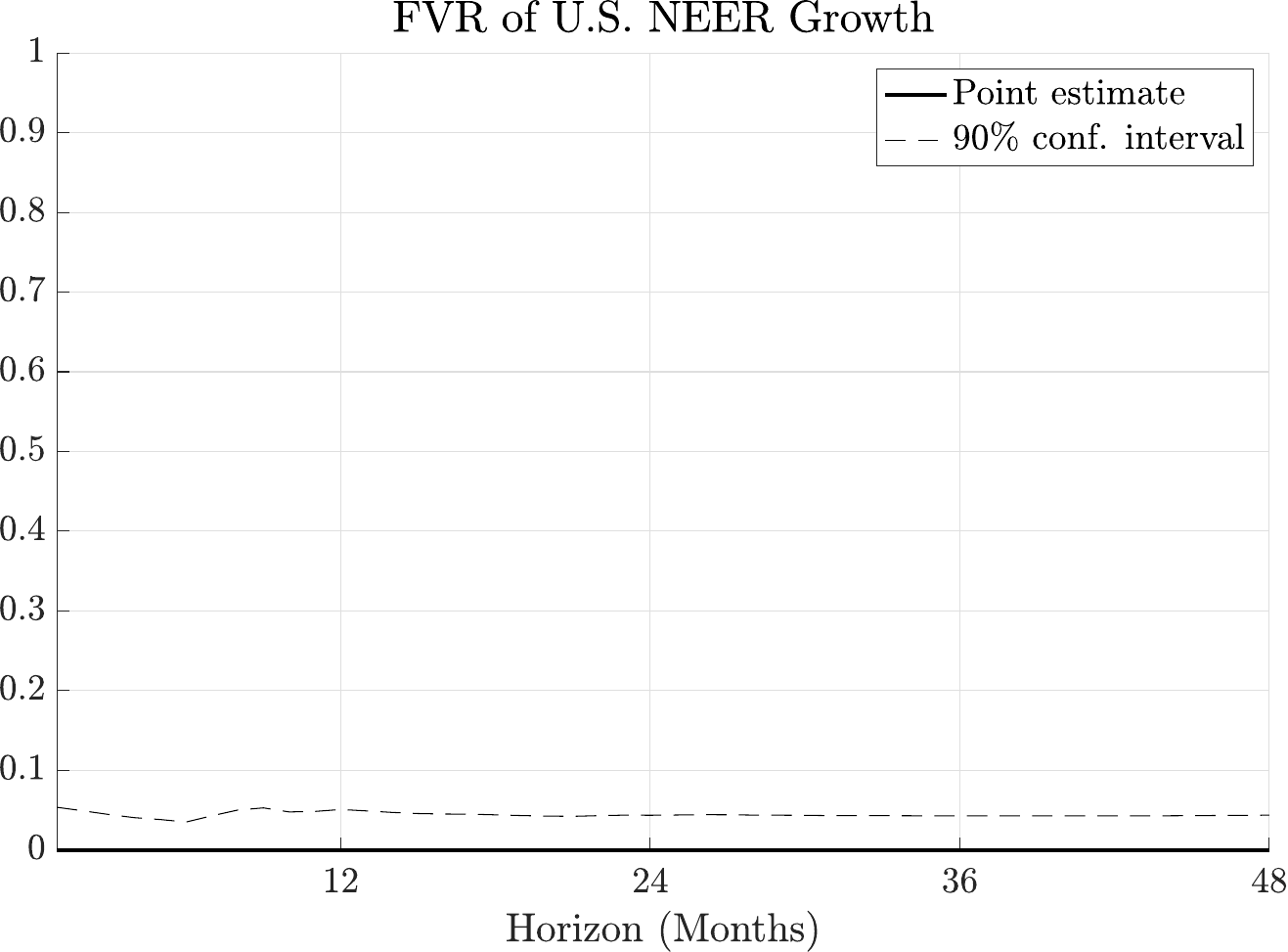} \includegraphics[width=0.48\linewidth]{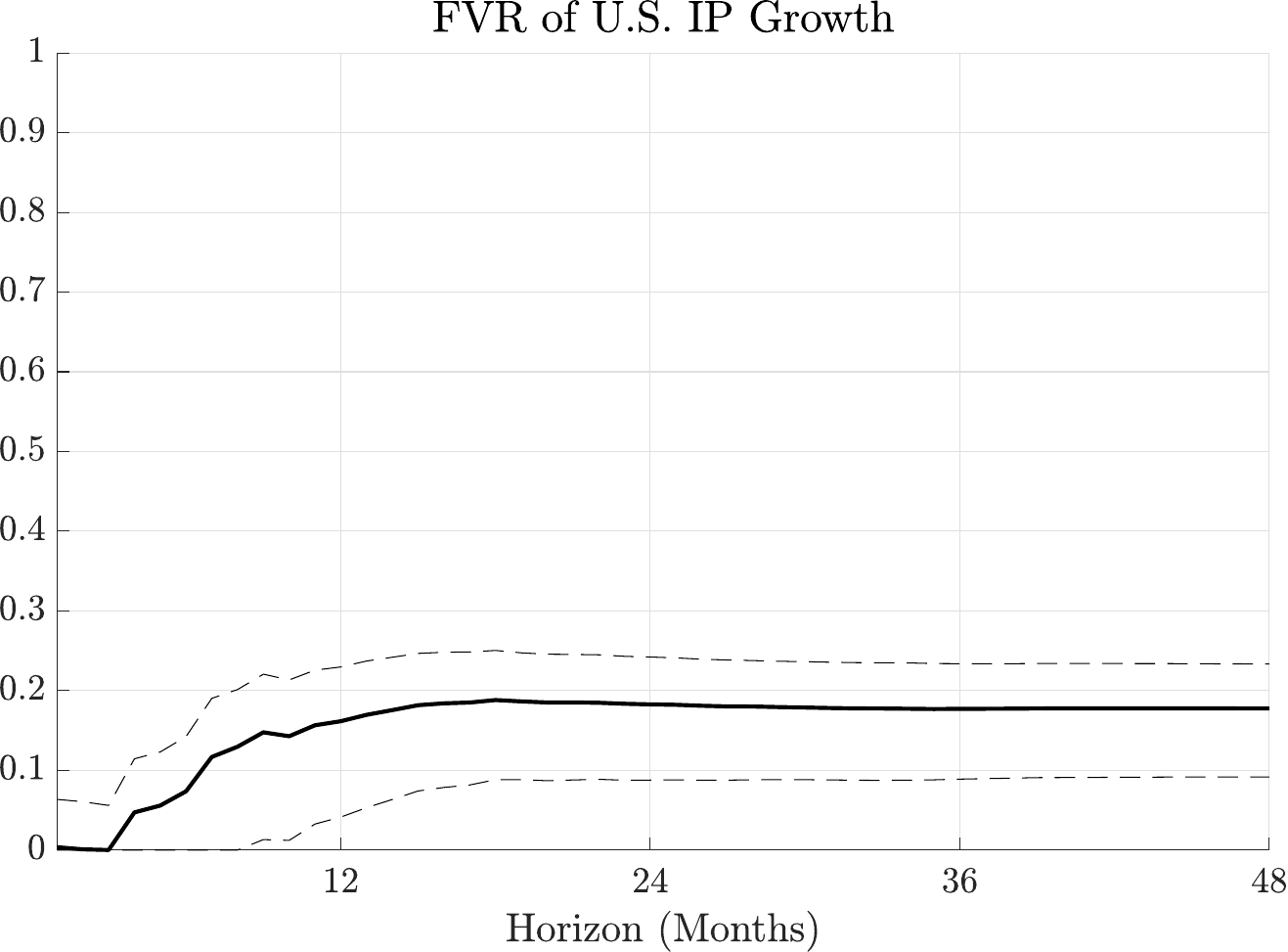} \\[2ex]
\includegraphics[width=0.48\linewidth]{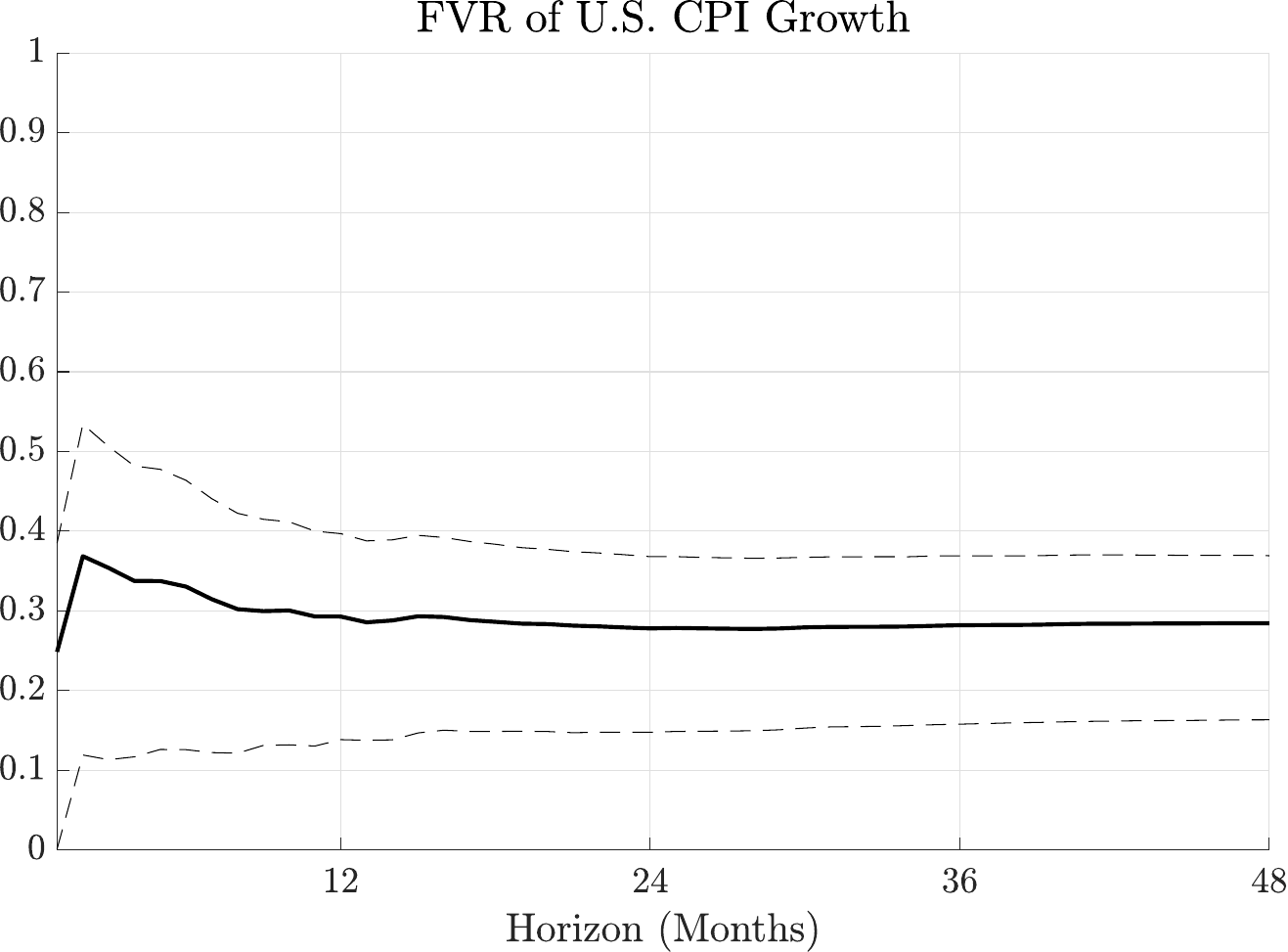} \includegraphics[width=0.48\linewidth]{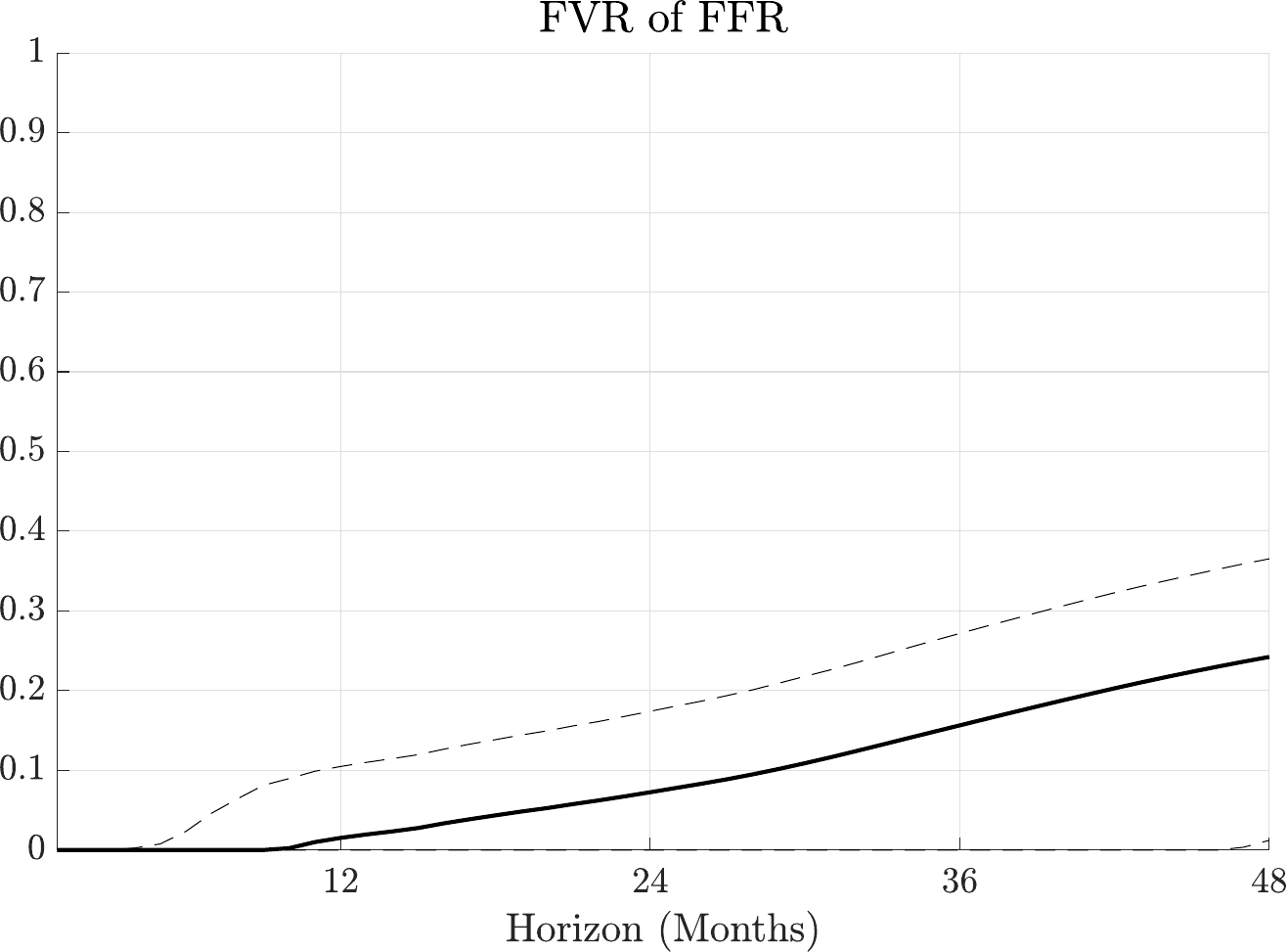} \\[2ex]
\includegraphics[width=0.48\linewidth]{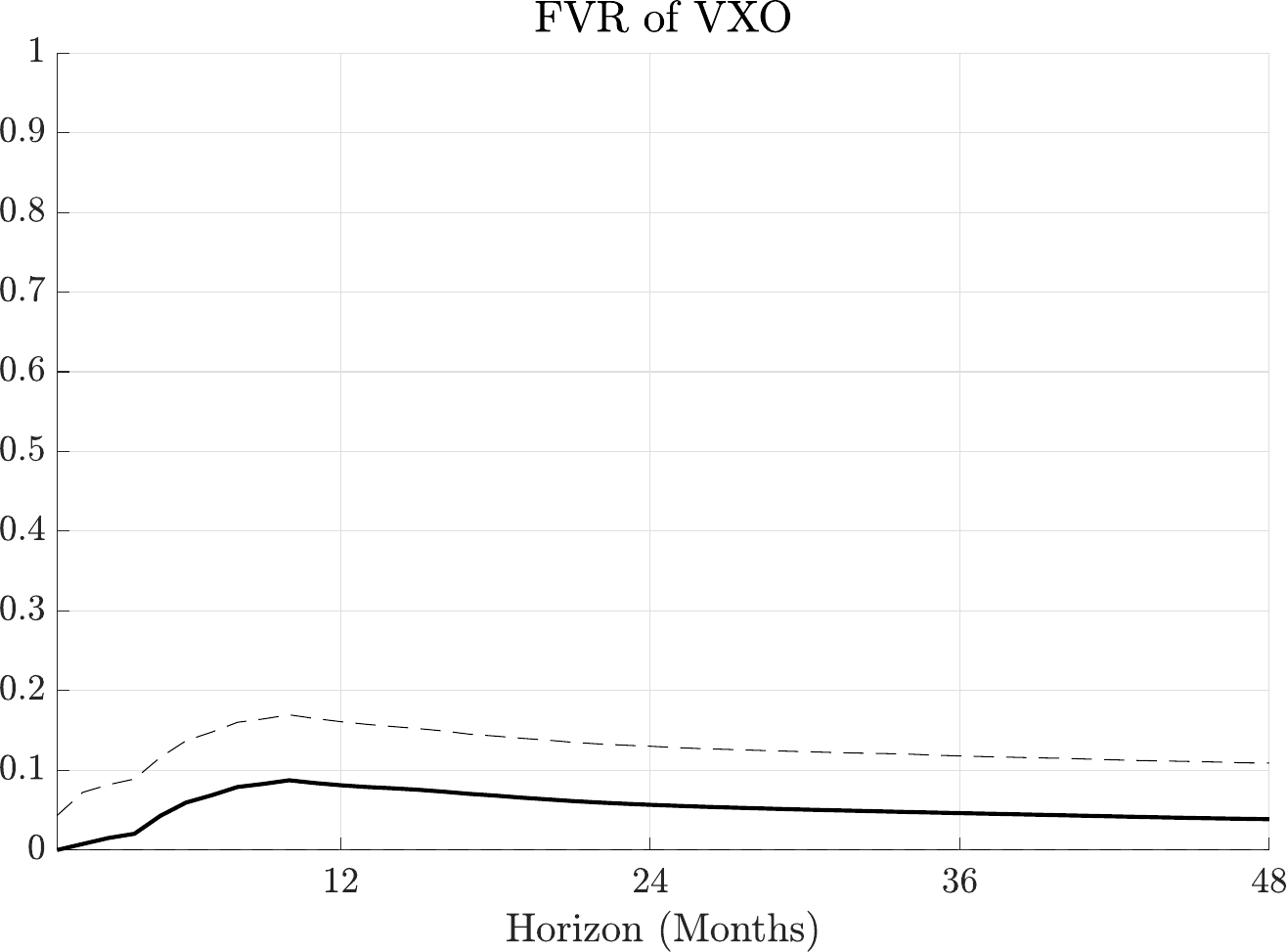} \includegraphics[width=0.48\linewidth]{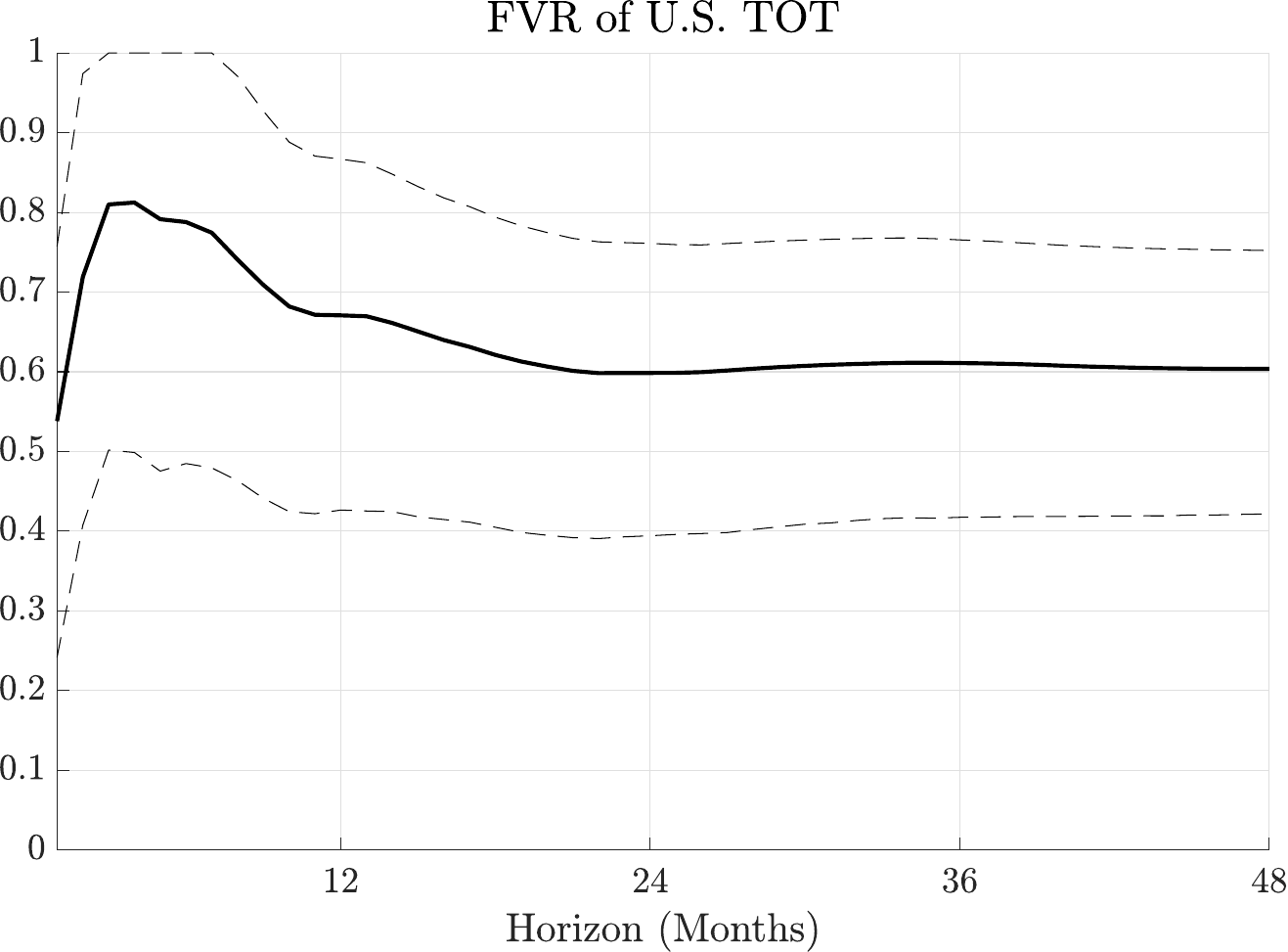} \\[2ex]
\caption{Point estimates and 90\% confidence intervals for oil news shock FVRs, across different variables and forecast horizons, produced using a conventional SVAR-IV approach. For visual clarity, we force bias-corrected estimates/bounds to lie in $[0, 1]$.}
\label{fig:kaenzig_svariv_fvr_2}
\end{figure}

\clearpage

\section{Nonparametric sieve VAR inference}
\label{sec:var_sieve}

In this appendix section we show that the bound estimates proposed in \cref{sec:how_to} are jointly asymptotically normal under nonparametric conditions on the DGP, as long as the VAR lag length is chosen to increase with the sample size at an appropriate rate. The nonparametric viewpoint does not change the practical steps necessary to implement the inference strategy; it only provides regularity conditions under which it is asymptotically innocuous to approximate the true VAR($\infty$) data generating process by a finite-lag VAR. We utilize the classic sieve VAR results of \citet{Lewis1985} (who build on the univariate results of \citealp{Berk1974}) to prove asymptotic normality of those nonlinear functionals of the estimated VAR spectrum that appear in our bounds. Our main result below is similar in spirit to the abstract theorem in \citet[Thm. 2]{Saikkonen2000}, although our regularity conditions are more easily verifiable as they are tailored to our parameters of interest (however, unlike \citeauthor{Saikkonen2000}, we only consider stationary data).

The purpose of this section is merely to demonstrate that existing sieve VAR theory implies that empirical SVMA-IV analysis can be carried out in a nonparametric fashion. We do not claim to provide conceptually new insights into  sieve VAR econometrics. Although here we only prove the validity of the sieve VAR strategy for delta method inference, we expect that similar results could be established for bootstrap sieve VAR inference in the SVMA-IV model, see \citet{Goncalves2007}, \citet{Meyer2015}, and references therein.

\subsection{Assumptions, parameters of interest, and estimator}
We first define the general class of parameters of interest for empirical SVMA-IV analysis, and we place assumptions on the DGP and VAR lag length. Our goal is to stay close to the set-up in \citet{Lewis1985}, so as to demonstrate how existing asymptotic results can be readily adapted to study sieve VAR estimators for SVMA-IV purposes.

We assume that the data are generated by a reduced-form VAR($\infty$) model with i.i.d. innovations. The observations are denoted by $W_t \equiv (y_t',z_t)' \in \mathbb{R}^{n_W}$, $t=1,2,\dots,T$, where $n_W \equiv n_y+1$. In order to make clear the connection with \citet{Lewis1985}, we assume that the data are known to have mean zero. It is straight-forward to extend all results to allow for non-zero means by including an intercept in the estimated VAR. Let $\|B\|\equiv (\tr(B'B))^{1/2}$ denote the Frobenius norm. 
\begin{asn} \label{asn:var_dgp}
The process $\lbrace W_t \rbrace$ is generated by the mean-zero stationary VAR($\infty$) model
\[A(L)W_t = e_t.\]
Here $A(z) \equiv I_{n_W} - \sum_{\ell=1}^\infty A_\ell z^\ell$ for $z \in \mathbb{C}$, and $A_\ell \in \mathbb{R}^{n_W\times n_W}$ for all $\ell$. We impose the following conditions:
\begin{enumerate}[i)]
\item $\det(A(z)) \neq 0$ for all $|z| \leq 1$, and $\sum_{\ell=1}^\infty \|A_\ell\| < \infty$.
\item $\lbrace e_t \rbrace$ is an $n_W$-dimensional i.i.d. process with $E(e_t)=0_{n_W \times 1}$, $\Sigma \equiv \var(e_t)$ is positive definite, and $E\|e_t\|^8<\infty$.
\end{enumerate}
\end{asn}
\noindent These conditions are the same as in \citet{Lewis1985}, except that we here assume that $e_t$ has 8 moments instead of just 4.\footnote{We only use more than four moments in the proofs of \cref{thm:var_sumcov,thm:var_sumnegl,thm:var_asynorm} below, where the extra moments make the arguments more transparent.} \citet{Meyer2015} discuss the generality of assuming a reduced-form VAR($\infty$) with i.i.d. disturbances, see also \citet{Kreiss2011} for more details in the univariate case. Assuming the SVMA-IV model \eqref{eqn:svma}--\eqref{eqn:iv} holds, the i.i.d. assumption on the one-step-ahead reduced-form forecast errors $e_t$ is automatically satisfied, provided that the structural shocks $(\varepsilon_t',v_t)'$ are themselves i.i.d. and either (i) invertible, or (ii) Gaussian (regardless of invertibility). Although we are here deliberately aiming at conceptual clarity rather than full generality, we expect it would be straight-forward to weaken the i.i.d.-ness assumption on $e_t$ by appealing to a suitable multivariate version of the sieve VAR result of \citet{Goncalves2007}, who assume heteroskedastic martingale difference innovations.

Next, we define the class of parameters of interest for empirical SVMA-IV analysis. Define the two matrix-valued functions
\[A_{\cos}(\omega) \equiv \sum_{\ell=1}^\infty A_\ell\cos(\omega\ell),\quad A_{\sin}(\omega) \equiv \sum_{\ell=1}^\infty A_\ell\sin(\omega\ell),\quad \omega \in [0,2\pi].\]
The parameter of interest is of the form
\[\psi \equiv \int_0^{2\pi} h(\omega)'g(A_{\cos}(\omega),A_{\sin}(\omega),\Sigma)\,d\omega,\]
where we define the functions $h \colon [0,2\pi] \to \mathbb{R}^K$ and $g \colon \mathcal{A}_\delta \times \mathbb{S}_{n_W} \to \mathbb{R}^K$, the set $\mathcal{A}_\delta = \lbrace (B_1,B_2) \in \mathbb{R}^{n_W \times n_W} \times \mathbb{R}^{n_W \times n_W} \colon |\det(I_{n_W}-B_1-i B_2)| \geq \delta \rbrace $, and the fixed number $\delta>0$ that is strictly smaller than $\inf_{\omega \in [0,2\pi]} |\det(A(e^{i\omega}))|$. $\mathbb{S}_{n_W}$ denotes the set of $n_W \times n_W$ symmetric positive definite matrices.

For appropriate choices of $h(\cdot)$ and $g(\cdot)$, the above class of parameters includes almost all the parameters/bounds in SVMA-IV analysis.\footnote{The only exception is the parameter $\sup_{\omega \in [0,2\pi]} s_{\tilde{z}^\dagger}(\omega)$, which is discussed in the main text.} For example, the class contains (i) elements $\Sigma_{ij}$ of $\Sigma$, (ii) the degree of recoverability $R_\infty^2 = \int_0^{2\pi} s_{y\tilde{z}}(\omega)^*s_{y}(\omega)^{-1} s_{y\tilde{z}}(\omega) \,d\omega$, and (iii) autocovariances $E(w_{i,t}w_{j,t-\ell}) = \int_0^{2\pi} e^{i\omega \ell} s_{w,ij}(\omega)\,d\omega$. Here $w_t \equiv (y_t',\tilde{z}_t)'$, and for all $\omega \in [0,2\pi]$,
\begin{align*}
s_{w}(\omega) = \left( \begin{array}{cc}
s_y(\omega) & s_{y\tilde{z}}(\omega) \\
s_{\tilde{z}y}(\omega) & s_{\tilde{z}}(\omega)
\end{array} \right) = \frac{1}{2\pi} &\left( \begin{array}{cc}
(I_{n_y},0_{n_y \times 1}) (A_{\cos}(\omega)+iA_{\sin}(\omega))^{-1} & 0_{n_y \times 1} \\
0_{1 \times n_y} & 1
\end{array} \right) \\
& \times \Sigma  \left( \begin{array}{cc}
 (A_{\cos}(\omega)-iA_{\sin}(\omega))^{-1'}(I_{n_y},0_{n_y \times 1})' & 0_{n_y \times 1} \\
0_{1 \times n_y} & 1
\end{array} \right).
\end{align*}
Other SVMA-IV parameters can be constructed as nonlinear transformations of a finite number of autocovariances. By the Cram\'{e}r-Wold device, it is without loss of generality to consider vector-valued (rather than matrix-valued) functions $h(\cdot)$ and $g(\cdot)$. In the following, we further assume $K=1$ so that both $h(\cdot)$ and $g(\cdot)$ are scalar. This eases the notation without sacrificing essential generality, as should be clear from the proofs.

We place certain smoothness conditions on the parameter of interest, thus permitting a delta method argument.

\begin{asn} \label{asn:var_param}
The function $h(\cdot)$ is continuous on $[0,2\pi]$. On any non-empty, compact subset of the domain $\mathcal{A}_\delta \times \mathbb{S}_{n_W}$, the function $g(\cdot,\cdot,\cdot)$ is twice continuously differentiable. Denote the partial derivatives by $g_1(B_1,B_2,S) \equiv \frac{\partial g(B_1,B_2,S)}{\partial \ve(B_1)}$, $g_2(B_1,B_2,S) \equiv \frac{\partial g(B_1,B_2,S)}{\partial \ve(B_2)}$, and $g_3(B_1,B_2,S) \equiv \frac{\partial g(B_1,B_2,S)}{\partial \ve(S)}$. At the true VAR parameters $\lbrace A_\ell \rbrace$ and $\Sigma$, each of the functions $\omega \mapsto g_j(A_{\cos}(\omega),A_{\sin}(\omega),\Sigma)$, $j=1,2,3$, belongs to $L_2(0,2\pi)$ (element-wise).
\end{asn}

\noindent The smoothness conditions in \cref{asn:var_param} are easily verified for all parameters of interest in SVMA-IV analysis, since \cref{asn:var_dgp} ensures that the true VAR spectrum is non-singular.

Finally, we define a sieve VAR estimator as the sample analogue of the population parameter of interest. For any $p \in \mathbb{N}$, define $X_t(p) \equiv (W_{t-1}',\dots,W_{t-p}')' \in \mathbb{R}^{n_Wp}$ and the least-squares VAR estimator
\[\hat{\beta}(p) \equiv \left(\hat{A}_1(p),\dots,\hat{A}_p(p) \right) \equiv \left(\sum_{t=p+1}^T W_t(p)X_t(p)'\right) \left( \sum_{t=p+1}^T X_t(p)X_t(p)' \right)^{-1}.\]
Let $\hat{\Sigma}(p) \equiv (T-p)^{-1}\sum_{t=p+1}^T \hat{e}_t(p)\hat{e}_t(p)'$, where $\hat{e}_t(p) \equiv W_t - \hat{\beta}(p)X_t(p)$. Define also
\[\hat{A}_{\cos}(\omega;p) \equiv \sum_{\ell=1}^{p} \hat{A}_\ell\cos(\omega\ell),\quad \hat{A}_{\sin}(\omega;p) \equiv \sum_{\ell=1}^p \hat{A}_\ell\sin(\omega\ell),\quad \omega \in [0,2\pi].\]
The VAR($p$) estimator of the parameter of interest $\psi$ is then
\[\hat{\psi}(p) \equiv \int_0^{2\pi} h(\omega)'g(\hat{A}_{\cos}(\omega;p),\hat{A}_{\sin}(\omega;p),\hat{\Sigma})\,d\omega.\]
The VAR lag length $p=p_T$ must be chosen to grow with the sample size $T$ at an appropriate rate, unless the true DGP is a finite-order VAR.
\begin{asn} \label{asn:var_lagorder}
$p_T \in \mathbb{N}$ is a deterministic function of the sample size $T$ such that $p_T^3/T \to 0$ and $T^{1/2}\sum_{\ell=p_T+1}^\infty \|A_\ell\| \to 0$ as $T \to \infty$.
\end{asn}
\noindent These conditions are adopted from \citet[Thm. 2]{Lewis1985}, see also \citet{Berk1974}. The last condition in \cref{asn:var_lagorder} amounts to oversmoothing (i.e., choosing the lag length $p$ so large that the variance dominates the mean square error), which ensures that the nonparametric bias does not show up in asymptotic limiting distributions. If the partial autocorrelations of the data decay exponentially fast with the lag length, \cref{asn:var_lagorder} is satisfied by choosing $p_T \propto T^\phi$ for any $\phi \in (0,1/3)$. If the true DGP is a finite-order VAR, we may select $p_T$ to be any constant greater than the true lag length.

\subsection{Main convergence results}
\label{sec:var_mainresults}
We now state our main results on the asymptotic normality of the sieve VAR estimator and the consistency of the asymptotic variance estimator.

In preparation for stating our results, define for all $T$ the vector $\nu_T = (\nu_{1,T}',\dots,\nu_{p_T,T}')' \in \mathbb{R}^{n_W^2p_T}$, where
\[\nu_{\ell,T} \equiv \int_0^{2\pi} h(\omega)\left\lbrace g_1(A_{\cos}(\omega),A_{\sin}(\omega),\Sigma)\cos(\omega \ell) + g_2(A_{\cos}(\omega),A_{\sin}(\omega),\Sigma)\sin(\omega \ell) \right\rbrace \, d\omega \in \mathbb{R}^{n_W^2}\]
for $\ell=1,2,\dots,p_T$. Define also
\[\xi \equiv \int_0^{2\pi} h(\omega)g_3(A_{\cos}(\omega),A_{\sin}(\omega),\Sigma)\, d\omega \in \mathbb{R}^{n_W^2}.\]
We also define the estimators $\hat{\nu}_T$ and $\hat{\xi}(p_T)$ of $\nu_T$ and $\xi$ obtained by substituting $A_{\cos}(\cdot)$ and $A_{\sin}(\cdot)$ with $\hat{A}_{\cos}(\cdot;p_T)$ and $\hat{A}_{\sin}(\cdot;p_T)$ in the above formulas. Finally, we define $\Gamma(p) \equiv E(X_t(p)X_t(p)')$ for all $p \in \mathbb{N}$ and the sample analogue $\hat{\Gamma}(p) \equiv (T-p)^{-1}\sum_{t=p+1}^T X_t(p)X_t(p)'$. In the rest of this section, all convergence statements are understood to be taken as $T \to \infty$.

Our first main proposition states that the sieve VAR estimator of the parameter of interest is asymptotically normal under our nonparametric conditions on the data generating process, the conditions on the estimated VAR lag order, and the regularity conditions on the parameter of interest.
\begin{prop} \label{thm:var_prop_asynorm}
Let \cref{asn:var_dgp,asn:var_param,asn:var_lagorder} hold. Assume $\sigma_{\psi}^2 \equiv \lim_{T\to\infty}  \nu_T'(\Gamma(p_T)^{-1} \otimes \Sigma)\nu_T + \xi'\var(e_t \otimes e_t)\xi$ is strictly positive and that the limit exists. Then
\[(T-p_T)^{1/2}(\hat{\psi}(p_T)-\psi) \stackrel{d}{\to} N(0,\sigma_{\psi}^2).\]
\end{prop}
\noindent Under our regularity conditions on the parameter of interest, the convergence rate of the sieve VAR estimator $\hat{\psi}(p_T)$ is $(T-p_T)^{-1/2} = O(T^{-1/2})$. The condition that $\sigma_{\psi}^2$ exists and is nonzero rules out degenerate parameters that can be estimated super-consistently. This condition could for example be violated if the true parameter of interest is on the boundary of its parameter space (e.g., if the true FVD is 0, or the true degree of invertibility is 1). Such issues are not unique to SVMA-IV and could similarly arise in SVAR inference.

Our second main proposition states that the usual delta method standard errors for a VAR($p_T$) model are valid asymptotically.

\begin{prop} \label{thm:var_prop_consvar}
Let the assumptions of \cref{thm:var_prop_asynorm} hold. Let $\hat{\sigma}_{\psi}^2(p_T) \equiv \hat{\nu}_T'(\hat{\Gamma}(p_T)^{-1} \otimes \hat{\Sigma}(p_T))\hat{\nu}_T + \hat{\xi}(p_T)'\hat{\Xi}(p_T)\hat{\xi}(p_T)$, where $\hat{\chi}_t(p_T) \equiv \ve(\hat{e}_t(p_T)\hat{e}_t(p_T)' - \hat{\Sigma}(p_T))$ and $\hat{\Xi}(p_T) \equiv (T-p_T)^{-1}\sum_{t=p_T+1}^T \hat{\chi}_t(p_T) \hat{\chi}_t(p_T)' $. Then
\[\hat{\sigma}_{\psi}^2(p_T) \stackrel{p}{\to} \sigma_{\psi}^2.\]
\end{prop}
\noindent Observe that $\hat{\sigma}_{\psi}^2(p_T)$ is precisely the asymptotic variance estimator for $\hat{\psi}(p_T)$ that one would compute from the delta method formula based on a VAR($p_T$) model for the data.

To summarize, \cref{thm:var_prop_asynorm,thm:var_prop_consvar} imply that delta method inference based on the estimated VAR($p_T$) process is valid asymptotically even if the true DGP is a VAR($\infty$). Hence, the partial identification robust confidence intervals proposed in \cref{sec:how_to} are valid under our regularity conditions. This conclusion is consistent with the finite-sample simulation evidence presented in \cref{sec:simulation}.

\clearpage

\section{Additional proofs and auxiliary lemmas}
Here we prove \cref{thm:semidef} and all additional results stated in this appendix. We first prove results related to the SVMA-IV identification analysis. Then we address the sieve VAR convergence results.

\subsection{Proof of \texorpdfstring{\cref{thm:semidef}}{Lemma \ref{thm:semidef}}}
\label{sec:proof_semidef}
We focus on the semidefiniteness statement. Decompose $B = B^{1/2}B^{1/2*}$ and define $\tilde{b} = B^{-1/2}b$. The statement of the lemma is equivalent with the statement that $I_n - x^{-1}\tilde{b}\tilde{b}^*$ is positive semidefinite if and only if $x \geq b^*b$. Let $\nu$ be an arbitrary $n$-dimensional complex vector satisfying $\nu^*\nu=1$. Then
\[\nu^*\left(I_n - x^{-1}\tilde{b}\tilde{b}^*\right)\nu = 1 - \frac{\tilde{b}^*\tilde{b}}{x} \cos^2\left(\theta(\nu,\tilde{b})\right),\]
where $\theta(\nu,\tilde{b})$ is the angle between $\nu$ and $\tilde{b}$. Evidently, $x^{-1} \tilde{b}^*\tilde{b} \leq 1$ is precisely the condition needed to ensure that the above display is nonnegative for every choice of $\nu$. \qed

\subsection{Auxiliary lemma for proof of \texorpdfstring{\cref{thm:identif_fvd}}{Proposition \ref{thm:identif_fvd}}}

\begin{lem} \label{thm:forec_var_incr}
Let $x_t$ and $\tilde{x}_t$ be two stationary $n$-dimensional Gaussian time series whose spectral densities $s_x(\omega)$ and $s_{\tilde{x}}(\omega)$ are such that $s_{\tilde{x}}(\omega)-s_x(\omega)$ is positive semidefinite for all $\omega \in [0,2\pi]$. Then $\var(\mu'x_{t+\ell} \mid \lbrace x_\tau \rbrace_{-\infty<\tau\leq t}) \leq \var(\mu'\tilde{x}_{t+\ell} \mid \lbrace \tilde{x}_\tau \rbrace_{-\infty<\tau\leq t})$ for all $\ell=1,2,\dots$ and all constant vectors $\mu \in \mathbb{R}^n$.
\end{lem}
\begin{proof}
We may define an $n$-dimensional stationary Gaussian process $\nu_t$ with spectral density $s_{\nu}(\omega)=s_{\tilde{x}}(\omega)-s_x(\omega)$, $\omega \in [0,2\pi]$, and such that the $\nu_t$ process is independent of the $x_t$ process. Then the process $\check{x}_t = x_t + \nu_t$ has the same distribution as the $\tilde{x}_t$ process. Hence,
\begin{align*}
\var(\mu'\tilde{x}_{t+\ell} \mid \lbrace \tilde{x}_\tau \rbrace_{-\infty<\tau\leq t}) &= \var(\mu'\check{x}_{t+\ell} \mid \lbrace \check{x}_\tau \rbrace_{-\infty<\tau\leq t}) \\
&\geq \var(\mu'\check{x}_{t+\ell} \mid \lbrace x_\tau,\nu_t \rbrace_{-\infty<\tau\leq t}) \\
&= \var(\mu'x_{t+\ell} \mid \lbrace x_\tau,\nu_t \rbrace_{-\infty<\tau\leq t}) + \var(\mu'\nu_{t+\ell} \mid \lbrace x_\tau,\nu_t \rbrace_{-\infty<\tau\leq t}) \\ 
&\geq \var(\mu'x_{t+\ell} \mid \lbrace x_\tau,\nu_t \rbrace_{-\infty<\tau\leq t}) \\
&= \var(\mu'x_{t+\ell} \mid \lbrace x_\tau \rbrace_{-\infty<\tau\leq t}).
\end{align*}
The second equality above uses that the independence of the $x_t$ and $\nu_t$ processes implies that $x_{t+\ell}$ and $\nu_{t+\ell}$ are independent also conditional on $\lbrace x_\tau,\nu_t \rbrace_{-\infty<\tau\leq t}$.
\end{proof}

\subsection{Proof of \texorpdfstring{\cref{thm:identif_fvd}}{Proposition \ref{thm:identif_fvd}}}
\label{sec:proof_fvd}
The proof proceeds in two steps. First, for a given known $\alpha$, we show that $\mathit{FVD}_{i,\ell}$ is sharply bounded above by 1 and below by \eqref{eqn:FVD_lb}. Second, we show that the lower bound is monotonically decreasing in $\alpha$, so that the overall lower bound is attained by $\alpha_{UB}$.

\begin{enumerate}[1.]

\item Given $\alpha \in (\alpha_{LB}, \alpha_{UB}]$, the numerator of $\mathit{FVD}_{i,\ell}$ is point-identified (see below), so we need only concern ourselves with the denominator. We can write the denominator as
\begin{align} 
\var(y_{i,t+\ell} \mid \lbrace \varepsilon_\tau \rbrace_{-\infty<\tau\leq t}) &= \sum_{m=0}^{\ell-1} \Theta_{i,1,m}^2 + \sum_{j=2}^{n_\varepsilon} \sum_{m=0}^{\ell-1} \Theta_{i,j,m}^2 \nonumber \\
&= \frac{1}{\alpha^2}\sum_{m=0}^{\ell-1} \cov(y_{i,t},\tilde{z}_{t-m})^2 + \sum_{j=2}^{n_\varepsilon} \sum_{m=0}^{\ell-1} \Theta_{i,j,m}^2. \label{eqn:fvd_denom}
\end{align}
Given $\alpha$, the first term in \eqref{eqn:fvd_denom} is point-identified (note that it equals the numerator of the FVD), while the second is not. To upper-bound $\mathit{FVD}_{i,\ell}$, we seek to make that second term as small as possible. In fact, we can always set it to $0$. To see this, let $\lbrace \Theta_{\bullet, j,m} \rbrace_{2 \leq j\leq n_\varepsilon,0\leq m <\infty}$ denote some sequence of impulse responses for the structural shocks $j \neq 1$ that is consistent with the second-moment properties of the data. Since $\alpha \in (\alpha_{LB}, \alpha_{UB}]$, such a sequence exists by \cref{thm:identif_alpha}. Now, for a given forecast horizon $\ell$, instead consider the new sequence $\lbrace \breve{\Theta}_{\bullet, j, m}\rbrace_{2 \leq j\leq n_\varepsilon,0\leq m <\infty}$, defined via
\begin{equation*}
\breve{\Theta}_{\bullet, j, m} = \begin{cases} 0_{n_y \times 1} & \text{if $m \leq \ell-1$}, \\ \Theta_{\bullet, j, m-\ell} & \text{if $m > \ell-1$}. \end{cases}
\end{equation*}
Then the stochastic process induced by $\lbrace \breve{\Theta}_{\bullet, j, m}\rbrace_{2 \leq j\leq n_\varepsilon,0\leq m <\infty}$ has the exact same second-moment properties as the (by assumption admissible) stochastic process induced by $\lbrace \Theta_{\bullet, j,m} \rbrace_{2 \leq j\leq n_\varepsilon,0\leq m <\infty}$. However, by construction, we now have $\mathit{FVD}_{i,\ell} = 1$, as claimed.

For the lower bound, we want to make the second term in \eqref{eqn:fvd_denom} as large as possible. Given a known $\alpha \in (\alpha_{LB},\alpha_{UB}]$, define
\[\tilde{y}_t^{(\alpha)} = (\tilde{y}_{1,t}^{(\alpha)},\dots,\tilde{y}_{n_y,t}^{(\alpha)})' \equiv y_t -\frac{1}{\alpha}\sum_{\ell=0}^\infty \cov(y_t,\tilde{z}_{t-\ell})\varepsilon_{1,t-\ell} = \sum_{j=2}^{n_\varepsilon} \sum_{\ell=0}^\infty \Theta_{\bullet,j,\ell}\varepsilon_{j,t-\ell} ,\]
whose spectral density is given by the expression stated in the proposition. We have
\[\var(\tilde{y}_{i,t+\ell}^{(\alpha)} \mid \lbrace \tilde{y}_{\tau}^{(\alpha)} \rbrace_{-\infty<\tau\leq t}) \geq \var(\tilde{y}_{i,t+\ell}^{(\alpha)} \mid \lbrace \varepsilon_{j,\tau} \rbrace_{2\leq j\leq n_\varepsilon,-\infty<\tau\leq t}) = \sum_{j=2}^{n_\varepsilon} \sum_{m=0}^{\ell-1} \Theta_{i,j,m}^2,\]
so the second term in \eqref{eqn:fvd_denom} has an point-identified upper bound. Thus, given $\alpha$, $\mathit{FVD}_{i,\ell}$ is bounded below by the expression \eqref{eqn:FVD_lb}.

We now argue that the lower bound \eqref{eqn:FVD_lb} is attained by an admissible model with the given $\alpha$. To that end, consider the Wold decomposition of $\tilde{y}_t^{(\alpha)}=\sum_{\ell=0}^\infty \tilde{\Theta}_\ell \tilde{\varepsilon}_{t-\ell}$, where the $\tilde{\Theta}_\ell$ matrices are $n_y \times n_y$, and $\tilde{\varepsilon}_t$ is $n_y$-dimensional i.i.d. standard normal and spanned by $\lbrace \tilde{y}_{\tau}^{(\alpha)} \rbrace_{-\infty<\tau \leq t}$.\footnote{Since $\alpha>\alpha_{LB}$, the Wold decomposition has no deterministic term, cf. the proof of \cref{thm:identif_alpha}.} Then $\var(\tilde{y}_{i,t+\ell}^{(\alpha)} \mid \lbrace \tilde{y}_{\tau}^{(\alpha)} \rbrace_{-\infty<\tau\leq t}) = \sum_{j=2}^{n_\varepsilon} \sum_{m=0}^{\ell-1} \tilde{\Theta}_{i,j,m}^2$, so the following model attains the lower bound \eqref{eqn:FVD_lb} and is consistent with the given spectrum $s_w(\cdot)$:
\begin{align}
y_t &= \frac{1}{\alpha}\sum_{\ell=0}^\infty \cov(y_t,\tilde{z}_{t-\ell})\overline{\varepsilon}_{1,t} + \sum_{\ell=0}^\infty \tilde{\Theta}_\ell \tilde{\varepsilon}_{t-\ell}, \nonumber \\
\tilde{z}_t &= \alpha \overline{\varepsilon}_{1,t} + \sqrt{\var(\tilde{z}_t)-\alpha^2} \times \overline{v}_t, \label{eqn:FVD_lb_sharp_repr} \\
(\overline{\varepsilon}_{1,t},\tilde{\varepsilon}_t',\overline{v}_t)' &\stackrel{i.i.d.}{\sim} N(0, I_{n_y+2}). \nonumber
\end{align}

\item \cref{thm:forec_var_incr} implies that $\var(\tilde{y}_{i,t+\ell}^{(\alpha)} \mid \lbrace \tilde{y}_{\tau}^{(\alpha)} \rbrace_{-\infty<\tau\leq t})$ is increasing in $\alpha$. Hence, the expression \eqref{eqn:FVD_lb} is decreasing in $\alpha$, as claimed. At $\alpha=\alpha_{UB}$, the representation \eqref{eqn:FVD_lb_sharp_repr} has $\tilde{z}_t = \alpha_{UB}\overline{\varepsilon}_{1,t}$, so we can represent $\tilde{y}_t^{(\alpha_{UB})} = y_t - E(y_t \mid \lbrace \overline{\varepsilon}_{1,\tau}\rbrace_{-\infty<\tau \leq t}) = y_t - E(y_t \mid \lbrace \tilde{z}_{\tau}\rbrace_{-\infty<\tau \leq t})$. \qed

\end{enumerate}

\subsection{Proof of \texorpdfstring{\cref{thm:iv_multi}}{Proposition \ref{thm:iv_multi}}}
The ``only if'' part was proved already in the text of \cref{sec:iv_multi_details}. For the ``if'' part, assume that the cross-spectrum has the given factor structure. Since $\tilde{z}_t$ is serially uncorrelated, we can write $s_{\tilde{z}}(\cdot) = s_{\tilde{z}}$. Because $s_w(\omega)$ is positive definite, the Schur complement
\[s_{\tilde{z}} - s_{y\tilde{z}}(\omega)^*s_y(\omega)^{-1}s_{y\tilde{z}}(\omega) = s_{\tilde{z}} - \eta \zeta(\omega)^*s_y(\omega)^{-1}\zeta(\omega)\eta'\]
is also positive definite. Pre-multiplying the above expression by $\eta's_{\tilde{z}}^{-1}$, post-multiplying by $s_{\tilde{z}}^{-1}\eta$, and rearranging the positive definiteness condition, we obtain the implication that
\[2\pi\zeta(\omega)^*s_y(\omega)^{-1}\zeta(\omega) < \frac{2\pi}{\eta's_{\tilde{z}}^{-1}\eta},\quad \omega \in [0,2\pi].\]
Now choose any $\overline{\alpha}\geq 0$ such that $\overline{\alpha}^2$ lies strictly between the left- and right-hand sides in the above inequality. The matrix
\[\overline{\Sigma}_v \equiv 2\pi s_{\tilde{z}} - \overline{\alpha}^2 \eta\eta'\]
is then positive definite by \cref{thm:semidef}. Moreover, the same lemma implies that
\[s_y(\omega) - \frac{2\pi}{\overline{\alpha}^2}\zeta(\omega)\zeta(\omega)^*\]
is positive definite for all $\omega \in [0,2\pi]$. If we set $\overline{\Theta}_{\bullet, 1}(L) = (2\pi/\overline{\alpha})\zeta(L)$, the same arguments as in the proof of \cref{thm:identif_alpha} show that there exists an $n_y \times n_y$ matrix polynomial $\tilde{\Theta}(L)$ such that the following model achieves the desired spectrum $s_w(\omega)$:
\begin{align*}
y_t &= \overline{\Theta}_{\bullet, 1}(L) \overline{\varepsilon}_{1,t} + \tilde{\Theta}(L)\tilde{\varepsilon}_t, \\
\tilde{z}_t &= \overline{\alpha} \eta \overline{\varepsilon}_{1,t} + \overline{\Sigma}_v^{1/2}\overline{v}_t, \\
(\overline{\varepsilon}_{1,t},\tilde{\varepsilon}_t',\overline{v}_t')' &\stackrel{i.i.d.}{\sim} N(0, I_{n_y+n_z+1}).
\end{align*}
Note that $\eta$ assumes the role of $\lambda$. \qed

\subsection{Proof of \texorpdfstring{\cref{prop:proxy_SVAR}}{Proposition \ref{prop:proxy_SVAR}}}
According to the model \eqref{eqn:svma}, we can write
\[u_t = \sum_{\ell=0}^\infty M_\ell \varepsilon_{t-\ell},\]
for some $n_y \times n_\varepsilon$ matrices $\lbrace M_\ell \rbrace$. Let $M_{\bullet,j,\ell}$ denote the $j$-th column of $M_\ell$. Then
\[\tilde{\varepsilon}_{1,t} = \gamma'u_t = \sum_{j=1}^{n_\varepsilon}\sum_{\ell=0}^\infty a_{j,\ell}\varepsilon_{j,t-\ell},\]
where $a_{j,\ell} = \gamma'M_{\bullet,j,\ell}$. We have $\var(\tilde{\varepsilon}_{1,t})=1$ by construction of $\gamma$, so $\sum_{j=1}^{n_\varepsilon}\sum_{\ell=0}^\infty a_{j,\ell}^2=1$. The expression for $\tilde{\Theta}_{\bullet,1,\ell}$ in the proposition also immediately follows from the above display and the fact $\cov(y_t,\varepsilon_{j,t-\ell})=\Theta_{\bullet,j,\ell}$. Next, observe that
\begin{align*}
R_0^2 &= \var(E(\varepsilon_{1,t} \mid \lbrace y_\tau \rbrace_{-\infty<\tau \leq t})) \\
&= \var(E(\varepsilon_{1,t} \mid \lbrace u_\tau \rbrace_{-\infty<\tau \leq t})) \\
&= \var(E(\varepsilon_{1,t} \mid u_t)) \\
&= \cov(u_t,\varepsilon_{1,t})'\Sigma_u^{-1}\cov(u_t,\varepsilon_{1,t}) \\
&= M_{\bullet,1,0}'\Sigma_u^{-1}M_{\bullet,1,0}.
\end{align*}
Since $\Sigma_{u\tilde{z}} = \sum_{\ell=0}^\infty M_\ell \cov(\varepsilon_{t-\ell},\tilde{z}_t) = \alpha M_{\bullet,1,0}$, we therefore have
\[\gamma = \frac{1}{\sqrt{\Sigma_{u\tilde{z}}'\Sigma_u^{-1}\Sigma_{u\tilde{z}}}}\Sigma_u^{-1}\Sigma_{u\tilde{z}} = \frac{1}{\sqrt{M_{\bullet,1,0}'\Sigma_u^{-1}M_{\bullet,1,0}}}\Sigma_u^{-1}M_{\bullet,1,0} = \frac{1}{\sqrt{R_0^2}}\Sigma_u^{-1}M_{\bullet,1,0}.\]
This implies
\[a_{1,0} = \gamma'M_{\bullet,1,0} = \sqrt{M_{\bullet,1,0}'\Sigma_u^{-1}M_{\bullet,1,0}} = \sqrt{R_0^2}.\]
Finally,
\[\tilde{\Theta}_{\bullet,1,0} = \cov(y_t, \tilde{\varepsilon}_{1,t}) = \cov(y_t, u_t)\gamma = \cov(u_t,u_t)\gamma = \Sigma_u \gamma = \frac{1}{\sqrt{R_0^2}}M_{\bullet,1,0},\]
and $M_{\bullet,1,0} = \cov(u_t,\varepsilon_{1,t}) = \cov(y_t,\varepsilon_{1,t}) = \Theta_{\bullet,1,0}$. \qed

\subsection{Auxiliary lemmas for sieve VAR results}
Here we define notation and state auxiliary lemmas used to prove the propositions in \cref{sec:var_sieve}. The lemmas are proved below. For any matrix $B$, let $\|B\|_1$ denote the largest singular value of $B$. Recall that $\|B\|_1 \leq \|B\|$ and $\|BC\| \leq \|B\|\|C\|_1$ for conformable matrices $B$ and $C$. Let $e_t(p) \equiv W_t - \beta(p)X_t(p)$ for all $t$ and $p$. Finally, define
\[A_{\cos}(\omega;p) \equiv \sum_{\ell=1}^p A_\ell\cos(\omega\ell),\quad A_{\sin}(\omega;p) \equiv \sum_{\ell=1}^p A_\ell\sin(\omega\ell),\quad \omega \in [0,2\pi],\; p \in \mathbb{N}.\]

\begin{lem}[\citealp{Lewis1985}, p. 397] \label{thm:var_Gammaconv}
Let \cref{asn:var_dgp,asn:var_lagorder} hold. Then $E(\|\hat{\Gamma}(p_T) - \Gamma(p_T)\|^2) =O(p_T^2/T)$.
\end{lem}

\begin{lem} \label{thm:var_betaconv}
Let \cref{asn:var_dgp,asn:var_lagorder} hold. Then $\|\hat{\beta}(p_T)-\beta(p_T)\| = O_p((p_T/T)^{1/2})$.
\end{lem}

\begin{lem} \label{thm:var_Sigmaconv}
Let \cref{asn:var_dgp,asn:var_lagorder} hold. Then $\hat{\Sigma}(p_T) - (T-p_T)^{-1}\sum_{t=p_T+1}^Te_te_t' = o_p(T^{-1/2})$.
\end{lem}

\begin{lem}[\citealp{Lewis1985}, Thm. 2] \label{thm:var_approx}
Let \cref{asn:var_dgp,asn:var_lagorder} hold. Let $\tilde{\nu}_T \in \mathbb{R}^{n_W^2 p_T}$ be a deterministic sequence of vectors such that $\|\tilde{\nu}_T\|^2 \leq M < \infty$ for all $T$. Define
\[\zeta_T \equiv (T-p_T)^{-1/2}\sum_{t=p_T+1}^T \tilde{\nu}_T'\left(\Gamma(p_T)^{-1}X_t(p_T) \otimes e_t\right).\]
Then
\[(T-p_T)^{1/2}\tilde{\nu}_T'\ve(\hat{\beta}(p_T)-\beta(p_T)) - \zeta_T \stackrel{p}{\to} 0.\]
\end{lem}

\begin{lem} \label{thm:var_sumcov}
Let \cref{asn:var_dgp} hold. Then for all $j_1,j_2,j_3,j_4 \in \lbrace 1,2,\dots,n_W\rbrace$, all $p,T \in \mathbb{N}$ such that $p < T$, and all $m_1,m_2,m_3,m_4 \in \mathbb{Z}$ we have
\[\frac{1}{T-p}\sum_{t=p+1}^T \sum_{s=p+1}^T \big|\cov(e_{j_1,t+m_1}e_{j_2,t+m_2}e_{j_3,t}e_{j_4,t},e_{j_1,s+m_3}e_{j_2,s+m_4}e_{j_3,s}e_{j_4,s})\big| \leq 9 E\|e_t\|^8.\]
\end{lem}

\begin{lem} \label{thm:var_sumnegl}
Let \cref{asn:var_dgp,asn:var_lagorder} hold.  Then
\begin{align*}
&\left\| \frac{1}{T-p_T} \sum_{t=p_T+1}^T \ve(e_t X_t(p_T)')\ve(e_t X_t(p_T)')' - E\left[\ve(e_t X_t(p_T)')\ve(e_t X_t(p_T)')'\right] \right\|^2 \\
&= O_p(p_T^2/T),
\end{align*}
and
\[\left\| \frac{1}{T-p_T} \sum_{t=p_T+1}^T \ve(e_t X_t(p_T)')\ve(e_te_t'-\Sigma)' \right\|^2 = O_p(p_T/T).\]
\end{lem}

\begin{lem} \label{thm:var_asynorm}
Let \cref{asn:var_dgp,asn:var_lagorder} hold. Define a sequence $\tilde{\nu}_T$ as in \cref{thm:var_approx}, and assume $v_\zeta \equiv \lim_{T \to \infty} \tilde{\nu}_T'(\Gamma(p_T)^{-1} \otimes \Sigma)\tilde{\nu}_T$ exists. Then
\[(T-p_T)^{1/2}\tilde{\nu}_T'\ve(\hat{\beta}(p_T)-\beta(p_T)) \stackrel{d}{\to} N(0,v_\zeta),\]
\[(T-p_T)^{1/2}\ve(\hat{\Sigma}(p_T)-\Sigma) \stackrel{d}{\to} N(0,\var(e_t \otimes e_t)),\]
and these two random vectors are asymptotically independent.
\end{lem}

\begin{lem} \label{thm:var_Acossinconv}
Let \cref{asn:var_dgp,asn:var_lagorder} hold. Then
\[\sup_{\omega \in [0,2\pi]} \left(\|\hat{A}_{\cos}(\omega;p_T)-A_{\cos}(\omega;p_T)\|^2 + \|\hat{A}_{\sin}(\omega;p_T)-A_{\sin}(\omega;p_T)\|^2 \right) = O_p(p_T/T).\]
\end{lem}

\begin{lem} \label{thm:var_compact}
Let \cref{asn:var_dgp,asn:var_lagorder} hold. For $M>0$, define $\mathcal{A}_M^0 \equiv \lbrace (B_1,B_2) \in \mathcal{A}_\delta \times \mathbb{R}^{n_W \times n_W} \colon \|B_j - \sum_{\ell=1}^\infty A_\ell\| \leq M,\; j=1,2\rbrace$ and $\mathcal{S}_M^0 = \lbrace \tilde{\Sigma} \in \mathbb{S}_{n_W} \colon \|\tilde{\Sigma}-\Sigma\| \leq M \rbrace$. Then there exists an $M < \infty$ such that
\[P\left((\hat{A}_{\cos}(\omega;p_T),\hat{A}_{\sin}(\omega;p_T)) \in \mathcal{A}_M^0 \text{ for all } \omega \in [0,2\pi],\; \hat{\Sigma}(p_T) \in \mathcal{S}_M^0 \right) \to 1.\]
\end{lem}

\begin{lem} \label{thm:var_deltamethod}
Let \cref{asn:var_dgp,asn:var_param,asn:var_lagorder} hold. Define $\nu_T$ and $\xi$ as in \cref{sec:var_mainresults}.
Then
\[(T-p_T)^{1/2}\left\lbrace (\hat{\psi}(p_T)-\psi) - \nu_T'\ve(\hat{\beta}(p_T)-\beta(p_T)) - \xi'\ve(\hat{\Sigma}-\Sigma) \right\rbrace \stackrel{p}{\to} 0.\]
\end{lem}

\subsection{Proof of \texorpdfstring{\cref{thm:var_betaconv}}{Lemma \ref{thm:var_betaconv}}}
The result follows almost directly from the proof of Thm. 1 in \citet{Lewis1985}. As in that proof, define
\[U_{1,T} \equiv \frac{1}{T-p_T}\sum_{t=p_T+1}^T (e_t(p_T)-e_t)X_t(p_T)',\quad U_{2,T} \equiv \frac{1}{T-p_T}\sum_{t=p_T+1}^T e_t X_t(p_T)'.\]
\citeauthor{Lewis1985}'s arguments show that $\|U_{1,T}\| = O_p(p_T^{1/2}\sum_{\ell=p_T+1}^\infty \|A_\ell\|) = o_p((p_T/T)^{1/2})$ and $\|U_{2,T}\| = O_p((p_T/T)^{1/2})$ under \cref{asn:var_dgp,asn:var_lagorder}. The rest of the arguments in \citeauthor{Lewis1985}'s proof now yields the desired convergence rate of $\hat{\beta}(p_T)$. \qed

\subsection{Proof of \texorpdfstring{\cref{thm:var_Sigmaconv}}{Lemma \ref{thm:var_Sigmaconv}}}
Recall the notation $U_{1,T}$ and $U_{2,T}$ in the proof of \cref{thm:var_betaconv}. Since
\begin{align*}
\hat{\Sigma} &= \frac{1}{T-p_T}\sum_{t=p_T+1}^T e_te_t'  + \frac{1}{T-p_T}\sum_{t=p_T+1}^T (\hat{e}_t-e_t)e_t' \\
&\quad +\frac{1}{T-p_T}\sum_{t=p_T+1}^T e_t(\hat{e}_t-e_t)' + \frac{1}{T-p_T}\sum_{t=p_T+1}^T (\hat{e}_t-e_t)(\hat{e}_t-e_t)' \\
&\equiv \frac{1}{T-p_T}\sum_{t=p_T+1}^T e_te_t' + R_{1,T} + R_{1,T}' + R_{2,T},
\end{align*}
we need to show $R_{1,T} = o_p(T^{-1/2})$ and $R_{2,T} = o_p(T^{-1/2})$.

Decompose $R_{1,T}$ as
\[R_{1,T} = \frac{1}{T-p_T}\sum_{t=p_T+1}^T (\hat{e}_t-e_t(p_T))e_t' + \frac{1}{T-p_T}\sum_{t=p_T+1}^T (e_t(p_T)-e_t)e_t' \equiv \tilde{R}_{1,T} + \tilde{R}_{2,T}.\]
Since $\hat{e}_t(p_T)-e_t(p_T) = (\beta(p_T)-\hat{\beta}(p_T))X_t(p_T)$, we have
\[\|\tilde{R}_{1,T}\| \leq \|\hat{\beta}(p_T)-\beta(p_T)\|\, \|U_{2,T}\| = O_p((p_T/T)^{1/2})O_p((p_T/T)^{1/2}) = o(T^{-1/2}).\]
Moreover, since $e_t-e_t(p_T) = \sum_{\ell=p_T+1}^\infty A_\ell W_{t-\ell}$,
\begin{align*}
E\|\tilde{R}_{2,T}\| &\leq \frac{1}{T-p_T}\sum_{t=p_T+1}^T \sum_{\ell=p_T+1}^\infty \|A_\ell\| E(\|W_{t-\ell} e_t'\|)  \\
&\leq \sum_{\ell=p_T+1}^\infty \|A_\ell\| (E\|W_t\|^2 \, E\|e_t\|^2)^{1/2} \\
&= \text{constant} \times \sum_{\ell=p_T+1}^\infty \|A_\ell\| \\
&= o_p(T^{-1/2}).
\end{align*}
Now decompose $R_{2,T}$ as
\begin{align*}
\frac{1}{T-p_T}\sum_{t=p_T+1}^T (\hat{e}_t-e_t)(\hat{e}_t-e_t)' &= \frac{1}{T-p_T}\sum_{t=p_T+1}^T (\hat{e}_t-e_t(p_T))(\hat{e}_t-e_t(p_T))' \\
&\qquad + \frac{1}{T-p_T}\sum_{t=p_T+1}^T (\hat{e}_t-e_t(p_T))(e_t(p)-e_t)' \\
&\qquad + \frac{1}{T-p_T}\sum_{t=p_T+1}^T (e_t(p)-e_t)(\hat{e}_t-e_t(p_T))' \\
&\qquad + \frac{1}{T-p_T}\sum_{t=p_T+1}^T (e_t(p_T)-e_t)(e_t(p_T)-e_t)' \\
&\equiv \hat{R}_{1,T} + \hat{R}_{2,T} + \hat{R}_{2,T}' + \hat{R}_{3,T}.
\end{align*}
We have
\begin{align*}
\|\hat{R}_{1,T}\| &\leq \|\hat{\beta}(p_T)-\beta(p_T)\|^2 \|\hat{\Gamma}(p_T)\|_1 \\
& \leq \|\hat{\beta}(p_T)-\beta(p_T)\|^2 (\|\hat{\Gamma}(p_T)-\Gamma(p_T)\|_1 + \|\Gamma(p_T)\|_1) \\
&= O_p(p_T/T),
\end{align*}
using \cref{thm:var_Gammaconv} and \cref{thm:var_betaconv}. Further,
\[\|\hat{R}_{2,T}\| \leq \|\hat{\beta}(p_T)-\beta(p_T)\|\,\|U_{1,T}\| = O_p((p_T/T)^{1/2})o_p((p/T)^{1/2}) = o_p(T^{-1/2}).\]
Finally,
\begin{align*}
E\|\hat{R}_{3,T}\| &\leq E\|e_t(p_T)-e_t\|^2 \\
&\leq \textstyle \sum_{\ell=p_T+1}^\infty \sum_{m=p_T+1}^\infty \|A_\ell\|\,\|A_m\|\, E(\|W_{t-\ell}\|\,\|W_{t-m}\|) \\
&\leq \text{constant} \times \left({\textstyle  \sum_{\ell=p_T+1}^\infty   }\|A_\ell\| \right)^2 \\
&= o(T^{-1}). \qed
\end{align*}

\subsection{Proof of \texorpdfstring{\cref{thm:var_sumcov}}{Lemma \ref{thm:var_sumcov}}}
By stationarity,
\begin{align}
& \frac{1}{T-p}\sum_{t=p+1}^T \sum_{s=p+1}^T \big|\cov(e_{j_1,t+m_1}e_{j_2,t+m_2}e_{j_3,t}e_{j_4,t},e_{j_1,s+m_3}e_{j_2,s+m_4}e_{j_3,s}e_{j_4,s})\big| \nonumber \\
&= \sum_{\ell=-(T-p-1)}^{T-p-1} \left(1 - \frac{|\ell|}{T-p} \right) \big|\cov(e_{j_1,\ell+m_1}e_{j_2,\ell+m_2}e_{j_3,\ell}e_{j_4,\ell},e_{j_1,m_3}e_{j_2,m_4}e_{j_3,0}e_{j_4,0})\big|. \label{eqn:sumcov}
\end{align}
We first argue that each term in the sum \eqref{eqn:sumcov} is bounded. This follows from Cauchy-Schwarz:
\begin{align*}
&\big|\cov(e_{j_1,\ell+m_1}e_{j_2,\ell+m_2}e_{j_3,\ell}e_{j_4,\ell},e_{j_1,m_3}e_{j_2,m_4}e_{j_3,0}e_{j_4,0})\big| \\
&\leq \left(\var(e_{j_1,\ell+m_1}e_{j_2,\ell+m_2}e_{j_3,\ell}e_{j_4,\ell})\var(e_{j_1,m_3}e_{j_2,m_4}e_{j_3,0}e_{j_4,0}) \right)^{1/2} \\
&\leq \max_{1 \leq j \leq n_W} E(e_{j,t}^8) \\
&\leq E\|e_t\|^8.
\end{align*}
Next, we show that at most 9 of the terms in the sum \eqref{eqn:sumcov} are nonzero. Consider the term corresponding to a given index $\ell$ in the sum. For the covariance in the term to be nonzero, it must be the case that $\lbrace \ell+m_1,\ell+m_2,\ell \rbrace \cap \lbrace m_3,m_4,0 \rbrace \neq \emptyset$ (otherwise the two variables in the covariance would be independent). At most 9 values of $\ell$ have this property.

Putting the preceding two results together, we obtain the statement of the lemma. \qed

\subsection{Proof of \texorpdfstring{\cref{thm:var_sumnegl}}{Lemma \ref{thm:var_sumnegl}}}
We first remark that \cref{asn:var_dgp} implies $\lbrace W_t \rbrace$ is a strictly non-deterministic time series with Wold innovation $e_t$. Thus, the Wold representation $W_t = B(L)e_t$ has $B(L) = \sum_{\ell=0}^\infty B_\ell L^\ell = A(L)^{-1}$, and so for fixed $i,j$,  the elements $B_{i,j,\ell}$ of $B_\ell$ are absolutely summable across $\ell$ \citep[p. 418]{Brockwell1991}.

Define the $n_W^2p_T \times n_W^2 p_T$ matrix
\[R_{1,T} \equiv \frac{1}{T-p_T} \sum_{t=p_T+1}^T \ve(e_t X_t(p_T)')\ve(e_t X_t(p_T)')' - E\left[\ve(e_t X_t(p_T)')\ve(e_t X_t(p_T)')'\right]\]
with elements $R_{1,T,i,j}$. Then $\|R_{1,T}\|^2 = \sum_{i,j=1}^{n_W^2p_T} R_{1,T,i,j}^2$, and the first statement of the lemma follows if we can show that $E(R_{1,T,i,j}^2) = O(T^{-1})$ uniformly in $i,j$. Since $E(R_{1,T,i,j})=0$ for all $i,j$, we need to show that $\var(R_{1,T,i,j})=O(T^{-1})$ uniformly in $i,j$. The typical element $R_{1,T,i,j}$ has the form
\[\frac{1}{T-p_T} \sum_{t=p_T+1}^T e_{j_1,t} W_{j_2,t-m_1} e_{j_3,t} W_{j_4,t-m_2} - E\left[e_{j_1,t} W_{j_2,t-m_1} e_{j_3,t} W_{j_4,t-m_2}\right]\]
for appropriate $j_1,j_2,j_3,j_4,m_1,m_2 \in \mathbb{N}$. Here $W_{j,t}$ is the $j$-th element of $W_t$, and similarly for $e_t$. The variance of the above expression is given by
\begin{equation} \label{eqn:var_sumcov}
\frac{1}{(T-p_T)^2}\sum_{t=p_T+1}^T \sum_{s=p_T+1}^T \cov(e_{j_1,t} W_{j_2,t-m_1} e_{j_3,t} W_{j_4,t-m_2},e_{j_1,s} W_{j_2,s-m_1} e_{j_3,s} W_{j_4,s-m_2}).
\end{equation}
Using the above-mentioned Wold decomposition of $\lbrace W_t \rbrace$, we can write
\[ \textstyle W_{j_2,t-m_1} = \sum_{b_1=1}^{n_W} \sum_{\ell_1=0}^\infty  B_{j_2,b_1,\ell_1} e_{b_1,t-m_1-\ell_1},\]
say. Hence, the expression \eqref{eqn:var_sumcov} equals
\begin{align*}
&\frac{1}{T-p_T}\sum_{\ell_1,\ell_2,\ell_3,\ell_4=0}^\infty \sum_{b_1,b_2,b_3,b_4=1}^{n_W}  B_{j_2,b_1,\ell_1}B_{j_4,b_2,\ell_2}B_{j_2,b_3,\ell_3}B_{j_4,b_3,\ell_4} \\
&\quad \times \frac{1}{T-p_T} \sum_{s,t=p_T+1}^T \cov\left( e_{j_1,t} e_{b_1,t-m_1-\ell_1} e_{j_3,t} e_{b_2,t-m_2-\ell_2},e_{j_1,s} e_{b_3,s-m_1-\ell_3} e_{j_3,s} e_{b_4,s-m_2-\ell_4}\right).
\end{align*}
According to \cref{thm:var_sumcov}, the above display is bounded by
\begin{equation} \label{eqn:var_sumcovbound}
\frac{1}{T-p_T}\sum_{\ell_1,\ell_2,\ell_3,\ell_4=0}^\infty \sum_{b_1,b_2,b_3,b_4=1}^{n_W}  |B_{j_2,b_1,\ell_1}B_{j_4,b_2,\ell_2}B_{j_2,b_3,\ell_3}B_{j_4,b_4,\ell_4}| \times 9 E\|e_t\|^8 = O(T^{-1}),
\end{equation}
where the equality uses the previously-mentioned absolute summability of $\lbrace B_\ell \rbrace$. This concludes the proof of the first statement of the lemma.

We prove the second statement of the lemma in a similar fashion. Define the $n_W^2 p_T \times n_W^2$ matrix
\[R_{2,T} \equiv \frac{1}{T-p_T} \sum_{t=p_T+1}^T \ve(e_t X_t(p_T)')\ve(e_te_t'-\Sigma)'.\]
Decompose it as
\begin{align*}
R_{2,T} &= \frac{1}{T-p_T} \sum_{t=p_T+1}^T \ve(e_t X_t(p_T)')\ve(e_te_t')' - \frac{1}{T-p_T} \sum_{t=p_T+1}^T \ve(e_t X_t(p_T)')\ve(\Sigma)' \\
&\equiv \tilde{R}_{1,T} - \tilde{R}_{2,T}.
\end{align*}
Since $\lbrace \ve(e_t X_t(p_T)') \rbrace$ is a serially uncorrelated $(n_Wp_T)$-dimensional sequence, it is easy to show that $E\|\tilde{R}_{2,T}\|^2 = O_p(p_T/T)$. Consider now the matrix $\tilde{R}_{1,T}$. Its typical element
\[(T-p_T)^{-1} \textstyle \sum_{t=p_T+1}^T e_{j_1,t} W_{j_2,t-m}e_{j_3,t}e_{j_4,t}\]
has mean zero due to the independence of $e_t$ and $W_{t-m}$ for $m\geq 1$. We need to show that it has variance of order $O(T^{-1})$. Said variance equals
\begin{align*}
& \frac{1}{(T-p_T)^2}\sum_{s,t=p_T+1}^T \cov\left(e_{j_1,t} W_{j_2,t-m}e_{j_3,t}e_{j_4,t}, e_{j_1,s} W_{j_2,s-m}e_{j_3,s}e_{j_4,s} \right) \\
&=\frac{1}{T-p_T}\sum_{\ell_1,\ell_2=0}^\infty \sum_{b_1,b_2=1}^{n_W}  B_{j_2,b_1,\ell_1}B_{j_2,b_2,\ell_2} \\
&\qquad \times \frac{1}{T-p_T} \sum_{s,t=p_T+1}^T \cov\left( e_{j_1,t} e_{b_1,t-m-\ell_1} e_{j_3,t} e_{j_4,t},e_{j_1,s} e_{b_2,s-m-\ell_2} e_{j_3,s} e_{j_4,s}\right).
\end{align*}
This expression is of order $O(T^{-1})$, for the same reason as \eqref{eqn:var_sumcovbound} above. \qed

\subsection{Proof of \texorpdfstring{\cref{thm:var_asynorm}}{Lemma \ref{thm:var_asynorm}}}
This result is very similar to Thm. 2 in \citet{Lewis1985}, with the twist that we here deal also with the convergence of $\hat{\Sigma}$. Define $v_{\zeta,T} \equiv \tilde{\nu}_T'(\Gamma(p_T)^{-1} \otimes \Sigma)\tilde{\nu}_T$ for all $T$. If $v_\zeta \equiv \lim_{T\to\infty} v_{\zeta,T}=0$, it is easy to show that $(T-p_T)^{1/2}\tilde{\nu}_T'\ve(\hat{\beta}(p_T)-\beta(p_T)) = o_p(1)$ using \cref{thm:var_approx} and an mean-square bound, so in the following we assume $v_\zeta > 0$. By \cref{thm:var_approx} and the Cram\'{e}r-Wold device, we need to show that, for any $\lambda \in \mathbb{R}^{n_W^2}$,
\[\sum_{t=p_T+1}^{T} J_{t,T} \stackrel{d}{\to} N(0, 1) ,\]
where we define the triangular array
\[J_{t,T} \equiv \frac{\tilde{\nu}_T'\left(\Gamma(p_T)^{-1}X_t(p_T) \otimes e_t\right) + \lambda'\ve(e_t e_t'-\Sigma)}{(T-p_T)^{1/2}\big(v_{\zeta,T} + \lambda' \var(e_t \otimes e_t)\lambda\big)^{1/2}},\quad t=p_T+1,\dots,T,\; T \in \mathbb{N}.\]
Since $\lbrace e_t \rbrace$ is i.i.d., $e_t$ is independent of $X_t(p_T)$, so $\lbrace J_{t,T} \rbrace_{p_T+1 \leq t \leq T}$ is a martingale difference sequence with respect to the filtration generated by $\lbrace e_t \rbrace$. Also, since $E[X_t(p_T)]=0$, we have $E(J_{t,T}^2) = (T-p_T)^{-1}$. The statement of the lemma then follows from \citet[Thm. 24.3]{Davidson1994} if we can show
\begin{equation} \label{eqn:var_mdsvar}
\textstyle \sum_{t=p_T+1}^T J_{t,T}^2 \stackrel{p}{\to} 1
\end{equation}
and
\begin{equation} \label{eqn:var_negl}
\max_{p_T+1 \leq t \leq T}|J_{t,T}| \stackrel{p}{\to} 0.
\end{equation}

We first prove \eqref{eqn:var_mdsvar}, following the univariate argument in \citet[pp. 633--636]{Goncalves2007}. Decompose
\begin{align*}
\sum_{t=p_T+1}^T J_{t,T}^2 - 1 &= \lbrace v_{\zeta,T} + \lambda' \var(e_t \otimes e_t)\lambda \rbrace^{-1} \\
&\qquad \times \bigg\lbrace \frac{1}{T-p_T}\sum_{t=p_T+1}^T \left[\left(\tilde{\nu}_T'(\Gamma(p_T)^{-1}X_t(p_T) \otimes e_t)\right)^2 - v_{\zeta,T}\right] \\
&\qquad\qquad + \frac{2}{T-p_T}\sum_{t=p_T+1}^T\tilde{\nu}_T'(\Gamma(p_T)^{-1}X_t(p_T) \otimes e_t)\ve(e_t e_t'-\Sigma)'\lambda  \\
&\qquad\qquad+ \frac{1}{T-p_T}\sum_{t=p_T+1}^T \left[ \left(\lambda'\ve(e_t e_t'-\Sigma)\right)^2 - \lambda' \var(e_t \otimes e_t)\lambda\right] \bigg\rbrace \\
&\equiv \lbrace v_{\zeta,T} + \lambda' \var(e_t \otimes e_t)\lambda \rbrace^{-1} \big\lbrace R_{1,T} + 2 R_{2,T} + R_{3,T} \big\rbrace.
\end{align*}
The i.i.d. law of large numbers implies that $R_{3,T}=o_p(1)$. We now show that also $R_{1,T}=o_p(1)$ and $R_{2,T}=o_p(1)$. First,
\begin{align*}
&|R_{1,T}| \\
&= \bigg| \tilde{\nu}_T'(\Gamma(p_T)^{-1} \otimes I_{n_W}) \bigg\lbrace \frac{1}{T-p_T}\sum_{t=p_T+1}^T \ve(e_t X_t(p_T)')\ve(e_t X_t(p_T)')' \\
&\qquad\qquad\qquad\qquad\qquad\qquad - E\left[ \ve(e_t X_t(p_T)')\ve(e_t X_t(p_T)')' \right] \bigg\rbrace (\Gamma(p_T)^{-1} \otimes I_{n_W}) \tilde{\nu}_T \bigg| \\
&\leq \|\tilde{\nu}_T\|^2\, \|\Gamma(p_T)^{-1}\|_1^2 \\
&\qquad \times \bigg\| \frac{1}{T-p_T}\sum_{t=p_T+1}^T \ve(e_t X_t(p_T)')\ve(e_t X_t(p_T)')' - E\left[ \ve(e_t X_t(p_T)')\ve(e_t X_t(p_T)')' \right] \bigg\| \\
&= o_p(1),
\end{align*}
where the last line follows from \cref{thm:var_sumnegl} and \cref{asn:var_dgp,asn:var_lagorder}. Second, we analogously have
\begin{align*}
|R_{2,T}| &\leq \|\tilde{\nu}_T\|\,\|\lambda\|\, \|\Gamma(p_T)^{-1}\|_1 \, \bigg\|\frac{1}{T-p_T}\sum_{t=p_T+1}^T \ve(e_t X_t(p_T)')\ve(e_te_t-\Sigma)' \bigg\| \\
&= o_p(1),
\end{align*}
again using \cref{thm:var_sumnegl} and \cref{asn:var_dgp,asn:var_lagorder}. This concludes the proof of \eqref{eqn:var_mdsvar}.

To prove \eqref{eqn:var_negl}, first note that since $E\|e_t\|^{4+\epsilon}<\infty$ for some $\epsilon>0$, a standard argument for i.i.d. variables gives that $(T-p_T)^{-1/2}\max_{p_T+1 \leq t \leq T} |\lambda'\ve(e_te_t'-\Sigma)| = o_p(1)$. Next, the same calculations as in equation (2.12) in \citet[p. 401]{Lewis1985} yield
\begin{align*}
& P\left(\max_{p_T+1\leq t \leq T} \frac{(\tilde{\nu}_T'\left(\Gamma(p_T)^{-1}X_t(p_T) \otimes e_t\right))^2}{T-p_T} \geq \tilde{\epsilon} \right) \\
&\leq \frac{1}{\tilde{\epsilon}^2}\frac{p_T^2}{(T-p_T)}\|\tilde{\nu}_T\|^4 \|\Gamma(p_T)^{-1}\|_1^4 E\|e_t\|^4 E\|W_t\|^4 \\
&\to 0
\end{align*}
for any $\tilde{\epsilon}>0$. Putting the previous two facts together, we obtain \eqref{eqn:var_negl}. \qed

\subsection{Proof of \texorpdfstring{\cref{thm:var_Acossinconv}}{Lemma \ref{thm:var_Acossinconv}}}
For any $\omega \in [0,2\pi]$,
\begin{align*}
\|\hat{A}_{\cos}(\omega;p_T)-A_{\cos}(\omega;p_T)\|^2 &= \textstyle \sum_{\ell=1}^{p_T} \|\hat{A}_\ell-A_\ell\|^2 \cos^2(\omega\ell) \\
&\leq \textstyle \sum_{\ell=1}^{p_T} \|\hat{A}_\ell-A_\ell\|^2 \\
&= \|\hat{\beta}(p_T)-\beta(p_T)\|^2 \\
&= O(p_T/T),
\end{align*}
using \cref{thm:var_betaconv}. The argument for $A_{\sin}$ is identical.  \qed

\subsection{Proof of \texorpdfstring{\cref{thm:var_compact}}{Lemma \ref{thm:var_compact}}}
We start off by showing that the estimated VAR spectrum is nonsingular, asymptotically. Extend the definition of the Frobenius norm to complex matrices, so $\|B\|^2 \equiv \tr(B^*B)$. The matrix perturbation bound $|\det(B)-\det(C)| \leq n \|C-B\|\max\lbrace \|B\|,\|C\|\rbrace^{n-1}$ for $n \times n$ complex matrices $B$ and $C$ \citep[Problem I.6.11, p. 22]{Bhatia1997} implies
\begin{align}
&\left|\det(A(e^{i\omega})) - \det(I_{n_W} - \hat{A}_{\cos}(\omega;p_T)- i \hat{A}_{\sin}(\omega;p_T))\right| \nonumber \\
&\leq n_W \left\|\sum_{\ell=p_T+1}^\infty A_\ell e^{i\omega} - \sum_{\ell=1}^{p_T} (\hat{A}_\ell-A_\ell) e^{i\omega}  \right\| \max\left\lbrace \left\|\sum_{\ell=1}^\infty A_\ell e^{i\omega} \right\|,\left\|\sum_{\ell=1}^{p_T} \hat{A}_\ell e^{i\omega} \right\| \right\rbrace^{n_W-1}. \label{eqn:var_detbound}
\end{align}
\cref{thm:var_Acossinconv} implies
\[\sup_{\omega \in [0,2\pi]} \left\|\sum_{\ell=1}^{p_T} (\hat{A}_\ell-A_\ell) e^{i\omega}\right\| = o_p(1).\]
By \cref{asn:var_dgp,asn:var_lagorder}, the right-hand side of \eqref{eqn:var_detbound} therefore tends to 0 in probability uniformly in $\omega$, implying
\[\inf_{\omega \in [0,2\pi]} \left|\det(I_{n_W} - \hat{A}_{\cos}(\omega;p_T)- i \hat{A}_{\sin}(\omega;p_T))\right| = \inf_{\omega \in [0,2\pi]} |\det(A(e^{i\omega}))| + o_p(1) > \delta + o_p(1).\]
Thus, with probability approaching 1,
\[(\hat{A}_{\cos}(\omega;p_T),\hat{A}_{\sin}(\omega;p_T)) \in \mathcal{A}_\delta\quad \text{for all } \omega \in [0,2\pi].\]
We now show that, asymptotically, the estimated VAR spectrum lies in a region where $g(\cdot)$ is smooth. Let $M \equiv \max\lbrace 2\sum_{\ell=1}^\infty \|A_\ell\|,\|\Sigma\|\rbrace + 1$. By \cref{asn:var_param}, $g(\cdot,\cdot,\cdot)$ is continuously differentiable on $\mathcal{A}_M^0 \times \mathcal{S}_M^0$. Since
\[\left\|\hat{A}_{\cos}(\omega;p_T)-\sum_{\ell=1}^\infty A_\ell\right\| \leq \left\|\hat{A}_{\cos}(\omega;p_T)-A_{\cos}(\omega;p_T)\right\| + 2\sum_{\ell=1}^{\infty}\|A_\ell\| = 2\sum_{\ell=1}^{\infty}\|A_\ell\| + o_p(1)\]
uniformly in $\omega$ by \cref{thm:var_Acossinconv} and \cref{asn:var_lagorder} (and similarly for sin instead of cos), it follows that, with probability approaching 1,
\[(\hat{A}_{\cos}(\omega;p_T),\hat{A}_{\sin}(\omega;p_T)) \in \mathcal{A}_M^0\quad \text{for all } \omega \in [0,2\pi]. \]
Moreover, by the law of large numbers for i.i.d. variables and \cref{thm:var_Sigmaconv}, we also have $\hat{\Sigma}(p_T) \in \mathcal{S}_M^0$ with probability approaching 1. \qed

\subsection{Proof of \texorpdfstring{\cref{thm:var_deltamethod}}{Lemma \ref{thm:var_deltamethod}}}
We start out by applying a first-order Taylor expansion to the parameter of interest $\psi$. By \cref{thm:var_compact} and \cref{asn:var_param}, we can write
\begin{align*}
 & g(\hat{A}_{\cos}(\omega;p_T),\hat{A}_{\sin}(\omega;p_T),\hat{\Sigma}) - g(A_{\cos}(\omega;p_T),A_{\sin}(\omega;p_T),\Sigma) \\
&= g_1(A_{\cos}(\omega),A_{\sin}(\omega),\Sigma)'\ve(\hat{A}_{\cos}(\omega;p_T)-A_{\cos}(\omega)) \\
&\quad + g_2(A_{\cos}(\omega),A_{\sin}(\omega),\Sigma)'\ve(\hat{A}_{\sin}(\omega;p_T)-A_{\sin}(\omega))  \\
&\quad + g_3(A_{\cos}(\omega),A_{\sin}(\omega),\Sigma)'\ve(\hat{\Sigma}-\Sigma) \\
&\quad + \hat{R}_T(\omega),
\end{align*}
where the fact that $g(\cdot,\cdot,\cdot)$ is twice continuously differentiable implies that there exists a $C>0$ such that the remainder satisfies
\[|\hat{R}_T(\omega)| \leq C\left(\|\hat{A}_{\cos}(\omega;p_T)-A_{\cos}(\omega)\|^2 + \|\hat{A}_{\sin}(\omega;p_T)-A_{\sin}(\omega)\|^2 + \|\hat{\Sigma}-\Sigma\|^2\right)\]
for all $\omega$, with probability approaching 1. Since
\[\|\hat{A}_{\cos}(\omega;p_T)-A_{\cos}(\omega)\| \leq \textstyle \sum_{\ell=p_T+1}^\infty\|A_\ell\| + \|\hat{A}_{\cos}(\omega;p_T)-A_{\cos}(\omega;p_T)\| = O_p((p_T/T)^{1/2})\]
by \cref{thm:var_Acossinconv} and \cref{asn:var_lagorder} (and similarly with sin instead of cos), and since $\|\hat{\Sigma}-\Sigma\| = O_p(T^{-1/2})$ by \cref{thm:var_Sigmaconv}, we obtain
\[\int_0^{2\pi}|\hat{R}_T(\omega)|\, d\omega = O_p(p_T/T).\]
Using the continuity and thus boundedness of $h(\cdot)$, we therefore get
\begin{align}
\hat{\psi}(p_T)-\psi &= \int_0^{2\pi} h(\omega) g_1(A_{\cos}(\omega),A_{\sin}(\omega),\Sigma)'\ve(\hat{A}_{\cos}(\omega;p_T)-A_{\cos}(\omega))\, d\omega \nonumber \\ 
&\quad + \int_0^{2\pi} h(\omega) g_2(A_{\cos}(\omega),A_{\sin}(\omega),\Sigma)'\ve(\hat{A}_{\sin}(\omega;p_T)-A_{\sin}(\omega)) \, d\omega \nonumber \\
&\quad + \int_0^{2\pi}h(\omega)  g_3(A_{\cos}(\omega),A_{\sin}(\omega),\Sigma)'\ve(\hat{\Sigma}-\Sigma)  \, d\omega \nonumber \\
&\quad + O_p(p_T/T) \nonumber \\
&= \int_0^{2\pi}h(\omega) g_1(A_{\cos}(\omega),A_{\sin}(\omega),\Sigma)'\ve(\hat{A}_{\cos}(\omega;p_T)-A_{\cos}(\omega;p_T))\, d\omega \nonumber \\
&\quad + \int_0^{2\pi} h(\omega) g_2(A_{\cos}(\omega),A_{\sin}(\omega),\Sigma)'\ve(\hat{A}_{\sin}(\omega;p_T)-A_{\sin}(\omega;p_T)) \, d\omega \nonumber \\
&\quad + \int_0^{2\pi}h(\omega) g_1(A_{\cos}(\omega),A_{\sin}(\omega),\Sigma)'\ve(A_{\cos}(\omega;p_T)-A_{\cos}(\omega))\, d\omega \label{eqn:var_biascos} \\
&\quad + \int_0^{2\pi} h(\omega) g_2(A_{\cos}(\omega),A_{\sin}(\omega),\Sigma)'\ve(A_{\sin}(\omega;p_T)-A_{\sin}(\omega)) \, d\omega \label{eqn:var_biassin} \\
&\quad + \xi'\ve(\hat{\Sigma}-\Sigma) \nonumber \\
&\quad + O_p(p_T/T). \nonumber
\end{align}
We now bound the nonparametric bias term \eqref{eqn:var_biascos}; the argument for \eqref{eqn:var_biassin} is similar. 
Note that $h(\cdot)$ is bounded, and
\begin{align*}
&\int_0^{2\pi} \|g_1(A_{\cos}(\omega),A_{\sin}(\omega),\Sigma)'\ve(A_{\cos}(\omega;p_T)-A_{\cos}(\omega))\|\,d\omega \\
&\leq \int_0^{2\pi} \|g_1(A_{\cos}(\omega),A_{\sin}(\omega),\Sigma)\|\,d\omega \times \sup_{\omega \in [0,2\pi]} \|A_{\cos}(\omega;p_T)-A_{\cos}(\omega)\| \\
&\leq \int_0^{2\pi} \|g_1(A_{\cos}(\omega),A_{\sin}(\omega),\Sigma)\|\,d\omega \times \sum_{\ell=p_T+1}^\infty \|A_\ell\| \\
&= o(T^{-1/2}),
\end{align*}
by \cref{asn:var_lagorder}. We also used that \cref{asn:var_param} implies $\omega \mapsto \|g_1(A_{\cos}(\omega),A_{\sin}(\omega),\Sigma)\|$ is in $L_2(0,2\pi)$, implying that this function is integrable. Thus, the terms \eqref{eqn:var_biascos}--\eqref{eqn:var_biassin} are each $o(T^{-1/2})$.

To complete the proof, observe that
\begin{align*}
&\int_0^{2\pi}h(\omega) g_1(A_{\cos}(\omega),A_{\sin}(\omega),\Sigma)'\ve(\hat{A}_{\cos}(\omega;p_T)-A_{\cos}(\omega;p_T))\, d\omega \\
&\quad + \int_0^{2\pi} h(\omega) g_2(A_{\cos}(\omega),A_{\sin}(\omega),\Sigma)'\ve(\hat{A}_{\sin}(\omega;p_T)-A_{\sin}(\omega;p_T)) \, d\omega \\
&= \int_0^{2\pi}h(\omega) g_1(A_{\cos}(\omega),A_{\sin}(\omega),\Sigma)'\sum_{\ell=1}^{p_T} \ve(\hat{A}_\ell-A_\ell)\cos(\omega \ell)\, d\omega \\
&\quad + \int_0^{2\pi} h(\omega) g_2(A_{\cos}(\omega),A_{\sin}(\omega),\Sigma)'\sum_{\ell=1}^{p_T} \ve(\hat{A}_\ell-A_\ell)\sin(\omega \ell) \, d\omega \\
&= \sum_{\ell=1}^{p_T} \nu_{\ell,T}'\ve(\hat{A}_\ell-A_\ell).
\end{align*}
In conclusion,
\begin{align*}
\hat{\psi}(p_T)-\psi &= \nu_T'\ve(\hat{\beta}(p_T)-\beta(p_T)) + \xi'\ve(\hat{\Sigma}-\Sigma) + o(T^{-1/2}) + O_p(p_T/T).
\end{align*}
The above remainder terms are both $o_p((T-p_T)^{-1/2})$ by \cref{asn:var_lagorder}. \qed

\subsection{Proof of \texorpdfstring{\cref{thm:var_prop_asynorm}}{Proposition \ref{thm:var_prop_asynorm}}}
The proposition follows immediately from \cref{thm:var_asynorm,thm:var_deltamethod} if we can show that $\|\nu_T\|^2$ is bounded asymptotically. Let $g_{j,i}(\cdot,\cdot,\cdot)$ denote the $i$-th element of $g_j(\cdot,\cdot,\cdot)$, $j=1,2$, $i=1,2,\dots,n_W^2$. Let $M \equiv \sup_{\omega \in [0,2\pi]} |h(\omega)|<\infty$. Then
\begin{align*}
\|\nu_T\|^2 &= \sum_{i=1}^{n_W^2} \sum_{\ell=1}^{p_T} \Big(\int_0^{2\pi} h(\omega)\big\lbrace g_{1,i}(A_{\cos}(\omega),A_{\sin}(\omega),\Sigma)\cos(\omega \ell) \\
&\qquad\qquad\qquad\qquad\qquad + g_{2,i}(A_{\cos}(\omega),A_{\sin}(\omega),\Sigma)\sin(\omega \ell)\big\rbrace\,d\omega \Big)^2 \\
&\leq 2 M^2 \sum_{i=1}^{n_W^2} \sum_{\ell=1}^{p_T} \bigg\lbrace \left(\int_0^{2\pi} g_{1,i}(A_{\cos}(\omega),A_{\sin}(\omega),\Sigma)\cos(\omega \ell)\,d\omega\right)^2 \\
&\qquad\qquad\qquad\quad + \left(\int_0^{2\pi}  g_{2,i}(A_{\cos}(\omega),A_{\sin}(\omega),\Sigma)\sin(\omega \ell)\,d\omega\right)^2\bigg\rbrace.
\end{align*}
The sum
\begin{equation} \label{eqn:var_L2normproj}
\sum_{\ell=1}^{p_T} \left(\frac{1}{2\pi}\int_0^{2\pi} g_{1,i}(A_{\cos}(\omega),A_{\sin}(\omega),\Sigma)\cos(\omega \ell)\,d\omega\right)^2
\end{equation}
equals the $L_2(0,2\pi)$ norm of the projection of the function $\omega \mapsto g_{1,i}(A_{\cos}(\omega),A_{\sin}(\omega),\Sigma)$ onto the space of orthonormal functions $\lbrace \omega \mapsto \cos(\omega \ell)\rbrace_{1\leq \ell \leq p_T}$. Bessel's inequality therefore states that \eqref{eqn:var_L2normproj} is bounded above by the squared $L_2(0,2\pi)$ norm of the function $\omega \mapsto g_{1,i}(A_{\cos}(\omega),A_{\sin}(\omega),\Sigma)$. We can similarly bound the expression \eqref{eqn:var_L2normproj} with $g_{2,i}(\cdot,\cdot,\cdot)$ in place of $g_{1,i}(\cdot,\cdot,\cdot)$ and with $\sin(\omega\ell)$ in place of $\cos(\omega\ell)$. Hence,
\[\|\nu_T\|^2 \leq 8\pi^2 M^2 \textstyle \sum_{i=1}^{n_W^2} \left(\|g_{1,i}(A_{\cos}(\cdot),A_{\sin}(\cdot),\Sigma)\|_{L_2(0,2\pi)}^2 + \|g_{2,i}(A_{\cos}(\cdot),A_{\sin}(\cdot),\Sigma)\|_{L_2(0,2\pi)}^2 \right),\]
using obvious notation for the $L_2$ norms. These norms are finite by \cref{asn:var_param}. \qed

\subsection{Proof of \texorpdfstring{\cref{thm:var_prop_consvar}}{Proposition \ref{thm:var_prop_consvar}}}
We start by showing that $\|\hat{\nu}_T-\nu_T\| =o_p(1)$ and $\|\hat{\xi}(p_T)-\xi\| = o_p(1)$. By \cref{thm:var_compact}, and the twice continuous differentiability assumed in \cref{asn:var_param}, there exists a constant $C<\infty$ such that, with probability approaching one,
\begin{align*}
&\|g_j(\hat{A}_{\cos}(\omega;p_T),\hat{A}_{\sin}(\omega;p_T),\hat{\Sigma}) - g_j(A_{\cos}(\omega),A_{\sin}(\omega),\Sigma)\| \\
&\leq C\left( \|\hat{A}_{\cos}(\omega;p_T)-A_{\cos}(\omega)\| + \|\hat{A}_{\sin}(\omega;p_T)-A_{\sin}(\omega)\| + \|\hat{\Sigma}(p_T)-\Sigma\|\right) 
\end{align*}
for $j=1,2,3$. By \cref{thm:var_Sigmaconv} and the i.i.d. central limit theorem, we have $\|\hat{\Sigma}(p_T)-\Sigma\|=O_p(T^{-1/2})$. Using additionally \cref{thm:var_Acossinconv}, we then have, for example, that
\begin{align*}
&\sum_{\ell=1}^{p_T}\left\|\int_0^{2\pi} h(\omega) \left[g_1(\hat{A}_{\cos}(\omega),\hat{A}_{\sin}(\omega),\hat{\Sigma})- g_1(A_{\cos}(\omega),A_{\sin}(\omega),\Sigma)\right]\cos(\omega \ell)  \, d\omega\right\| \\
&\leq \tilde{C} p_T \sup_{\omega \in [0,2\pi]}\left(\|\hat{A}_{\cos}(\omega;p_T)-A_{\cos}(\omega)\| + \|\hat{A}_{\sin}(\omega;p_T)-A_{\sin}(\omega)\| + \|\hat{\Sigma}-\Sigma(p_T)\|\right)  \\
&= O_p((p_T^3/T)^{1/2}) \\
&= o_p(1),
\end{align*}
where $\tilde{C}$ is some constant. This type of calculation implies $\|\hat{\nu}_T-\nu_T\| =o_p(1)$ and $\|\hat{\xi}(p_T)-\xi\| = o_p(1)$.

We now deal with the consistency of the two terms in $\hat{\sigma}_\psi^2(p_T)$ one at a time. First, decompose
\begin{align*}
\hat{\nu}_T'(\hat{\Gamma}(p_T)^{-1} \otimes \hat{\Sigma}(p_T))\hat{\nu}_T &= \nu_T'\left((\hat{\Gamma}(p_T)^{-1} \otimes \hat{\Sigma}(p_T)) - (\Gamma(p_T)^{-1} \otimes \Sigma)\right)\nu_T \\
&\quad +  (\hat{\nu}_T-\nu_T)'(\hat{\Gamma}(p_T)^{-1} \otimes \hat{\Sigma}(p_T))(\hat{\nu}_T-\nu_T) \\
&\quad + 2(\hat{\nu}_T-\nu_T)'(\hat{\Gamma}(p_T)^{-1} \otimes \hat{\Sigma}(p_T))\nu_T \\
&\equiv R_{1,T} + R_{2,T} + 2R_{3,T}.
\end{align*}
Using \cref{thm:var_Gammaconv}, we find
\begin{align*}
|R_{1,T}| &\leq \|\nu_T\|^2 \left\|(\hat{\Gamma}(p_T)^{-1} \otimes \hat{\Sigma}(p_T)) - (\Gamma(p_T)^{-1} \otimes \Sigma) \right\|_1 \\
&\leq M \left(\|\hat{\Gamma}(p_T)^{-1}-\Gamma(p_T)^{-1}\|_1\|\, \hat{\Sigma}(p_T)\|_1 +\|\Gamma(p_T)^{-1}\|_1\, \|\hat{\Sigma}(p_T)-\Sigma\|_1 \right) \\
&\leq M \left(\|\hat{\Gamma}(p_T)-\Gamma(p_T)\|\, \|\Gamma(p_T)^{-1}\|_1\, \|\hat{\Gamma}(p_T)^{-1}\|_1\, \| \hat{\Sigma}(p_T)\| +\|\Gamma(p_T)^{-1}\|_1\, \|\hat{\Sigma}(p_T)-\Sigma\| \right) \\
&= o_p(1).
\end{align*}
Similar calculations, along with the fact $\|\hat{\nu}_T-\nu_T\|=o_p(1)$, can be used to show that $R_{2,T}=o_p(1)$ and $R_{3,T}=o_p(1)$.

Second, define $\Xi \equiv \var(e_t \otimes e_t)$ and decompose
\begin{align*}
\hat{\xi}(p_T)'\hat{\Xi}(p_T)\hat{\xi}(p_T) - \xi'\Xi \xi &= \xi'(\hat{\Xi}(p_T)-\Xi)\xi \\
&\quad + (\hat{\xi}(p_T)-\xi)'\hat{\Xi}(p_T)(\hat{\xi}(p_T)-\xi) \\
& \quad + 2(\hat{\xi}(p_T)-\xi)'\hat{\Xi}(p_T)\xi \\
&\equiv \tilde{R}_{1,T} + \tilde{R}_{2,T} + 2\tilde{R}_{3,T}.
\end{align*}
Since $\|\hat{\xi}(p_T)-\xi\|=o_p(1)$, the statement of the proposition follows if we can show $\|\hat{\Xi}(p_T)-\Xi\|=o_p(1)$. Define $\chi_t \equiv \ve(e_te_t' - \Sigma)$, and note that $(T-p_T)^{-1}\sum_{t=p_T+1}^T \chi_t\chi_t' \stackrel{p}{\to} \Xi$ by the usual law of large numbers for i.i.d. variables. Because
\begin{align*}
\|\hat{\Xi}(p_T)-\Xi\| &\leq \frac{1}{T-p_T} \sum_{t=p_T+1}^T\|\hat{\chi}_t\hat{\chi}_t'-\chi_t\chi_t'\| \\
&\leq \frac{1}{T-p_T} \sum_{t=p_T+1}^T\|\hat{\chi}_t-\chi_t\|^2 + \frac{2}{T-p_T} \sum_{t=p_T+1}^T\|\hat{\chi}_t-\chi_t\|\,\|\chi_t\| \\
&\leq \frac{1}{T-p_T} \sum_{t=p_T+1}^T\|\hat{\chi}_t-\chi_t\|^2 + 2\bigg(\frac{1}{T-p_T} \sum_{t=p_T+1}^T\|\hat{\chi}_t-\chi_t\|^2\\
&\qquad\qquad\qquad\qquad\qquad\qquad\qquad\quad \times \frac{1}{T-p_T} \sum_{t=p_T+1}^T\|\chi_t\|^2\bigg)^{1/2}
\end{align*}
by Cauchy-Schwarz, we just need to show that
\[(T-p_T)^{-1} \textstyle \sum_{t=p_T+1}^T\|\hat{\chi}_t-\chi_t\|^2 = o_p(1).\]
Since
\begin{align*}
\|\hat{\chi}_t-\chi_t\| &= \|\hat{e}_t(p_T)\hat{e}_t(p_T)' - e_te_t'\| \\
&\leq \|\hat{e}_t(p_T)-e_t\|^2 + 2\|\hat{e}_t(p_T)-e_t\|\,\|e_t\|,
\end{align*}
we have
\begin{align*}
\frac{1}{T-p_T}\sum_{t=p_T+1}^T\|\hat{\chi}_t-\chi_t\|^2 &\leq \frac{2}{T-p_T}\sum_{t=p_T+1}^T\|\hat{e}_t-e_t\|^4  \\
&\quad + 4\left(\frac{1}{T-p_T}\sum_{t=p_T+1}^T\|\hat{e}_t-e_t\|^4 \frac{1}{T-p_T}\sum_{t=p_T+1}^T\|e_t\|^4\right)^{1/2}.
\end{align*}
The i.i.d. law of large numbers gives $(T-p_T)^{-1}\sum_{t=p_T+1}^T\|e_t\|^4=O_p(1)$. To complete the proof, we bound
\begin{align*}
\frac{1}{T-p_T}\sum_{t=p_T+1}^T\|\hat{e}_t-e_t\|^4 &\leq \frac{8}{T-p_T}\sum_{t=p_T+1}^T\|\hat{e}_t-e_t(p_T)\|^4 + \frac{8}{T-p_T}\sum_{t=p_T+1}^T\|e_t-e_t(p_T)\|^4 \\
&\equiv 8(\hat{R}_{1,T} + \hat{R}_{2,T})
\end{align*}
and show that the two terms on the right-hand side tend to zero, using similar arguments as in the proof of \cref{thm:var_Sigmaconv}. First,
\[\hat{R}_{1,T} \leq \|\hat{\beta}(p_T)-\beta(p_T)\|^4 \frac{1}{T-p_T}\sum_{t=p_T+1}^T\|X_t(p_T)\|^4 = O_p((p_T/T)^2) O_p(p_T^2) = o_p(1),\]
since
\[\textstyle E\|X_t(p_T)\|^4 = E\left(\sum_{\ell=1}^{p_T} \|W_{t-\ell}\|^2 \right)^2 = \sum_{\ell=1}^{p_T}\sum_{m=1}^{p_T}E\left(\|W_{t-\ell}\|^2\|W_{t-m}\|^2 \right) = O(p_T^2).\]
Second,
\begin{align*}
E(\hat{R}_{2,T}) &= E\|e_t-e_t(p_T)\|^4 \\
&\leq \textstyle E\left(\sum_{\ell=p_T+1}^\infty \|A_\ell\|\|W_{t-\ell}\| \right)^4 \\
&= \textstyle \sum_{\ell_1,\ell_2,\ell_3,\ell_4=p_T+1}^\infty  \|A_{\ell_1}\|\, \|A_{\ell_2}\|\, \|A_{\ell_3}\|\, \|A_{\ell_4}\| \, E(\|W_{t-\ell_1}\|\, \|W_{t-\ell_2}\|\, \|W_{t-\ell_3}\|\, \|W_{t-\ell_4}\|) \\
&\leq \text{constant} \times \left( \textstyle \sum_{\ell=p_T+1}^\infty \|A_\ell\| \right)^4 \\
&= o(1). \qed
\end{align*}

\clearpage

\phantomsection
\addcontentsline{toc}{section}{References}
\bibliography{decomp_iv_ref}